\pdfoutput=1 

\newcommand*{\READY}{}

\ifdefined\READY
\documentclass[sigplan,screen]{acmart} 
\else
\documentclass[sigplan,10pt,anonymous,review]{acmart}
\fi

\usepackage[T1]{fontenc}  
\usepackage{url}
\usepackage{alltt}
\usepackage{holtexbasic}
\usepackage{holtex}
\usepackage{graphicx}
\usepackage{tikz}

\usetikzlibrary{matrix,shapes,positioning}
\usetikzlibrary{fit,calc} 
\usetikzlibrary{backgrounds}  
\usetikzlibrary{decorations.pathreplacing} 
\usepackage{amsmath,amssymb,amsfonts}
\usepackage[inline]{enumitem} 
\usepackage[hang,flushmargin]{footmisc} 
\usepackage[]{algorithm2e}  
\usepackage{lineno} 

\usepackage{multirow} 

\newcommand\urlpre{https://bitbucket.org/jhlchan/project/src/aed43871448e/fermat/twosq/}
\ifdefined\READY
\newcommand{\script}[3][script]{\href{\urlpre#2Script.sml\#lines-#3}{[#1]}} 
\else
\newcommand{\script}[3][script]{\textit{[script]}} 
\fi

\newcommand{\eg}{\textit{e.g.}}
\newcommand{\ie}{\textit{i.e.}}
\newcommand{\etc}{\textit{etc.}}

\newcommand{\ee}{\ \HOLConst{=}\ } 

\newsavebox{\supbox}
\newsavebox{\subbox}

\renewcommand{\HOLinline}[1]{\ensuremath{#1}} 
\renewcommand{\HOLConst}[1]{{\fontsize{9pt}{1em}\selectfont\textsf{\upshape #1}}}
\renewcommand{\HOLNumLit}[1]{{\selectfont\textrm{\upshape #1}}} 
\renewcommand{\HOLSymConst}[1]{{\fontsize{9pt}{1em}\selectfont\textsf{\upshape #1}}}
\renewcommand{\HOLKeyword}[1]{\mbox{\upshape\fontsize{9pt}{1em}\selectfont\bfseries{\textsf{#1}}}}
\renewcommand{\HOLTokenBar}{\ensuremath{\mid}}

\newcommand{\q}[1]{\ensuremath{\texttt{#1}}}

\newenvironment{mquote}[1][2em]
  {\list{}{\leftmargin=#1\rightmargin=#1}\item[]}%
  {\endlist}

\newcommand{\windmill}[5]
{
\draw[fill=blue!30!white] 
     (#1,#2+#3) -- (#1,#2+#3+#5) -- (#1+#4,#2+#3+#5) -- (#1+#4,#2+#3)  -- (#1,#2+#3);
\draw[fill=blue!30!white] 
     (#1+#3,#2+#3) -- (#1+#3,#2+#3-#4) -- (#1+#3+#5,#2+#3-#4) -- (#1+#3+#5,#2+#3) -- (#1+#3,#2+#3);
\draw[fill=blue!30!white] 
     (#1+#3,#2) -- (#1+#3,#2-#5) -- (#1+#3-#4,#2-#5) -- (#1+#3-#4,#2)  -- (#1+#3,#2);
\draw[fill=blue!30!white] 
     (#1,#2) -- (#1,#2+#4) -- (#1-#5,#2+#4) -- (#1-#5,#2) -- (#1,#2);
\draw[fill=pink!40!white] 
     (#1,#2) -- (#1,#2+#3) -- (#1+#3,#2+#3) -- (#1+#3,#2) -- (#1,#2);
}

\newcommand{\windmillr}[5] 
{
\draw[fill=blue!30!white] 
     (#1+#3,#2+#3) -- (#1+#3,#2+#3+#5) -- (#1+#3-#4,#2+#3+#5) -- (#1+#3-#4,#2+#3)
      -- (#1+#3,#2+#3);
\draw[fill=blue!30!white] 
     (#1+#3,#2) -- (#1+#3,#2+#4) -- (#1+#3+#5,#2+#4) -- (#1+#3+#5,#2) -- (#1+#3,#2);
\draw[fill=blue!30!white] 
     (#1,#2) -- (#1,#2-#5) -- (#1+#4,#2-#5) -- (#1+#4,#2)  -- (#1,#2);
\draw[fill=blue!30!white] 
     (#1,#2+#3) -- (#1,#2+#3-#4) -- (#1-#5,#2+#3-#4) -- (#1-#5,#2+#3) -- (#1,#2+#3);
\draw[fill=pink!40!white] 
     (#1,#2) -- (#1,#2+#3) -- (#1+#3,#2+#3) -- (#1+#3,#2) -- (#1,#2);
}

\newcommand{\mind}[4][dashed, ultra thick] 
{
\draw[#1] (#2,#3) -- (#2,#3+#4) -- (#2+#4,#3+#4) -- (#2+#4,#3) -- (#2,#3); 
}

\newcommand\DoubleLine[5][3pt]{
  \path(#2)--(#3)coordinate[at start](h1)coordinate[at end](h2);
  \draw[#4]($(h1)!#1!90:(h2)$)--($(h2)!#1!-90:(h1)$);
  \draw[#5]($(h1)!#1!-90:(h2)$)--($(h2)!#1!90:(h1)$);
}

\tikzset{%
glow/.style={%
preaction={#1, draw, line join=round, line width=0.5pt, opacity=0.05,
preaction={#1, draw, line join=round, line width=1.0pt, opacity=0.05,
preaction={#1, draw, line join=round, line width=1.5pt, opacity=0.05,
preaction={#1, draw, line join=round, line width=2.0pt, opacity=0.05,
preaction={#1, draw, line join=round, line width=2.5pt, opacity=0.05,
preaction={#1, draw, line join=round, line width=3.0pt, opacity=0.05,
preaction={#1, draw, line join=round, line width=3.5pt, opacity=0.05,
preaction={#1, draw, line join=round, line width=4.0pt, opacity=0.05,
preaction={#1, draw, line join=round, line width=4.5pt, opacity=0.05,
preaction={#1, draw, line join=round, line width=5.0pt, opacity=0.05,
preaction={#1, draw, line join=round, line width=5.5pt, opacity=0.05,
preaction={#1, draw, line join=round, line width=6.0pt, opacity=0.05,
}}}}}}}}}}}}}}

\makeatletter
\let\c@proposition\c@theorem
\let\c@corollary\c@theorem
\let\c@lemma\c@theorem
\let\c@definition\c@theorem
\let\c@example\c@theorem
\makeatother

\setcopyright{acmlicensed}
\acmPrice{15.00}
\acmDOI{10.1145/3497775.3503673}
\acmYear{2022}
\copyrightyear{2022}
\acmSubmissionID{poplws22cppmain-p4-p}
\acmISBN{978-1-4503-9182-5/22/01}
\acmConference[CPP '22]{Proceedings of the 11th ACM SIGPLAN International Conference on Certified Programs and Proofs}{January 17--18, 2022}{Philadelphia, PA, USA}
\acmBooktitle{Proceedings of the 11th ACM SIGPLAN International Conference on Certified Programs and Proofs (CPP '22), January 17--18, 2022, Philadelphia, PA, USA}
\begin{document}

\title[Windmills of the Minds: An Algorithm for Fermat's Two Squares Theorem]{Windmills of the Minds: An Algorithm for\\ Fermat's Two Squares Theorem}

\ifdefined\READY
\author{Hing Lun Chan}
\affiliation{%
   \institution{Australian National University}
   \city{Canberra}
   \country{Australia}
}
\orcid{0000-0003-1811-1684}
\email{joseph.chan@anu.edu.au}

\pdfinfo{
   /Author (Hing Lun Chan)
   /Title  (Windmills of the minds: an algorithm for Fermat's Two Squares Theorem)
   /Subject (Theorem Proving)
   /Keywords (number theory;algorithm;prime)
}
\else
\fi

\date{}

\begin{abstract}
The two squares theorem of Fermat is a gem in number theory,
with a spectacular one-sentence ``proof from the Book''.
Here is a formalisation of this proof, with an interpretation using windmill patterns.
The theory behind involves involutions on a finite set, 
especially the parity of the number of fixed points in the involutions.
Starting as an existence proof that is non-constructive, there is an ingenious way to turn it into a constructive one. This gives an algorithm to compute the two squares by iterating the two involutions alternatively from a known fixed point.
\end{abstract}



\maketitle

\pdfbookmark{\contentsname}{Contents}

\setcounter{footnote}{0}

\section{Introduction}
\label{sec:introduction}
Fermat's two squares theorem, dated back to 1640, states that a prime $n$ that is one more than a multiple of $4$ can be uniquely expressed as a sum of odd and even squares (Section~\ref{sec:sum-of-two-squares}, Theorem~\ref{thm:fermat-two-squares-thm}).
Of the many proofs of this classical number theory result, this one-sentence proof by Zagier~\cite{Zagier-1990-acm} caused a sensation in 1990:

\bigskip
\begin{mquote} 
\textit{
\\ 
The involution on the finite set
\[
S = \{(x,y,z) \in \mathbb{N}^{3} \mid \HOLinline{\HOLFreeVar{n}\;\HOLSymConst{=}\;\HOLFreeVar{x}\HOLSymConst{\ensuremath{\sp{2}}}\;\HOLSymConst{\ensuremath{+}}\;\HOLNumLit{4}\HOLSymConst{\ensuremath{}}\HOLFreeVar{y}\HOLSymConst{\ensuremath{}}\HOLFreeVar{z}} \}
\]
defined by 
\begin{equation}
\label{eqn:zagier-map}
\HOLinline{(\HOLFreeVar{x}\HOLSymConst{,}\HOLFreeVar{y}\HOLSymConst{,}\HOLFreeVar{z})} \longmapsto
\begin{cases}
      (x + 2 z,z,y - z - x) & \text{if }x < y - z \\
      (2 y - x,y,x + z - y) & \text{if }y - z < x < 2y\\
      (x - 2 y,x + z - y,y) & \text{if }x > 2y
\end{cases}
\end{equation}
has exactly one fixed point, so \HOLinline{\ensuremath{|}\HOLFreeVar{S}\ensuremath{|}} is odd, and the involution defined by
$\HOLinline{(\HOLFreeVar{x}\HOLSymConst{,}\HOLFreeVar{y}\HOLSymConst{,}\HOLFreeVar{z})} \longmapsto \HOLinline{(\HOLFreeVar{x}\HOLSymConst{,}\HOLFreeVar{z}\HOLSymConst{,}\HOLFreeVar{y})}$
also has a fixed point.
\\ 
}
\end{mquote}
Those who are perplexed by this multi-line sentence are not alone.
Even knowing involution, a self-inverse function, and fixed points, those values unchanged by a function, the proof is not obvious at a glance!

Listed as number 20 in Formalizing 100 Theorems~\cite{Wiedijk-2020},
there are many formal proofs of this theorem.
Some are based on textbook proofs, others follow the ideas in Zagier's proof.
All show the existence of the two squares, only a few (Coq~\cite{Thery-2004} and Lean~\cite{Hughes-2019}) include the uniqueness part.
Therefore, a formalisation of this one-sentence proof, in a constructive way, is an interesting exercise in theorem-proving.
As a bonus, the exercise is a path of discovery due to recent progress in understanding this proof.

As Don Zagier remarked after the one sentence, his proof was a condensed version of a 1984 proof by Roger Heath-Brown~\cite{Heath-Brown-1984-acm}, who in turn acknowledged prior work in number theory taken up by Joseph Liouville~\cite{Williams-2010-alt}.
This one-sentence proof invokes two involutions: the second one is obvious, but the first one in Equation~\eqref{eqn:zagier-map} has been called ``black magic''~\cite{Trimble-Lama-2008}. The algebraic formulation of this involution has been given a geometric interpretation by Alexander Spivak~\cite{Spivak-2007} in 2007. These are the windmills (Section~\ref{sec:windmills}). They explain why the magic works, and suggest an interplay of the involutions to identify fixed points of each other. Moreover, this provides an algorithm to find the two squares in Fermat's theorem.
Thus the one-sentence proof can be made constructive, as elucidated by Zagier~\cite{Zagier-2013-acm} in 2013.

\subsection{Contribution}
\label{sec:contribution}
This paper gives the first formal proof of an algorithm to compute the two squares in Fermat's two squares theorem, by following a constructive version of Zagier's proof in HOL4.

As noted before, Zagier's proof has been formalised, in HOL Light~\cite{Harrison-2010}, in NASA PVS~\cite{Narkawicz-2012-acm} and in Coq~\cite{Dubach-Muehlboeck-2021-acm}, although not in this constructive form.

All the ideas used in this paper can be found in Shiu~\cite{Shiu-1996} and Zagier~\cite{Zagier-2013-acm}.
The novel feature of this work is an elegant and pictorial approach for our formalisation.
The emphasis is in providing formal definitions and developing appropriate theories, not only for the present work, but also for supporting further work.

\subsection{Overview}
\label{sec:overview}
Major features in this formalisation are:
\begin{itemize}[leftmargin=*] 
\item the groundwork for Zagier's proof in Section~\ref{sec:sum-of-two-squares},
\item the two involutions for windmills in Section~\ref{sec:windmill-involutions},
\item the existence and uniqueness of two squares in Section~\ref{sec:two-squares-theorem},
\item an algorithm to compute the two squares in Section~\ref{sec:algorithm},
\item theories of involutions and iterations in Section~\ref{sec:orbits}, and
\item a correctness proof of our algorithm in Section~\ref{sec:correctness}.
\end{itemize}
After a review of the work done, we conclude in Section~\ref{sec:conclusion}.

\subsection{Notation}
\label{sec:notations}
Statements starting with a turnstile ($\vdash$) are HOL4 theorems,
automatically pretty-printed to \LaTeX{} from the relevant \text{theory} in the HOL4 development.
Generally, our notation allows an appealing combination of quantifiers ($\forall, \exists, \exists{!}$),
logical connectives (\HOLTokenConj{} for ``and'', \HOLTokenDisj{} for ``or'', \HOLTokenNeg{} for ``not'', also \HOLTokenImp{} for ``implies'' and \HOLTokenEquiv{} for ``if and only if''),
set theory ($\in$ for ``element of'', $\times$ for Cartesian product, and comprehensions such as \HOLinline{\HOLTokenLeftbrace{}\HOLBoundVar{x}\;\HOLTokenBar{}\;\HOLBoundVar{x}\;\HOLSymConst{\HOLTokenLt{}}\;\HOLNumLit{6}\HOLTokenRightbrace{}}),
and functional programming (\HOLTokenLambda{} for abstraction, and juxtaposition for application).
Repeated application of a function $f$ is indicated by exponents,
\eg, \HOLinline{\HOLFreeVar{f}\;(\HOLFreeVar{f}\;(\HOLFreeVar{f}\;\HOLFreeVar{x}))\;\HOLSymConst{=}\;\HOLFreeVar{f}\ensuremath{\sp{\HOLNumLit{3}}(\HOLFreeVar{x})}}.

For a function $f$ from set $S$ to set $T$,
we write \HOLinline{\HOLFreeVar{f}\;\ensuremath{:}\;\HOLFreeVar{S}\;\ensuremath{\leftrightarrow}\;\HOLFreeVar{T}} to mean a bijection.
The empty set is denoted by \HOLinline{\HOLSymConst{\HOLTokenEmpty{}}},
and a finite set, denoted by \HOLinline{\HOLConst{\HOLConst{finite}}\;\HOLFreeVar{S}},
has cardinality \HOLinline{\ensuremath{|}\HOLFreeVar{S}\ensuremath{|}}.

The set of natural numbers is denoted by $\mathbb{N}$, counting from $0$,
and \HOLinline{\HOLConst{count}\;\HOLFreeVar{n}}\;\HOLTokenDefEquality{}\;\HOLinline{\HOLTokenLeftbrace{}\HOLBoundVar{x}\;\HOLTokenBar{}\;\HOLBoundVar{x}\;\HOLSymConst{\HOLTokenLt{}}\;\HOLFreeVar{n}\HOLTokenRightbrace{}},
where \HOLTokenDefEquality{} means `equality by definition'.
For a natural number $n \in \mathbb{N}$,
\HOLinline{\HOLConst{square}\;\HOLFreeVar{n}} means it is a square: \HOLinline{\HOLSymConst{\HOLTokenExists{}}\HOLBoundVar{k}.\;\HOLFreeVar{n}\;\HOLSymConst{=}\;\HOLBoundVar{k}\HOLSymConst{\ensuremath{\sp{2}}}},
\HOLinline{\HOLConst{prime}\;\HOLFreeVar{n}} means it is a prime, and
\HOLinline{\HOLConst{\HOLConst{even}}\;\HOLFreeVar{n}} or \HOLinline{\HOLConst{\HOLConst{odd}}\;\HOLFreeVar{n}} denotes its parity.
The integer quotient and remainder of $m$ divided by $n$ are written as \HOLinline{\HOLFreeVar{m}\;\HOLConst{\HOLConst{div}}\;\HOLFreeVar{n}} and \HOLinline{\HOLFreeVar{m}\;\HOLConst{\HOLConst{mod}}\;\HOLFreeVar{n}}, respectively.
We write \HOLinline{\HOLFreeVar{n}\;\HOLConst{\ensuremath{\mid}}\;\HOLFreeVar{m}} when $n$ divides $m$,
which is equivalent to \HOLinline{\HOLFreeVar{m}\;\ensuremath{\equiv}\;\HOLNumLit{0}\;\ensuremath{(}\ensuremath{\bmod}\;\HOLFreeVar{n}\ensuremath{)}} when \HOLinline{\HOLFreeVar{n}\;\HOLSymConst{\HOLTokenNotEqual{}}\;\HOLNumLit{0}}.

These are basic notations. Others will be introduced as they first appear.


\paragraph*{HOL4 Sources}
\ifdefined\READY
Proof scripts are located in a repository at {\small\url{https://bitbucket.org/jhlchan/project/src/master/fermat/twosq/}}.
The scripts are compiled using HOL4, version \texttt{af01322db666}.
In this paper, each theorem has \emph{\script{windmill}{197}}, 
which is hyperlinked to the appropriate line of the corresponding proof script in repository.
\else
Proof scripts of all theorems are located in a repository (omitted for anonymous review).
The scripts are compiled using HOL4, version \texttt{6dcb52a09341}.
In this paper, each theorem has \emph{\script{windmill}{197}}, 
which is hyperlinked to the appropriate line of the corresponding proof script in repository.
(For anonymous review, this feature has been removed. See Appendix~\ref{app:cross-reference-theorems} for a cross-reference of theorems in this paper and the proof scripts in supplementary material.)
\fi

\section{Sum of Two Squares}
\label{sec:sum-of-two-squares}
The only even prime is \HOLinline{\HOLNumLit{2}\;\HOLSymConst{=}\;\HOLNumLit{1}\ensuremath{{\sp{\HOLNumLit{2}}}}\;\HOLSymConst{\ensuremath{+}}\;\HOLNumLit{1}\ensuremath{{\sp{\HOLNumLit{2}}}}}, a sum of two squares.
An odd prime, upon division by 4, leaves a remainder of either $1$ or $3$.
Only an odd prime of the first type can be expressed as a sum of two squares,
as supported by numerical evidence from Table~\ref{tbl:odd-primes-sum}.

\begin{table}[h] 
\caption{Examples of odd primes that can be expressed as a sum of two squares.}
\Description{This table provides examples of odd primes that can be expressed as a sum of two squares.}
\label{tbl:odd-primes-sum}  
\[
\begin{array}{r@{\quad\ee\quad}l@{\quad\ee\quad}l}
     5 & 4(1) + 1  & 1^{2} + 2^{2}\\
    13 & 4(3) + 1  & 3^{2} + 2^{2}\\
    17 & 4(4) + 1  & 1^{2} + 4^{2}\\
    29 & 4(7) + 1  & 5^{2} + 2^{2}\\
    37 & 4(9) + 1  & 1^{2} + 6^{2}\\
    41 & 4(10) + 1  & 5^{2} + 4^{2}\\
    53 & 4(13) + 1  & 7^{2} + 2^{2}\\
    61 & 4(15) + 1  & 5^{2} + 6^{2}\\
\end{array}
\]
\end{table}

\noindent
Pierre de Fermat, in a letter to Marin Mersenne on Christmas day 1640,
claimed that he had an ``irrefutable'' proof of this:
\begin{theorem}[\textbf{Two Squares Theorem}]
\label{thm:fermat-two-squares-thm}
\script{twoSquares}{619}
A prime $n$ can be expressed uniquely as a sum of odd and even squares
if and only if \HOLinline{\HOLFreeVar{n}\;\HOLSymConst{=}\;\HOLNumLit{4}\HOLSymConst{\ensuremath{}}\HOLFreeVar{k}\;\HOLSymConst{\ensuremath{+}}\;\HOLNumLit{1}} for some $k$.
\begin{HOLmath}
\HOLTokenTurnstile{}\HOLConst{prime}\;\HOLFreeVar{n}\;\HOLSymConst{\HOLTokenImp{}}\\
\;\;\;\;\;\;\;(\HOLFreeVar{n}\;\ensuremath{\equiv}\;\HOLNumLit{1}\;\ensuremath{(}\ensuremath{\bmod}\;\HOLNumLit{4}\ensuremath{)}\;\HOLSymConst{\HOLTokenEquiv{}}\\
\;\;\;\;\;\;\;\;\;\;\HOLSymConst{\HOLTokenUnique{}}(\HOLBoundVar{u}\HOLSymConst{,}\HOLBoundVar{v}).\;\HOLConst{\HOLConst{odd}}\;\HOLBoundVar{u}\;\HOLSymConst{\HOLTokenConj{}}\;\HOLConst{\HOLConst{even}}\;\HOLBoundVar{v}\;\HOLSymConst{\HOLTokenConj{}}\;\HOLFreeVar{n}\;\HOLSymConst{=}\;\HOLBoundVar{u}\HOLSymConst{\ensuremath{\sp{2}}}\;\HOLSymConst{\ensuremath{+}}\;\HOLBoundVar{v}\HOLSymConst{\ensuremath{\sp{2}}})
\end{HOLmath}
\end{theorem}

\noindent
This paper concentrates on formalising an elementary proof of this result by Roger Heath-Brown, later simplified by Don Zagier.
As shown in his one-sentence proof in Section~\ref{sec:introduction},
the idea is this: look at the representations of $n$ not by squares, but in another form.
Consider the following set $S_{n}$ of triples $(x,y,z)$:
\begin{equation}
\label{eqn:mills-set}
S_{n} = \{(x,y,z) \in \mathbb{N}\times{\mathbb{N}}\times{\mathbb{N}}
\mid \HOLinline{\HOLFreeVar{n}\;\HOLSymConst{=}\;\HOLFreeVar{x}\HOLSymConst{\ensuremath{\sp{2}}}\;\HOLSymConst{\ensuremath{+}}\;\HOLNumLit{4}\HOLSymConst{\ensuremath{}}\HOLFreeVar{y}\HOLSymConst{\ensuremath{}}\HOLFreeVar{z}} \}.
\end{equation}
For a prime $n$ of the form $4k + 1$, we have $(1,1,k) \in S_{n}$.
Thus the set $S_{n}$ is non-empty, and there are only finitely many triples in $S_{n}$.
A triple with \HOLinline{\HOLFreeVar{y}\;\HOLSymConst{=}\;\HOLFreeVar{z}} will give \HOLinline{\HOLFreeVar{n}\;\HOLSymConst{=}\;\HOLFreeVar{x}\HOLSymConst{\ensuremath{\sp{2}}}\;\HOLSymConst{\ensuremath{+}}\;\HOLNumLit{4}\HOLSymConst{\ensuremath{}}\HOLFreeVar{y}\HOLSymConst{\ensuremath{\sp{2}}}}, \ie, a sum of two squares.
If we can show that $S_{n}$ has only one such triple,
we have a proof of Fermat's Theorem~\ref{thm:fermat-two-squares-thm}, with both existence and uniqueness.

Meanwhile, some general theories will be developed as an exercise in formal proofs, so that they can be applied to similar problems.
In addition, we extend the theories to establish not only an algorithm, but also a proof of its correctness, to compute the two unique squares for primes of the form \HOLinline{\HOLNumLit{4}\HOLSymConst{\ensuremath{}}\HOLFreeVar{k}\;\HOLSymConst{\ensuremath{+}}\;\HOLNumLit{1}}.

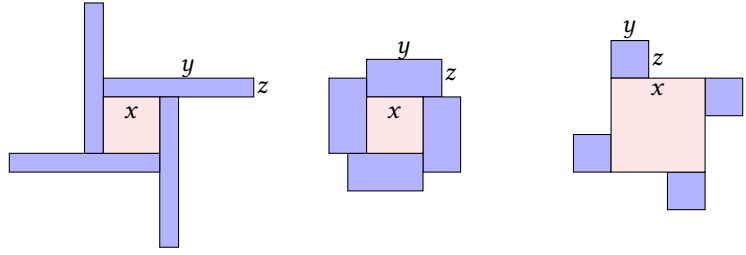
\begin{figure*}[h]  
\begin{center}
\begin{tikzpicture}[scale=0.25] 
\windmill{0}{0}{3}{8}{1}  
\node at (1.5,2.2) {$x$};
\node at (4.5,4.5) {$y$};
\node at (8.5,3.5) {$z$};
\windmill{14}{0}{3}{4}{2} 
\node at (15.5,2.2) {$x$};
\node at (16.0,5.5) {$y$};
\node at (18.5,4.2) {$z$};
\windmill{27}{-1}{5}{2}{2} 
\node at (29.5,3.4) {$x$};
\node at (28.0,6.6) {$y$};
\node at (29.5,5.0) {$z$};
\end{tikzpicture}
\caption{Typical windmills, where \HOLinline{\HOLConst{windmill}\;\HOLFreeVar{x}\;\HOLFreeVar{y}\;\HOLFreeVar{z}\;\HOLSymConst{=}\;\HOLFreeVar{x}\HOLSymConst{\ensuremath{\sp{2}}}\;\HOLSymConst{\ensuremath{+}}\;\HOLNumLit{4}\HOLSymConst{\ensuremath{}}\HOLFreeVar{y}\HOLSymConst{\ensuremath{}}\HOLFreeVar{z}}. The rightmost one has \HOLinline{\HOLFreeVar{y}\;\HOLSymConst{=}\;\HOLFreeVar{z}}.}
\Description{This figure shows typical windmills. The central square is x by x, the four rectangles arranged clockwise around the square are all y by z. The rightmost one has y and z equal.}
\label{fig:windmill-sample} 
\end{center}
\end{figure*}

\subsection{Windmills}
\label{sec:windmills}
The following expression will be our main focus:
\begin{definition}
\label{def:windmill-def}
A windmill consists of a central square with four identical rectangular arms.
\begin{HOLmath}
\;\;\HOLConst{windmill}\;\HOLFreeVar{x}\;\HOLFreeVar{y}\;\HOLFreeVar{z}\;\HOLTokenDefEquality{}\;\HOLFreeVar{x}\HOLSymConst{\ensuremath{\sp{2}}}\;\HOLSymConst{\ensuremath{+}}\;\HOLNumLit{4}\HOLSymConst{\ensuremath{}}\HOLFreeVar{y}\HOLSymConst{\ensuremath{}}\HOLFreeVar{z}
\end{HOLmath}
\end{definition}

\noindent
Some typical windmills are shown in Figure~\ref{fig:windmill-sample}.
The first term \HOLinline{\HOLFreeVar{x}\HOLSymConst{\ensuremath{\sp{2}}}} is given by a central square of side $x$,
and the second term \HOLinline{\HOLNumLit{4}\HOLSymConst{\ensuremath{}}\HOLFreeVar{y}\HOLSymConst{\ensuremath{}}\HOLFreeVar{z}} is given by four arms, each a rectangle of width $y$ and height $z$, arranged clockwise around the square.

Therefore each triple in the set $S_{n}$ of Equation~\eqref{eqn:mills-set} can be represented by a windmill, that is,
each triple $(x,y,z)$ satisfies \HOLinline{\HOLFreeVar{n}\;\HOLSymConst{=}\;\HOLConst{windmill}\;\HOLFreeVar{x}\;\HOLFreeVar{y}\;\HOLFreeVar{z}}.
Given a prime \HOLinline{\HOLFreeVar{n}\;\HOLSymConst{=}\;\HOLNumLit{4}\HOLSymConst{\ensuremath{}}\HOLFreeVar{k}\;\HOLSymConst{\ensuremath{+}}\;\HOLNumLit{1}},
we shall look for a windmill with four square arms (the one on the far right in Figure~\ref{fig:windmill-sample}),
\ie, \HOLinline{\HOLFreeVar{y}\;\HOLSymConst{=}\;\HOLFreeVar{z}}, so that \HOLinline{\HOLFreeVar{n}\;\HOLSymConst{=}\;\HOLFreeVar{x}\HOLSymConst{\ensuremath{\sp{2}}}\;\HOLSymConst{\ensuremath{+}}\;\HOLNumLit{4}\HOLSymConst{\ensuremath{}}\HOLFreeVar{y}\HOLSymConst{\ensuremath{}}\HOLFreeVar{y}} \ee \HOLinline{\HOLFreeVar{x}\HOLSymConst{\ensuremath{\sp{2}}}\;\HOLSymConst{\ensuremath{+}}\;(\HOLNumLit{2}\HOLSymConst{\ensuremath{}}\HOLFreeVar{y})\HOLSymConst{\ensuremath{\sp{2}}}}.
First, we collect all triples $(x,y,z)$ which are solutions of \HOLinline{\HOLFreeVar{n}\;\HOLSymConst{=}\;\HOLFreeVar{x}\HOLSymConst{\ensuremath{\sp{2}}}\;\HOLSymConst{\ensuremath{+}}\;\HOLNumLit{4}\HOLSymConst{\ensuremath{}}\HOLFreeVar{y}\HOLSymConst{\ensuremath{}}\HOLFreeVar{z}}:
\begin{definition}
\label{def:mills-def}
\noindent
The mills of a number is its set of windmills.
\begin{HOLmath}
\;\;\HOLConst{mills}\;\HOLFreeVar{n}\;\HOLTokenDefEquality{}\;\HOLTokenLeftbrace{}(\HOLBoundVar{x}\HOLSymConst{,}\HOLBoundVar{y}\HOLSymConst{,}\HOLBoundVar{z})\;\HOLTokenBar{}\;\HOLFreeVar{n}\;\HOLSymConst{=}\;\HOLConst{windmill}\;\HOLBoundVar{x}\;\HOLBoundVar{y}\;\HOLBoundVar{z}\HOLTokenRightbrace{}
\end{HOLmath}
\end{definition}
\noindent
This is the formal definition of the set $S_{n}$ of Equation~\eqref{eqn:mills-set}.
The conditions for a proper windmill, with all lengths nonzero, are:
\begin{HOLmath}
\HOLTokenTurnstile{}\HOLSymConst{\HOLTokenNeg{}}\HOLConst{square}\;\HOLFreeVar{n}\;\HOLSymConst{\HOLTokenConj{}}\;\HOLFreeVar{n}\;\ensuremath{\not\equiv}\;\HOLNumLit{0}\;\ensuremath{(}\ensuremath{\bmod}\;\HOLNumLit{4}\ensuremath{)}\;\HOLSymConst{\HOLTokenImp{}}\\
\;\;\;\;\;\;\;\HOLSymConst{\HOLTokenForall{}}\HOLBoundVar{x}\;\HOLBoundVar{y}\;\HOLBoundVar{z}.\;(\HOLBoundVar{x}\HOLSymConst{,}\HOLBoundVar{y}\HOLSymConst{,}\HOLBoundVar{z})\;\HOLSymConst{\HOLTokenIn{}}\;\HOLConst{mills}\;\HOLFreeVar{n}\;\HOLSymConst{\HOLTokenImp{}}\;\HOLBoundVar{x}\;\HOLSymConst{\HOLTokenNotEqual{}}\;\HOLNumLit{0}\;\HOLSymConst{\HOLTokenConj{}}\;\HOLBoundVar{y}\;\HOLSymConst{\HOLTokenNotEqual{}}\;\HOLNumLit{0}\;\HOLSymConst{\HOLTokenConj{}}\;\HOLBoundVar{z}\;\HOLSymConst{\HOLTokenNotEqual{}}\;\HOLNumLit{0}
\end{HOLmath}
\begin{equation}
\label{eqn:mills-triple-nonzero}
\end{equation}
When $n$ is a square, $\HOLinline{\HOLFreeVar{n}\;\HOLSymConst{=}\;\HOLFreeVar{x}\HOLSymConst{\ensuremath{\sp{2}}}\;\HOLSymConst{\ensuremath{+}}\;\HOLNumLit{4}\HOLSymConst{\ensuremath{}}\HOLFreeVar{y}}(0)$ for any value of $y$.
This would make \HOLinline{\HOLConst{mills}\;\HOLFreeVar{n}} infinite. Otherwise:
\begin{theorem}
\label{thm:mills-finite}
\script{windmill}{714}
The number of windmills for a number $n$ is finite if and only if $n$ is not a square.
\begin{HOLmath}
\HOLTokenTurnstile{}\HOLConst{\HOLConst{finite}}\;(\HOLConst{mills}\;\HOLFreeVar{n})\;\HOLSymConst{\HOLTokenEquiv{}}\;\HOLSymConst{\HOLTokenNeg{}}\HOLConst{square}\;\HOLFreeVar{n}
\end{HOLmath}
\end{theorem}
\noindent
Given an odd $n$ that is not a square, we can determine all its windmill triples $(x,y,z)$ by noting that, since \HOLinline{\HOLNumLit{4}\HOLSymConst{\ensuremath{}}\HOLFreeVar{y}\HOLSymConst{\ensuremath{}}\HOLFreeVar{z}} is even, $x$ must be odd, and $y$ and $z$ form the product \HOLinline{\HOLFreeVar{y}\HOLSymConst{\ensuremath{}}\HOLFreeVar{z}\;\HOLSymConst{=}\;(\HOLFreeVar{n}\;\HOLSymConst{\ensuremath{-}}\;\HOLFreeVar{x}\HOLSymConst{\ensuremath{\sp{2}}})\;\HOLConst{\HOLConst{div}}\;\HOLNumLit{4}}.
In Table~\ref{tbl:windmills-29} this is worked out for \HOLinline{\HOLFreeVar{n}\;\HOLSymConst{=}\;\HOLNumLit{29}}, using successive odd $x$ and factors for the product $yz$.
The corresponding windmills are shown in Figure~\ref{fig:windmills-29}.

\begin{table*}[h] 
\caption{Determine all the windmill triples of \HOLinline{\HOLFreeVar{n}\;\HOLSymConst{=}\;\HOLNumLit{29}}, by odd $x$ and factors of $yz$.}
\Description{This table shows how to detemine all the windmill triples for n = 29, using odd x and factors of the product yz.}
\label{tbl:windmills-29}  
\begin{tabular}{r@{\quad}@{\quad}r@{\quad\ee\quad}r@{\quad}@{\quad}l@{\quad}l}
odd $x$ & \HOLinline{\HOLFreeVar{n}\;\HOLSymConst{\ensuremath{-}}\;\HOLFreeVar{x}\HOLSymConst{\ensuremath{\sp{2}}}} & \HOLinline{\HOLNumLit{4}\HOLSymConst{\ensuremath{}}\HOLFreeVar{y}\HOLSymConst{\ensuremath{}}\HOLFreeVar{z}} & triple $(x,y,z)$ & comment\\
\hline
$1$ & $29 - 1^{2}$ & $28 \ee 4(7)$ & $(1,1,7), (1,7,1)$ & factors of $7$ are $1, 7$.\\
$3$ & $29 - 3^{2}$ & $20 \ee 4(5)$ & $(3,1,5), (3,5,1)$ & factors of $5$ are $1, 5$.\\
$5$ & $29 - 5^{2}$ &  $4 \ee 4(1)$ & $(5,1,1)$          & factor of $1$ is $1$.\\
\end{tabular}
\end{table*}

\begin{figure*}[h]  
\begin{center}
\begin{tikzpicture}[scale=0.22] 
\draw[step=1, color=white!60!black] (0,0) grid (65,17);
\windmill{8}{8}{1}{1}{7}  
\windmill{23}{7}{3}{1}{5} 
\windmill{34}{6}{5}{1}{1} 
\windmill{44}{7}{3}{5}{1} 
\windmill{57}{8}{1}{7}{1} 
\node at (4,1)  {$(1,1,7)$};
\node at (22,1) {$(3,1,5)$};
\node at (36,1) {$(5,1,1)$};
\node at (45,1) {$(3,5,1)$};
\node at (55,1) {$(1,7,1)$};
\end{tikzpicture}
\caption{All the windmills of \HOLinline{\HOLFreeVar{n}\;\HOLSymConst{=}\;\HOLNumLit{29}}, determined from Table~\ref{tbl:windmills-29}.}
\Description{This figure shows all the windmills of n = 29, determined from the previous table.}
\label{fig:windmills-29} 
\end{center}
\end{figure*}
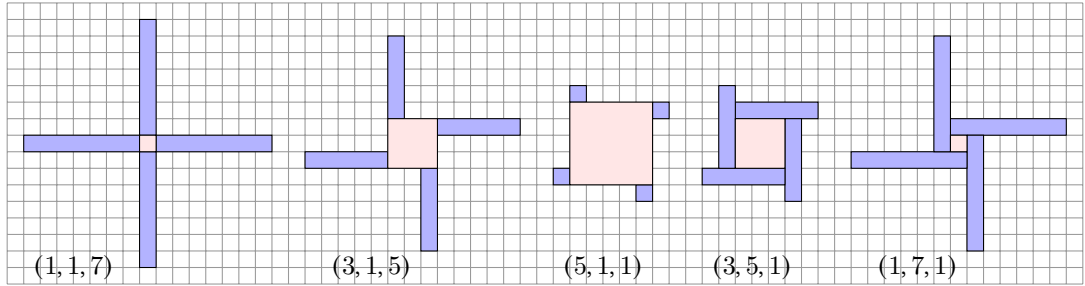

When a number $n$ has the form \HOLinline{\HOLNumLit{4}\HOLSymConst{\ensuremath{}}\HOLFreeVar{k}\;\HOLSymConst{\ensuremath{+}}\;\HOLNumLit{1}},
\[
n \ee 1^{2} + 4(1)k \ee \HOLinline{\HOLConst{windmill}\;\HOLNumLit{1}\;\HOLNumLit{1}\;\HOLFreeVar{k}},
\]
showing that its \HOLinline{\HOLConst{mills}\;\HOLFreeVar{n}\;\HOLSymConst{\HOLTokenNotEqual{}}\;\HOLSymConst{\HOLTokenEmpty{}}}:
\begin{HOLmath}
\HOLTokenTurnstile{}\HOLFreeVar{n}\;\ensuremath{\equiv}\;\HOLNumLit{1}\;\ensuremath{(}\ensuremath{\bmod}\;\HOLNumLit{4}\ensuremath{)}\;\HOLSymConst{\HOLTokenImp{}}\;(\HOLNumLit{1}\HOLSymConst{,}\HOLNumLit{1}\HOLSymConst{,}\HOLFreeVar{n}\;\HOLConst{\HOLConst{div}}\;\HOLNumLit{4})\;\HOLSymConst{\HOLTokenIn{}}\;\HOLConst{mills}\;\HOLFreeVar{n}
\end{HOLmath}
Moreover, when this form corresponds to a prime, this is the only triple $(x,y,z)$ with \HOLinline{\HOLFreeVar{x}\;\HOLSymConst{=}\;\HOLFreeVar{y}}:
\begin{theorem}
\label{thm:mills-trivial-prime}
\script{windmill}{428}
For a prime of the form \HOLinline{\HOLNumLit{4}\HOLSymConst{\ensuremath{}}\HOLFreeVar{k}\;\HOLSymConst{\ensuremath{+}}\;\HOLNumLit{1}}, the only windmill with the first and second parameters equal is \HOLinline{\HOLConst{windmill}\;\HOLNumLit{1}\;\HOLNumLit{1}\;\HOLFreeVar{k}}.
\begin{HOLmath}
\HOLTokenTurnstile{}\HOLConst{prime}\;\HOLFreeVar{n}\;\HOLSymConst{\HOLTokenConj{}}\;\HOLFreeVar{n}\;\ensuremath{\equiv}\;\HOLNumLit{1}\;\ensuremath{(}\ensuremath{\bmod}\;\HOLNumLit{4}\ensuremath{)}\;\HOLSymConst{\HOLTokenImp{}}\\
\;\;\;\;\;\;\;\HOLSymConst{\HOLTokenForall{}}\HOLBoundVar{x}\;\HOLBoundVar{z}.\;\HOLFreeVar{n}\;\HOLSymConst{=}\;\HOLConst{windmill}\;\HOLBoundVar{x}\;\HOLBoundVar{x}\;\HOLBoundVar{z}\;\HOLSymConst{\HOLTokenEquiv{}}\;\HOLBoundVar{x}\;\HOLSymConst{=}\;\HOLNumLit{1}\;\HOLSymConst{\HOLTokenConj{}}\;\HOLBoundVar{z}\;\HOLSymConst{=}\;\HOLFreeVar{n}\;\HOLConst{\HOLConst{div}}\;\HOLNumLit{4}
\end{HOLmath}
\end{theorem}
\begin{proof}
Note that \HOLinline{\HOLFreeVar{k}\;\HOLSymConst{=}\;\HOLFreeVar{n}\;\HOLConst{\HOLConst{div}}\;\HOLNumLit{4}} for prime \HOLinline{\HOLFreeVar{n}\;\HOLSymConst{=}\;\HOLNumLit{4}\HOLSymConst{\ensuremath{}}\HOLFreeVar{k}\;\HOLSymConst{\ensuremath{+}}\;\HOLNumLit{1}}.
Consider \HOLinline{(\HOLFreeVar{x}\HOLSymConst{,}\HOLFreeVar{y}\HOLSymConst{,}\HOLFreeVar{z})\;\HOLSymConst{\HOLTokenIn{}}\;\HOLConst{mills}\;\HOLFreeVar{n}} with \HOLinline{\HOLFreeVar{x}\;\HOLSymConst{=}\;\HOLFreeVar{y}}.
This implies,
\[
\HOLinline{\HOLFreeVar{n}\;\HOLSymConst{=}\;\HOLConst{windmill}\;\HOLFreeVar{x}\;\HOLFreeVar{x}\;\HOLFreeVar{z}} \ee \HOLinline{\HOLFreeVar{x}\HOLSymConst{\ensuremath{\sp{2}}}\;\HOLSymConst{\ensuremath{+}}\;\HOLNumLit{4}\HOLSymConst{\ensuremath{}}\HOLFreeVar{x}\HOLSymConst{\ensuremath{}}\HOLFreeVar{z}\;\HOLSymConst{=}\;\HOLFreeVar{x}\HOLSymConst{\ensuremath{}}(\HOLFreeVar{x}\;\HOLSymConst{\ensuremath{+}}\;\HOLNumLit{4}\HOLSymConst{\ensuremath{}}\HOLFreeVar{z})}.
\]
Therefore \HOLinline{\HOLFreeVar{x}\;\HOLConst{\ensuremath{\mid}}\;\HOLFreeVar{n}}.
As prime $n$ is not a square, \HOLinline{\HOLFreeVar{x}\;\HOLSymConst{\HOLTokenLt{}}\;\HOLFreeVar{n}}.
Hence \HOLinline{\HOLFreeVar{x}\;\HOLSymConst{=}\;\HOLNumLit{1}}, so \HOLinline{\HOLFreeVar{y}\;\HOLSymConst{=}\;\HOLNumLit{1}}, and \HOLinline{\HOLFreeVar{z}\;\HOLSymConst{=}\;\HOLFreeVar{k}}.  
\end{proof}

\subsection{Involution}
\label{sec:involution}
We are going to study involutions on \HOLinline{\HOLConst{mills}\;\HOLFreeVar{n}}, the set of windmills for $n$.
A function $f$ is an involution on a set $S$, denoted by \HOLinline{\HOLFreeVar{f}\;\HOLConst{involute}\;\HOLFreeVar{S}}, when it is its own inverse:
\begin{HOLmath}
\;\;\HOLFreeVar{f}\;\HOLConst{involute}\;\HOLFreeVar{S}\;\HOLTokenDefEquality{}\;\HOLSymConst{\HOLTokenForall{}}\HOLBoundVar{x}.\;\HOLBoundVar{x}\;\HOLSymConst{\HOLTokenIn{}}\;\HOLFreeVar{S}\;\HOLSymConst{\HOLTokenImp{}}\;\HOLFreeVar{f}\;\HOLBoundVar{x}\;\HOLSymConst{\HOLTokenIn{}}\;\HOLFreeVar{S}\;\HOLSymConst{\HOLTokenConj{}}\;\HOLFreeVar{f}\;(\HOLFreeVar{f}\;\HOLBoundVar{x})\;\HOLSymConst{=}\;\HOLBoundVar{x}
\end{HOLmath}
That is, $f$ is a bijection \HOLinline{\HOLFreeVar{f}\;\ensuremath{:}\;\HOLFreeVar{S}\;\ensuremath{\leftrightarrow}\;\HOLFreeVar{S}}, pairing up $x$ and~\HOLinline{\HOLFreeVar{f}\;\HOLFreeVar{x}}, both in $S$. When \HOLinline{\HOLFreeVar{x}\;\HOLSymConst{=}\;\HOLFreeVar{f}\;\HOLFreeVar{x}}, the element $x$ is fixed by the involution $f$.
We define the following sets:
\begin{definition}
\label{def:involute-pairs-fixes-def}
The pairs and fixes of an involution $f$ on a set $S$.
\begin{HOLmath}
\;\;\HOLConst{pairs}\;\HOLFreeVar{f}\;\HOLFreeVar{S}\;\HOLTokenDefEquality{}\;\HOLTokenLeftbrace{}\HOLBoundVar{x}\;\HOLTokenBar{}\;\HOLBoundVar{x}\;\HOLSymConst{\HOLTokenIn{}}\;\HOLFreeVar{S}\;\HOLSymConst{\HOLTokenConj{}}\;\HOLFreeVar{f}\;\HOLBoundVar{x}\;\HOLSymConst{\HOLTokenNotEqual{}}\;\HOLBoundVar{x}\HOLTokenRightbrace{}\\
\;\;\HOLConst{fixes}\;\HOLFreeVar{f}\;\HOLFreeVar{S}\;\HOLTokenDefEquality{}\;\HOLTokenLeftbrace{}\HOLBoundVar{x}\;\HOLTokenBar{}\;\HOLBoundVar{x}\;\HOLSymConst{\HOLTokenIn{}}\;\HOLFreeVar{S}\;\HOLSymConst{\HOLTokenConj{}}\;\HOLFreeVar{f}\;\HOLBoundVar{x}\;\HOLSymConst{=}\;\HOLBoundVar{x}\HOLTokenRightbrace{}\\
\end{HOLmath}
\end{definition}
\noindent
Clearly they are disjoint.
The subset $\HOLinline{\HOLConst{pairs}\;\HOLFreeVar{f}\;\HOLFreeVar{S}}$ consists of distinct involute pairs, so its cardinality is even:
\begin{HOLmath}
\HOLTokenTurnstile{}\HOLConst{\HOLConst{finite}}\;\HOLFreeVar{S}\;\HOLSymConst{\HOLTokenConj{}}\;\HOLFreeVar{f}\;\HOLConst{involute}\;\HOLFreeVar{S}\;\HOLSymConst{\HOLTokenImp{}}\;\HOLConst{\HOLConst{even}}\;\ensuremath{|}\HOLConst{pairs}\;\HOLFreeVar{f}\;\HOLFreeVar{S}\ensuremath{|}
\end{HOLmath}
So both \HOLinline{\ensuremath{|}\HOLFreeVar{S}\ensuremath{|}} and \HOLinline{\ensuremath{|}\HOLConst{fixes}\;\HOLFreeVar{f}\;\HOLFreeVar{S}\ensuremath{|}} have the same parity. This leads to:
\begin{theorem}
\label{thm:involute-two-fixes-both-odd}
\script{involuteFix}{1182}
If two involutions act on the same finite set $S$, their fixes have the same parity.
\begin{HOLmath}
\HOLTokenTurnstile{}\HOLConst{\HOLConst{finite}}\;\HOLFreeVar{S}\;\HOLSymConst{\HOLTokenConj{}}\;\HOLFreeVar{f}\;\HOLConst{involute}\;\HOLFreeVar{S}\;\HOLSymConst{\HOLTokenConj{}}\;\HOLFreeVar{g}\;\HOLConst{involute}\;\HOLFreeVar{S}\;\HOLSymConst{\HOLTokenImp{}}\\
\;\;\;\;\;\;\;(\HOLConst{\HOLConst{odd}}\;\ensuremath{|}\HOLConst{fixes}\;\HOLFreeVar{f}\;\HOLFreeVar{S}\ensuremath{|}\;\HOLSymConst{\HOLTokenEquiv{}}\;\HOLConst{\HOLConst{odd}}\;\ensuremath{|}\HOLConst{fixes}\;\HOLFreeVar{g}\;\HOLFreeVar{S}\ensuremath{|})
\end{HOLmath}
\end{theorem}
\noindent
We shall meet the two involutions on \HOLinline{\HOLConst{mills}\;\HOLFreeVar{n}}, a set which is finite for non-square $n$ (by Theorem~\ref{thm:mills-finite}).

\section{Windmill Involutions}
\label{sec:windmill-involutions}
Zagier's one-sentence proof is the interplay of two involutions on the set of windmills (\HOLinline{\HOLConst{mills}\;\HOLFreeVar{n}}) for a prime \HOLinline{\HOLFreeVar{n}\;\HOLSymConst{=}\;\HOLNumLit{4}\HOLSymConst{\ensuremath{}}\HOLFreeVar{k}\;\HOLSymConst{\ensuremath{+}}\;\HOLNumLit{1}}.

\subsection{Flip Map}
\label{sec:flip-map}
The first involution just swaps the $y$ and $z$ in the triple $(x,y,z)$:
\begin{definition}
\label{def:flip-def}
The flip map for a triple.
\begin{HOLmath}
\;\;\HOLConst{flip}\;(\HOLFreeVar{x}\HOLSymConst{,}\HOLFreeVar{y}\HOLSymConst{,}\HOLFreeVar{z})\;\HOLTokenDefEquality{}\;(\HOLFreeVar{x}\HOLSymConst{,}\HOLFreeVar{z}\HOLSymConst{,}\HOLFreeVar{y})
\end{HOLmath}
\end{definition}
\noindent
The set \HOLinline{\HOLFreeVar{S}\;\HOLSymConst{=}\;\HOLConst{mills}\;\HOLFreeVar{n}} of windmill triples of a number $n$ can be partitioned by $y, z$ into:
\[
\begin{array}{c}
S_{y < z} \ee \{(x,y,z) \in S \mid y < z\}\\
S_{y \ee z} \ee \{(x,y,z) \in S \mid y \ee z\}\\
S_{y > z} \ee \{(x,y,z) \in S, \mid y > z\}\\
\end{array}
\]
An example for \HOLinline{\HOLFreeVar{n}\;\HOLSymConst{=}\;\HOLNumLit{29}} is shown in Figure~\ref{fig:windmills-29-by-flip}. 
Clearly there is a bijection: $\HOLConst{flip}\colon S_{y < z} \leftrightarrow S_{y > z}$,
and $S_{y \ee z} \ee \HOLinline{\HOLConst{fixes}\;\HOLConst{flip}\;\HOLFreeVar{S}}$.
Thus the inverse of flip is itself:
\begin{HOLmath}
\HOLTokenTurnstile{}\HOLConst{flip}\;(\HOLConst{flip}\;(\HOLFreeVar{x}\HOLSymConst{,}\HOLFreeVar{y}\HOLSymConst{,}\HOLFreeVar{z}))\;\HOLSymConst{=}\;(\HOLFreeVar{x}\HOLSymConst{,}\HOLFreeVar{y}\HOLSymConst{,}\HOLFreeVar{z})
\end{HOLmath}
showing that:
\begin{theorem}
\label{thm:flip-involute-mills}
\script{windmill}{933}
The flip map is an involution on the set of windmills.
\begin{HOLmath}
\HOLTokenTurnstile{}\HOLConst{flip}\;\HOLConst{involute}\;\HOLConst{mills}\;\HOLFreeVar{n}
\end{HOLmath}
\end{theorem}

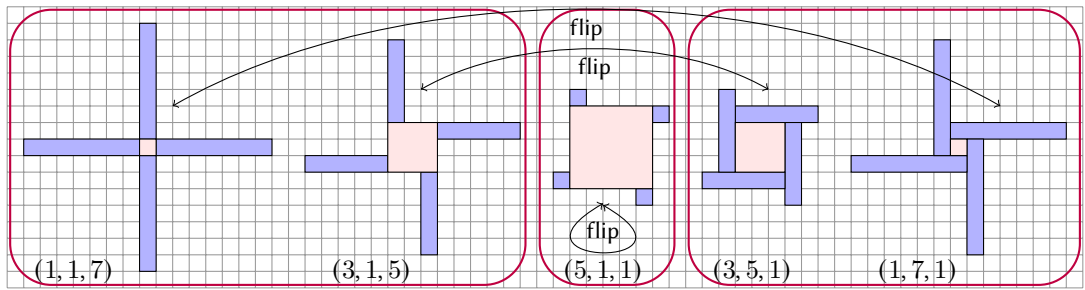
\begin{figure*}[h]  
\begin{center}
\begin{tikzpicture}[scale=0.22] 
\draw[step=1, color=white!60!black] (0,0) grid (65,17);
\windmill{8}{8}{1}{1}{7}  
\windmill{23}{7}{3}{1}{5} 
\windmill{34}{6}{5}{1}{1} 
\windmill{44}{7}{3}{5}{1} 
\windmill{57}{8}{1}{7}{1} 
\node at (4,1)  {$(1,1,7)$};
\node at (22,1) {$(3,1,5)$};
\node at (36,1) {$(5,1,1)$};
\node at (45,1) {$(3,5,1)$};
\node at (55,1) {$(1,7,1)$};
\coordinate (a) at (0.5,0.5);
\coordinate (b) at (31,16.5);
\node[draw,thick,color=purple,fit= (a) (b),rounded corners=.55cm,inner sep=2pt] {};
\coordinate (a) at (32.5,0.5);
\coordinate (b) at (40,16.5);
\node[draw,thick,color=purple,fit= (a) (b),rounded corners=.55cm,inner sep=2pt] {};
\coordinate (a) at (41.5,0.5);
\coordinate (b) at (64.5,16.5);
\node[draw,thick,color=purple,fit= (a) (b),rounded corners=.55cm,inner sep=2pt] {};
\begin{scope}[scale=10]
\coordinate (a) at (1.0,1.1); 
\coordinate (b) at (6.0,1.1); 
\draw[<->] (a) to [bend left,looseness=0.8] node[midway,below] {\HOLConst{flip}} (b);
\coordinate (a) at (2.5,1.2); 
\coordinate (b) at (4.6,1.2); 
\draw[<->] (a) to [bend left,looseness=0.8] node[midway,below] {\HOLConst{flip}} (b);
\coordinate (a) at (3.6,0.5); 
\coordinate (b) at (3.6,0.51); 
\draw[<->] (a) to [out=-30, in=-150, looseness=200]
          node[midway,above] {\HOLConst{flip}} (b); 
\end{scope}
\end{tikzpicture}
\caption{Partition of windmills of \HOLinline{\HOLFreeVar{n}\;\HOLSymConst{=}\;\HOLNumLit{29}} for \HOLConst{flip}: those with $y < z, y \ee z$, and $y > z$. Note left and right pairing.}
\Description{This figure shows a partition of the windmills of n = 29 for the flip map: those with y < z, y = z, and y > z. Note the left and right pairing between y < z and y > z.}
\label{fig:windmills-29-by-flip} 
\end{center}
\end{figure*}

\subsection{Zagier Map}
\label{sec:zagier-map}
The other involution is the one devised by Don Zagier, as shown in Equation~\eqref{eqn:zagier-map}:
\begin{definition}
\label{def:zagier-def}
The Zagier map for a triple.
\begin{HOLmath}
\;\;\HOLConst{zagier}\;(\HOLFreeVar{x}\HOLSymConst{,}\HOLFreeVar{y}\HOLSymConst{,}\HOLFreeVar{z})\;\HOLTokenDefEquality{}\\
\;\;\;\;\HOLKeyword{if}\;\HOLFreeVar{x}\;\HOLSymConst{\HOLTokenLt{}}\;\HOLFreeVar{y}\;\HOLSymConst{\ensuremath{-}}\;\HOLFreeVar{z}\;\HOLKeyword{then}\;(\HOLFreeVar{x}\;\HOLSymConst{\ensuremath{+}}\;\HOLNumLit{2}\HOLSymConst{\ensuremath{}}\HOLFreeVar{z}\HOLSymConst{,}\HOLFreeVar{z}\HOLSymConst{,}\HOLFreeVar{y}\;\HOLSymConst{\ensuremath{-}}\;\HOLFreeVar{z}\;\HOLSymConst{\ensuremath{-}}\;\HOLFreeVar{x})\\
\;\;\;\;\HOLKeyword{else}\;\HOLKeyword{if}\;\HOLFreeVar{x}\;\HOLSymConst{\HOLTokenLt{}}\;\HOLNumLit{2}\HOLSymConst{\ensuremath{}}\HOLFreeVar{y}\;\HOLKeyword{then}\;(\HOLNumLit{2}\HOLSymConst{\ensuremath{}}\HOLFreeVar{y}\;\HOLSymConst{\ensuremath{-}}\;\HOLFreeVar{x}\HOLSymConst{,}\HOLFreeVar{y}\HOLSymConst{,}\HOLFreeVar{x}\;\HOLSymConst{\ensuremath{+}}\;\HOLFreeVar{z}\;\HOLSymConst{\ensuremath{-}}\;\HOLFreeVar{y})\\
\;\;\;\;\HOLKeyword{else}\;(\HOLFreeVar{x}\;\HOLSymConst{\ensuremath{-}}\;\HOLNumLit{2}\HOLSymConst{\ensuremath{}}\HOLFreeVar{y}\HOLSymConst{,}\HOLFreeVar{x}\;\HOLSymConst{\ensuremath{+}}\;\HOLFreeVar{z}\;\HOLSymConst{\ensuremath{-}}\;\HOLFreeVar{y}\HOLSymConst{,}\HOLFreeVar{y})
\end{HOLmath}
\end{definition}
\noindent
Algebraically, this is indeed an involution, as HOL4 can verify without a blink:
\begin{HOLmath}
\HOLTokenTurnstile{}\HOLFreeVar{x}\;\HOLSymConst{\HOLTokenNotEqual{}}\;\HOLNumLit{0}\;\HOLSymConst{\HOLTokenConj{}}\;\HOLFreeVar{z}\;\HOLSymConst{\HOLTokenNotEqual{}}\;\HOLNumLit{0}\;\HOLSymConst{\HOLTokenImp{}}\;\HOLConst{zagier}\;(\HOLConst{zagier}\;(\HOLFreeVar{x}\HOLSymConst{,}\HOLFreeVar{y}\HOLSymConst{,}\HOLFreeVar{z}))\;\HOLSymConst{=}\;(\HOLFreeVar{x}\HOLSymConst{,}\HOLFreeVar{y}\HOLSymConst{,}\HOLFreeVar{z})
\end{HOLmath}
\begin{equation}
\label{eqn:zagier-involute}
\end{equation}
That HOL4 can verify this directly from definition is a showcase of its excellent algebraic simplifier, especially for natural numbers.
However, we would like to see the magic behind, in terms of the geometry of windmills.
Note that this definition differs slightly from Equation~\eqref{eqn:zagier-map} since the else-parts include boundary cases.
They actually correspond to improper windmills, and they are irrelevant for the values of $n$ satisfying Equation~\eqref{eqn:mills-triple-nonzero}.

\subsection{Mind of a Windmill}
\label{sec:windmill-mind}
The main purpose of introducing windmills is to read their minds.

\begin{figure*}[h] 
\begin{center}
\begin{tikzpicture}[scale=0.28] 
\windmill{0}{0}{3}{8}{1}  
\windmill{14}{0}{3}{8}{1} 
\windmill{27}{-1}{5}{1}{4} 
\node at (1.5,2.2) {$x$};
\node at (5.0,4.5) {$y$};
\node at (8.5,3.5) {$z$};
\draw[decorate,thick,decoration={brace,amplitude=5pt,mirror,raise=2pt}]
       (14,2.8) -- (17,2.8) node[midway,below,yshift=-7pt] {$x$};
\draw[decorate,thick,decoration={brace,amplitude=5pt,raise=2pt}]
       (13,4.2) -- (18,4.2) node[midway,above,yshift=7pt] {$x'$};
\node at (29.5,3.4) {$x'$};
\node at (27.5,8.6) {$y'$};
\node at (28.5,6.2) {$z'$};
\mind{13}{-1}{5}
\end{tikzpicture}
\caption{A typical \HOLinline{\HOLConst{windmill}\;\HOLFreeVar{x}\;\HOLFreeVar{y}\;\HOLFreeVar{z}\;\HOLSymConst{=}\;\HOLFreeVar{x}\HOLSymConst{\ensuremath{\sp{2}}}\;\HOLSymConst{\ensuremath{+}}\;\HOLNumLit{4}\HOLSymConst{\ensuremath{}}\HOLFreeVar{y}\HOLSymConst{\ensuremath{}}\HOLFreeVar{z}}, with a mind (in dashes) and transforms to another windmill.}
\Description{This figure shows a typical windmill, with a mind in dashes, on the left. The figure also illustrates how the left windmill transforms to another windmill on the right, with the same mind.}
\label{fig:windmill-mind} 
\end{center}
\end{figure*}
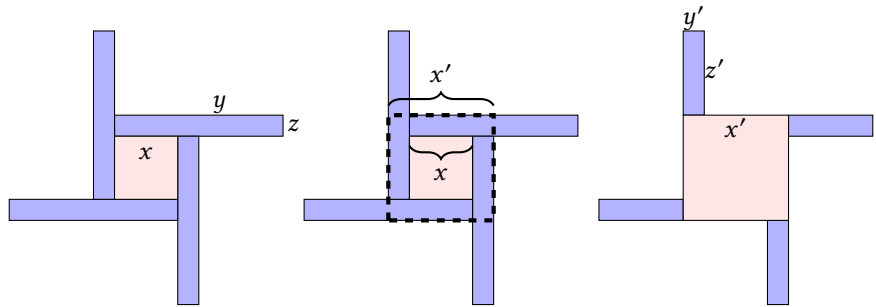

Referring to Figure~\ref{fig:windmill-mind}, a windmill has a mind (marked in dashes at middle), which is the maximum central square, with side $x'$, that can be fitted with the four arms.
When \HOLinline{\HOLFreeVar{x}\;\HOLSymConst{\HOLTokenLeq{}}\;\ensuremath{\HOLFreeVar{x}\sp{\prime{}}}}, the original square \HOLinline{\HOLFreeVar{x}\ensuremath{{\sp{\HOLNumLit{2}}}}} can grow to the mind \HOLinline{\ensuremath{\HOLFreeVar{x}\sp{\prime{}}}\ensuremath{{\sp{\HOLNumLit{2}}}}}, forming another windmill but keeping the overall shape (on the right).
Conversely, going from right to left, we can use the mind as a reference to shrink the square term from \HOLinline{\ensuremath{\HOLFreeVar{x}\sp{\prime{}}}\ensuremath{{\sp{\HOLNumLit{2}}}}} to \HOLinline{\HOLFreeVar{x}\ensuremath{{\sp{\HOLNumLit{2}}}}} by trimming four sides,
thereby restoring the arms to original.
Transforming a windmill's square term through the mind is the geometric interpretation of Equation~\eqref{eqn:zagier-map}.

\begin{table*}[h] 
\caption{The five cases of Zagier map, transforming a triple $(x,y,z)$ to $(x',y',z')$.}
\Description{This figure shows the five cases of Zagier map, transforming a triple (x,y,z) to (x',y',z').}
\label{tbl:zagier-map}  
\begin{tabular}{c@{\quad}l@{\quad}l@{\quad}r@{\quad}l@{\quad}r@{\quad}r@{\quad}r@{\quad}l@{\quad}l}
Case & Type & condition & Mind & Picture & $x'$ & $y'$ & $z'$ & condition & Type\\
\hline
$1$ & \multirow{ 2}{*}{$x < y$} & $x < y - z$   & $x + 2z$ & Figure~\ref{fig:zagier-map}~(a)
    & $x + 2z$ & $z$ & $y - x - z$ & $2y' < x'$ & \multirow{ 2}{*}{$y' < x'$}\\
$2$ &  & $y - z < x$ & $2y - x$ & Figure~\ref{fig:zagier-map}~(b)
    & $2y - x$ & $y$ & $x + z - y$ & $x' < 2y'$ & \\
\hline
$3$ & $x = y$ &         & $x$ & Figure~\ref{fig:zagier-map}~(c)
    & $x$      & $y$ & $z$         &  & $x' = y'$\\
\hline
$4$ & \multirow{ 2}{*}{$y < x$} & $x < 2y$  & $x$ & Figure~\ref{fig:zagier-map}~(d)
    & $2y - x$ & $y$ & $x + z - y$ & $y' - z' < x'$ & \multirow{ 2}{*}{$x' < y'$}\\
$5$ & & $2y < x$    & $x$ & Figure~\ref{fig:zagier-map}~(e)
    & $x - 2y$ & $x + z - y$ & $y$ & $x' < y' - z'$ & \\ 
\hline
\end{tabular}
\end{table*}

\begin{figure*}[htbp]  
\begin{center}
\begin{tikzpicture}[scale=0.3] 
\windmill{0}{0}{3}{8}{2}    
\windmill{15}{-2}{7}{2}{3}  
\mind{-2}{-2}{7}
\mind{15}{-2}{7}
\node at (1.5,2.2) {$x$};
\node at (4.0,5.6) {$y$};
\node at (8.5,4.0) {$z$};
\draw[->,thick] (10,3) -- (12,3) node[midway,above] {\HOLConst{zagier}};
\node at (18.5,4.0) {$x'$};
\node at (26.0,4.0) {$y'$};
\node at (23.5,1.5) {$z'$};
\draw[color=red,ultra thick] (5,3) -- (8,3);   
\draw[color=red,ultra thick] (22,3) -- (25,3); 
\draw[decorate,thick,decoration={brace,amplitude=5pt,mirror,raise=2pt}]
       (5,3) -- (8,3); 
\draw[decorate,thick,decoration={brace,amplitude=5pt,mirror,raise=2pt}]
       (22,3) -- (25,3); 
\node at (0.8,10) {\large{(a) Case $1$: \HOLinline{\HOLFreeVar{x}\;\HOLSymConst{\HOLTokenLt{}}\;\HOLFreeVar{y}\;\HOLSymConst{\ensuremath{-}}\;\HOLFreeVar{z}}.}};
\node at (33,3) {$\large{
\begin{array}{r@{\;=\;}l}
   x' & x + 2z\\
   y' & z\\
   z' & y - x - z\\
\end{array}
}$};
\end{tikzpicture}
\begin{tikzpicture}[scale=0.3] 
\windmill{0}{0}{3}{6}{4}    
\windmillr{15}{-3}{9}{6}{1}  
\mind{-3}{-3}{9}
\mind{15}{-3}{9}
\node at (1.5,2.2) {$x$};
\node at (3.5,7.5) {$y$};
\node at (6.5,4.5) {$z$};
\draw[->,thick] (10,3) -- (12,3) node[midway,above] {\HOLConst{zagier}};
\node at (19.5,5.0) {$x'$};
\node at (21.0,7.6) {$y'$};
\node at (17.0,6.8) {$z'$};
\draw[color=red,ultra thick] (0,6) -- (0,7); 
\draw[color=red,ultra thick] (18,6) -- (18,7); 
\draw[decorate,thick,decoration={brace,amplitude=2pt,raise=2pt}]
       (0,6) -- (0,7); 
\draw[decorate,thick,decoration={brace,amplitude=2pt,raise=2pt}]
       (18,6) -- (18,7); 
\node at (6,10) {}; 
\node at (6,9) {\large{(b) Case $2$: \HOLinline{\HOLFreeVar{y}\;\HOLSymConst{\ensuremath{-}}\;\HOLFreeVar{z}\;\HOLSymConst{\HOLTokenLt{}}\;\HOLFreeVar{x}} and \HOLinline{\HOLFreeVar{x}\;\HOLSymConst{\HOLTokenLt{}}\;\HOLFreeVar{y}}, so \HOLinline{\HOLFreeVar{x}\;\HOLSymConst{\HOLTokenLt{}}\;\HOLNumLit{2}\HOLSymConst{\ensuremath{}}\HOLFreeVar{y}}.}};
\node at (33,3) {$\large{
\begin{array}{r@{\;=\;}l}
   x' & 2y - x\\
   y' & y\\
   z' & x + z - y\\
\end{array}
}$};
\node at (0,-4) {}; 
\end{tikzpicture}
\begin{tikzpicture}[scale=0.3] 
\windmill{0}{0}{4}{4}{2}   
\windmill{16}{0}{4}{4}{2}  
\mind{0}{0}{4}
\mind{16}{0}{4}
\node at (2.2,3.5) {$x$};
\node at (2.2,6.5) {$y$};
\node at (4.5,5.2) {$z$};
\draw[->,thick] (10,3) -- (12,3) node[midway,above] {\HOLConst{zagier}};
\node at (18.0,3.5) {$x'$};
\node at (18.0,6.5) {$y'$};
\node at (14.8,5.2) {$z'$};
\draw[color=red,ultra thick] (0,4) -- (0,6); 
\draw[color=red,ultra thick] (16,4) -- (16,6); 
\draw[decorate,thick,decoration={brace,amplitude=4pt,raise=2pt}]
       (0,4) -- (0,6); 
\draw[decorate,thick,decoration={brace,amplitude=4pt,raise=2pt}]
       (16,4) -- (16,6); 
\node at (7,9) {}; 
\node at (6.5,8) {\large{(c) Case $3$: \HOLinline{\HOLFreeVar{y}\;\HOLSymConst{=}\;\HOLFreeVar{x}}, so \HOLinline{\HOLFreeVar{y}\;\HOLSymConst{\ensuremath{-}}\;\HOLFreeVar{z}\;\HOLSymConst{\HOLTokenLt{}}\;\HOLFreeVar{x}} and \HOLinline{\HOLFreeVar{x}\;\HOLSymConst{\HOLTokenLt{}}\;\HOLNumLit{2}\HOLSymConst{\ensuremath{}}\HOLFreeVar{y}}.}};
\node at (33,3) {$\large{
\begin{array}{r@{\;=\;}l}
   x' & 2y - x = x\\
   y' & y\\
   z' & x + z - y = z\\
\end{array}
}$};
\end{tikzpicture}
\begin{tikzpicture}[scale=0.3] 
\windmill{0}{0}{6}{4}{2}   
\windmillr{18}{2}{2}{4}{4}  
\mind{0}{0}{6}
\mind{16}{0}{6}
\node at (3.0,5.5) {$x$};
\node at (2.5,8.5) {$y$};
\node at (4.5,6.8) {$z$};
\draw[->,thick] (10,3) -- (12,3) node[midway,above] {\HOLConst{zagier}};
\node at (19.0,3.6) {$x'$};
\node at (18.0,8.6) {$y'$};
\node at (14.5,6.0) {$z'$};
\draw[color=red,ultra thick] (0,4) -- (0,8); 
\draw[color=red,ultra thick] (16,4) -- (16,8); 
\draw[decorate,thick,decoration={brace,amplitude=5pt,raise=2pt}]
       (0,4) -- (0,8); 
\draw[decorate,thick,decoration={brace,amplitude=5pt,raise=2pt}]
       (16,4) -- (16,8); 
\node at (3,11) {}; 
\node at (3,10) {\large{(d) Case $4$: \HOLinline{\HOLFreeVar{y}\;\HOLSymConst{\HOLTokenLt{}}\;\HOLFreeVar{x}}, but \HOLinline{\HOLFreeVar{x}\;\HOLSymConst{\HOLTokenLt{}}\;\HOLNumLit{2}\HOLSymConst{\ensuremath{}}\HOLFreeVar{y}}.}};
\node at (34,3) {$\large{
\begin{array}{r@{\;=\;}l}
   x' & 2y - x\\
   y' & y\\
   z' & x + z - y\\
\end{array}
}$};
\end{tikzpicture}
\begin{tikzpicture}[scale=0.3] 
\windmill{0}{0}{7}{3}{2}   
\windmill{19}{3}{1}{6}{3}  
\mind{0}{0}{7}
\mind{16}{0}{7}
\node at (3.5,6.2) {$x$};
\node at (1.5,9.5) {$y$};
\node at (3.5,8.0) {$z$};
\draw[->,thick] (10,3) -- (12,3) node[midway,above] {\HOLConst{zagier}};
\node at (19.5,3.5) {$x'$};
\node at (22.5,7.8) {$y'$};
\node at (25.6,5.8) {$z'$};
\draw[color=yellow,ultra thick] (0,9) -- (3,9); 
\draw[color=yellow,ultra thick] (4,-2) -- (7,-2); 
\draw[color=yellow,ultra thick] (16,3) -- (19,3); 
\draw[color=yellow,ultra thick] (20,4) -- (23,4); 
\draw[decorate,thick,decoration={brace,amplitude=5pt,mirror,raise=2pt}]
       (0,9) -- (3,9); 
\draw[decorate,thick,decoration={brace,amplitude=5pt,mirror,raise=2pt}]
       (4,-2) -- (7,-2); 
\draw[decorate,thick,decoration={brace,amplitude=5pt,mirror,raise=2pt}]
       (16,3) -- (19,3); 
\draw[decorate,thick,decoration={brace,amplitude=5pt,mirror,raise=2pt}]
       (20,4) -- (23,4); 
\node at (3,12) {}; 
\node at (3,11) {\large{(e) Case $5$: \HOLinline{\HOLNumLit{2}\HOLSymConst{\ensuremath{}}\HOLFreeVar{y}\;\HOLSymConst{\HOLTokenLt{}}\;\HOLFreeVar{x}}, so \HOLinline{\HOLFreeVar{y}\;\HOLSymConst{\HOLTokenLt{}}\;\HOLFreeVar{x}}.}};
\node at (34,3) {$\large{
\begin{array}{r@{\;=\;}l}
   x' & x - 2y\\
   y' & x + z - y\\
   z' & y\\
\end{array}
}$};
\end{tikzpicture}
\caption{All five cases of the Zagier map, from $(x,y,z)$ to $(x',y',z')$ through the mind of a windmill.}
\Description{This figure shows all the 5 cases of the Zagier map, transform through the mind of a windmill.}
\label{fig:zagier-map} 
\end{center}
\end{figure*}

The Zagier map transforms $(x,y,z)$ to $(x',y',z')$ via the mind of the windmill, keeping its overall shape.
There are three types, depending on whether $x < y$, $x = y$, or $y < x$.
Both the first and last types are divided into two cases, as the geometry for the mind is different.
Altogether there are five cases, as analysed in Table~\ref{tbl:zagier-map}, and illustrated in Figure~\ref{fig:zagier-map}.\footnote{Dubach and Muehlboeck~\cite{Dubach-Muehlboeck-2021-acm} also identified five types for windmills.}

Although five cases of Zagier map have been identified, note that the transformation rule:
\[
(x',y',z') \ee (2y - x, y, x + z - y)
\]
happens to be the same for case $2$ and case $4$.
The same rule actually applies to case $3$, which has \HOLinline{\HOLFreeVar{x}\;\HOLSymConst{=}\;\HOLFreeVar{y}}.
Thus the Zagier map can be succinctly expressed as in Definition~\ref{def:zagier-def} with only three branches.

Moreover, we can define the mind of a windmill triple as (see Table~\ref{tbl:zagier-map}):
\begin{HOLmath}
\;\;\HOLConst{mind}\;(\HOLFreeVar{x}\HOLSymConst{,}\HOLFreeVar{y}\HOLSymConst{,}\HOLFreeVar{z})\;\HOLTokenDefEquality{}\\
\;\;\;\;\HOLKeyword{if}\;\HOLFreeVar{x}\;\HOLSymConst{\HOLTokenLt{}}\;\HOLFreeVar{y}\;\HOLSymConst{\ensuremath{-}}\;\HOLFreeVar{z}\;\HOLKeyword{then}\;\HOLFreeVar{x}\;\HOLSymConst{\ensuremath{+}}\;\HOLNumLit{2}\HOLSymConst{\ensuremath{}}\HOLFreeVar{z}\\
\;\;\;\;\HOLKeyword{else}\;\HOLKeyword{if}\;\HOLFreeVar{x}\;\HOLSymConst{\HOLTokenLt{}}\;\HOLFreeVar{y}\;\HOLKeyword{then}\;\HOLNumLit{2}\HOLSymConst{\ensuremath{}}\HOLFreeVar{y}\;\HOLSymConst{\ensuremath{-}}\;\HOLFreeVar{x}\\
\;\;\;\;\HOLKeyword{else}\;\HOLFreeVar{x}
\end{HOLmath}
and verify that the mind is an invariant under the Zagier map for any triple:
\begin{HOLmath}
\HOLTokenTurnstile{}\HOLConst{mind}\;(\HOLConst{zagier}\;(\HOLFreeVar{x}\HOLSymConst{,}\HOLFreeVar{y}\HOLSymConst{,}\HOLFreeVar{z}))\;\HOLSymConst{=}\;\HOLConst{mind}\;(\HOLFreeVar{x}\HOLSymConst{,}\HOLFreeVar{y}\HOLSymConst{,}\HOLFreeVar{z})
\end{HOLmath}

Referring again to Table~\ref{tbl:zagier-map},
the windmills in \HOLinline{\HOLFreeVar{S}\;\HOLSymConst{=}\;\HOLConst{mills}\;\HOLFreeVar{n}} can be partitioned into three triple types:
\[
\begin{array}{c@{\quad}l}
S_{x < y} \ee \{(x,y,z) \in S \mid x < y\}
& \text{covering cases $1$ and $2$}\\
S_{x \ee y} \ee \{(x,y,z) \in S \mid x \ee y\}
& \text{covering case $3$}\\
S_{x > y} \ee \{(x,y,z) \in S \mid x > y\}
& \text{covering cases $4$ and $5$}\\
\end{array}
\]
Such a partition for \HOLinline{\HOLFreeVar{n}\;\HOLSymConst{=}\;\HOLNumLit{29}} is shown in Figure~\ref{fig:windmills-29-by-zagier}.
Table~\ref{tbl:zagier-map} also shows that, for triples with proper windmills:
\begin{itemize}[leftmargin=*] 
\item a triple of case $1$ maps to case $5$ and vice versa,
\item a triple of case $2$ maps to case $4$ and vice versa, and
\item a triple of case $3$ maps to itself.
\end{itemize}
Therefore the Zagier map is its own inverse for proper triples.
Combining Equation~\eqref{eqn:zagier-involute} and Equation~\eqref{eqn:mills-triple-nonzero}
for the windmills of a prime, we have:
\begin{theorem}
\label{thm:zagier-involute-mills-prime}
\script{windmill}{1475}
The Zagier map is an involution on \HOLinline{\HOLConst{mills}\;\HOLFreeVar{n}} for a prime $n$.
\begin{HOLmath}
\HOLTokenTurnstile{}\HOLConst{prime}\;\HOLFreeVar{n}\;\HOLSymConst{\HOLTokenImp{}}\;\HOLConst{zagier}\;\HOLConst{involute}\;\HOLConst{mills}\;\HOLFreeVar{n}
\end{HOLmath}
\end{theorem}

\begin{figure*}[h]  
\begin{center}
\begin{tikzpicture}[scale=0.22] 
\draw[step=1, color=white!60!black] (0,0) grid (65,17);
\windmill{8}{8}{1}{1}{7}  
\windmill{24}{8}{1}{7}{1} 
\windmill{35}{7}{3}{5}{1} 
\windmill{47}{7}{3}{1}{5} 
\windmill{58}{6}{5}{1}{1} 
\node at (4,1)  {$(1,1,7)$};
\node at (22,1) {$(1,7,1)$};
\node at (36,1) {$(3,5,1)$};
\node at (46,1) {$(3,1,5)$};
\node at (60,1) {$(5,1,1)$};
\coordinate (a) at (0.5,0.5);
\coordinate (b) at (16.5,16.5);
\node[draw,thick,color=purple,fit= (a) (b),rounded corners=.55cm,inner sep=2pt] {};
\coordinate (a) at (18,0.5);
\coordinate (b) at (40,16.5);
\node[draw,thick,color=purple,fit= (a) (b),rounded corners=.55cm,inner sep=2pt] {};
\coordinate (a) at (41.5,0.5);
\coordinate (b) at (64.5,16.5);
\node[draw,thick,color=purple,fit= (a) (b),rounded corners=.55cm,inner sep=2pt] {};
\mind[ultra thick]{8}{8}{1}
\mind{23}{7}{3}
\mind{34}{6}{5}
\mind{47}{7}{3}
\mind{58}{6}{5}
\begin{scope}[scale=10]
\coordinate (a) at (2.5,1.2); 
\coordinate (b) at (4.6,1.2); 
\draw[<->] (a) to [bend left] node[midway,below] {\HOLConst{zagier}} (b);
\coordinate (a) at (3.4,0.5); 
\coordinate (b) at (6.0,0.5); 
\draw[<->] (a) to [bend right, looseness=0.8] node[midway,above] {\HOLConst{zagier}} (b);
\coordinate (a) at (1.2,0.6); 
\coordinate (b) at (1.2,0.61); 
\draw[<->] (a) to [out=-30, in=-150, looseness=250]
          node[midway,above] {\HOLConst{zagier}} (b); 
\end{scope}
\end{tikzpicture}
\caption{Partition of windmills of \HOLinline{\HOLFreeVar{n}\;\HOLSymConst{=}\;\HOLNumLit{29}} for \HOLConst{zagier}: those with $x \ee y, x < y$, and $x > y$. Note pairing by minds.}
\Description{This figure shows a partition of windmills of n = 29 for the Zagier map, those with x = y, x < y, and x > y. Note the pairing of minds between those x < y and x > y.}
\label{fig:windmills-29-by-zagier} 
\end{center}
\end{figure*}
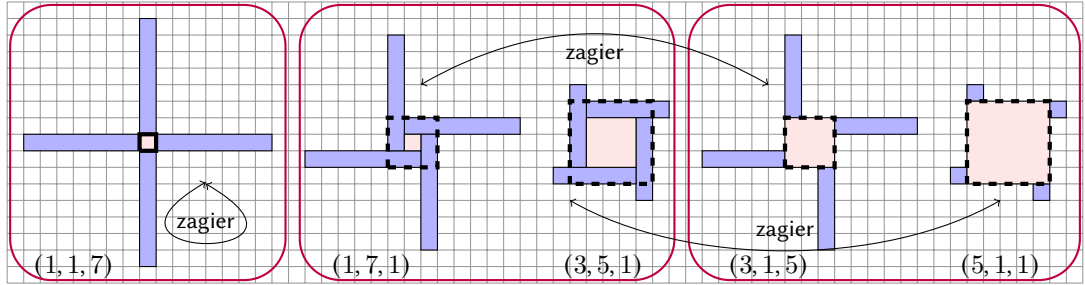

\section{Two Squares Theorem}
\label{sec:two-squares-theorem}
Now we have enough tools to formalise Fermat's two squares theorem.

\subsection{Existence of Two Squares}
\label{sec:existence}
For the Zagier map,
it is straightforward to verify, as indicated in Table~\ref{tbl:zagier-map}, that only a triple of case $3$ can map to itself:
\begin{HOLmath}
\HOLTokenTurnstile{}\HOLFreeVar{x}\;\HOLSymConst{\HOLTokenNotEqual{}}\;\HOLNumLit{0}\;\HOLSymConst{\HOLTokenImp{}}\;(\HOLConst{zagier}\;(\HOLFreeVar{x}\HOLSymConst{,}\HOLFreeVar{y}\HOLSymConst{,}\HOLFreeVar{z})\;\HOLSymConst{=}\;(\HOLFreeVar{x}\HOLSymConst{,}\HOLFreeVar{y}\HOLSymConst{,}\HOLFreeVar{z})\;\HOLSymConst{\HOLTokenEquiv{}}\;\HOLFreeVar{x}\;\HOLSymConst{=}\;\HOLFreeVar{y})
\end{HOLmath}
Hence $S_{x \ee y} \ee \HOLinline{\HOLConst{fixes}\;\HOLConst{zagier}\;(\HOLConst{mills}\;\HOLFreeVar{n})}$.
Applying Theorem~\ref{thm:mills-trivial-prime} which characterises such triples, for certain primes $S_{x \ee y}$ is a singleton:
\begin{theorem}
\label{thm:zagier-fixes-prime}
\script{twoSquares}{162}
A prime of the form \HOLinline{\HOLNumLit{4}\HOLSymConst{\ensuremath{}}\HOLFreeVar{k}\;\HOLSymConst{\ensuremath{+}}\;\HOLNumLit{1}} has only $(1,1,k)$ fixed by the Zagier map.
\begin{HOLmath}
\HOLTokenTurnstile{}\HOLConst{prime}\;\HOLFreeVar{n}\;\HOLSymConst{\HOLTokenConj{}}\;\HOLFreeVar{n}\;\ensuremath{\equiv}\;\HOLNumLit{1}\;\ensuremath{(}\ensuremath{\bmod}\;\HOLNumLit{4}\ensuremath{)}\;\HOLSymConst{\HOLTokenImp{}}\\
\;\;\;\;\;\;\;\HOLConst{fixes}\;\HOLConst{zagier}\;(\HOLConst{mills}\;\HOLFreeVar{n})\;\HOLSymConst{=}\;\HOLTokenLeftbrace{}(\HOLNumLit{1}\HOLSymConst{,}\HOLNumLit{1}\HOLSymConst{,}\HOLFreeVar{n}\;\HOLConst{\HOLConst{div}}\;\HOLNumLit{4})\HOLTokenRightbrace{}
\end{HOLmath}
\end{theorem}
\noindent
The fixed points of two involutions play crucial roles in the existence of two squares for Theorem~\ref{thm:fermat-two-squares-thm}:
\begin{theorem}[\textbf{Two Squares Existence}]
\label{thm:fermat-two-squares-exists}
\script{twoSquares}{441}
A prime of the form \HOLinline{\HOLNumLit{4}\HOLSymConst{\ensuremath{}}\HOLFreeVar{k}\;\HOLSymConst{\ensuremath{+}}\;\HOLNumLit{1}} is a sum of two squares of different parity.
\begin{HOLmath}
\HOLTokenTurnstile{}\HOLConst{prime}\;\HOLFreeVar{n}\;\HOLSymConst{\HOLTokenConj{}}\;\HOLFreeVar{n}\;\ensuremath{\equiv}\;\HOLNumLit{1}\;\ensuremath{(}\ensuremath{\bmod}\;\HOLNumLit{4}\ensuremath{)}\;\HOLSymConst{\HOLTokenImp{}}\\
\;\;\;\;\;\;\;\HOLSymConst{\HOLTokenExists{}}(\HOLBoundVar{u}\HOLSymConst{,}\HOLBoundVar{v}).\;\HOLConst{\HOLConst{odd}}\;\HOLBoundVar{u}\;\HOLSymConst{\HOLTokenConj{}}\;\HOLConst{\HOLConst{even}}\;\HOLBoundVar{v}\;\HOLSymConst{\HOLTokenConj{}}\;\HOLFreeVar{n}\;\HOLSymConst{=}\;\HOLBoundVar{u}\HOLSymConst{\ensuremath{\sp{2}}}\;\HOLSymConst{\ensuremath{+}}\;\HOLBoundVar{v}\HOLSymConst{\ensuremath{\sp{2}}}
\end{HOLmath}
\end{theorem}
\begin{proof}
A prime is not a square, so \HOLinline{\HOLConst{mills}\;\HOLFreeVar{n}} is finite by Theorem~\ref{thm:mills-finite}.
, and both Zagier and flip maps are involutions on \HOLinline{\HOLConst{mills}\;\HOLFreeVar{n}},
by Theorem~\ref{thm:zagier-involute-mills-prime} and Theorem~\ref{thm:flip-involute-mills}.
Note that Zagier map has a single fixed point by Theorem~\ref{thm:zagier-fixes-prime}.
Thus \HOLinline{\ensuremath{|}\HOLConst{fixes}\;\HOLConst{zagier}\;(\HOLConst{mills}\;\HOLFreeVar{n})\ensuremath{|}\;\HOLSymConst{=}\;\HOLNumLit{1}},
so \HOLinline{\ensuremath{|}\HOLConst{fixes}\;\HOLConst{flip}\;(\HOLConst{mills}\;\HOLFreeVar{n})\ensuremath{|}} is odd by Theorem~\ref{thm:involute-two-fixes-both-odd}.
Hence \HOLinline{\HOLConst{fixes}\;\HOLConst{flip}\;(\HOLConst{mills}\;\HOLFreeVar{n})\;\HOLSymConst{\HOLTokenNotEqual{}}\;\HOLSymConst{\HOLTokenEmpty{}}}, containing a triple $(x,y,y)$.
Thus \HOLinline{\HOLFreeVar{n}\;\HOLSymConst{=}\;\HOLConst{windmill}\;\HOLFreeVar{x}\;\HOLFreeVar{y}\;\HOLFreeVar{y}} $\ee$ \HOLinline{\HOLFreeVar{x}\HOLSymConst{\ensuremath{\sp{2}}}\;\HOLSymConst{\ensuremath{+}}\;\HOLNumLit{4}\HOLSymConst{\ensuremath{}}\HOLFreeVar{y}\HOLSymConst{\ensuremath{\sp{2}}}}.
Take \HOLinline{\HOLFreeVar{u}\;\HOLSymConst{=}\;\HOLFreeVar{x}}, and \HOLinline{\HOLFreeVar{v}\;\HOLSymConst{=}\;\HOLNumLit{2}\HOLSymConst{\ensuremath{}}\HOLFreeVar{y}}, then \HOLinline{\HOLFreeVar{n}\;\HOLSymConst{=}\;\HOLFreeVar{u}\HOLSymConst{\ensuremath{\sp{2}}}\;\HOLSymConst{\ensuremath{+}}\;\HOLFreeVar{v}\HOLSymConst{\ensuremath{\sp{2}}}}.
Evidently \HOLinline{\HOLFreeVar{v}} is even, and $u$ is odd since $n$ is odd.
\end{proof}
\noindent
Current formalisations of Zagier's proof (HOL Light~\cite{Harrison-2010}, NASA PVS~\cite{Narkawicz-2012-acm} and Coq~\cite{Dubach-Muehlboeck-2021-acm}), or its close relative Heath-Brown's proof (Mizar~\cite{Riccardi-2009} and ProofPower~\cite{Arthan-2016}), stop at just showing the existence of two squares for the primes in \text{Fermat's} Theorem~\ref{thm:fermat-two-squares-thm},
most likely because this already meets the Formalizing 100 Theorems challenge~\cite{Wiedijk-2020}.
See also related work in Section~\ref{sec:related-work}.

\subsection{Uniqueness of Two Squares}
\label{sec:uniqueness}
The uniqueness of the two squares in Fermat's Theorem~\ref{thm:fermat-two-squares-thm}
is a consequence of the following property of a prime:
\begin{theorem}[\textbf{Two Squares Uniquenss}]
\label{thm:fermat-two-squares-unique}
\script{twoSquares}{205}
If a prime $n$ can be expressed as a sum of two squares, the expression is unique up to commutativity.
\begin{HOLmath}
\HOLTokenTurnstile{}\HOLConst{prime}\;\HOLFreeVar{n}\;\HOLSymConst{\HOLTokenConj{}}\;\HOLFreeVar{n}\;\HOLSymConst{=}\;\HOLFreeVar{a}\HOLSymConst{\ensuremath{\sp{2}}}\;\HOLSymConst{\ensuremath{+}}\;\HOLFreeVar{b}\HOLSymConst{\ensuremath{\sp{2}}}\;\HOLSymConst{\HOLTokenConj{}}\;\HOLFreeVar{n}\;\HOLSymConst{=}\;\HOLFreeVar{c}\HOLSymConst{\ensuremath{\sp{2}}}\;\HOLSymConst{\ensuremath{+}}\;\HOLFreeVar{d}\HOLSymConst{\ensuremath{\sp{2}}}\;\HOLSymConst{\HOLTokenImp{}}\\
\;\;\;\;\;\;\;\HOLTokenLeftbrace{}\HOLFreeVar{a};\;\HOLFreeVar{b}\HOLTokenRightbrace{}\;\HOLSymConst{=}\;\HOLTokenLeftbrace{}\HOLFreeVar{c};\;\HOLFreeVar{d}\HOLTokenRightbrace{}
\end{HOLmath}
\end{theorem}
\noindent
The proof is purely number-theoretic, which has also been formalised by Laurent Th{\'e}ry in Coq~\cite{Thery-2004}.
Moreover, we have:
\begin{theorem}
\label{thm:mod-4-not-squares}
\script{helperTwosq}{419}
A number of the form \HOLinline{\HOLNumLit{4}\HOLSymConst{\ensuremath{}}\HOLFreeVar{k}\;\HOLSymConst{\ensuremath{+}}\;\HOLNumLit{3}} cannot be expressed as a sum of two squares.
\begin{HOLmath}
\HOLTokenTurnstile{}\HOLFreeVar{n}\;\ensuremath{\equiv}\;\HOLNumLit{3}\;\ensuremath{(}\ensuremath{\bmod}\;\HOLNumLit{4}\ensuremath{)}\;\HOLSymConst{\HOLTokenImp{}}\;\HOLSymConst{\HOLTokenForall{}}\HOLBoundVar{u}\;\HOLBoundVar{v}.\;\HOLFreeVar{n}\;\HOLSymConst{\HOLTokenNotEqual{}}\;\HOLBoundVar{u}\HOLSymConst{\ensuremath{\sp{2}}}\;\HOLSymConst{\ensuremath{+}}\;\HOLBoundVar{v}\HOLSymConst{\ensuremath{\sp{2}}}
\end{HOLmath}
\end{theorem}
\noindent
This is an elementary result from possible remainders after division by 4:
while a number, such as $u$ or $v$, may have a remainer $0, 1, 2$, or $3$,
a square, such as $u^{2}$ or $v^{2}$, can only have a remainder $0$ or $1$.
Thus the sum of such remainders can never be $3$.

Now we can complete the proof of Fermat's two squares Theorem~\ref{thm:fermat-two-squares-thm}:
\begin{HOLmath}
\HOLTokenTurnstile{}\HOLConst{prime}\;\HOLFreeVar{n}\;\HOLSymConst{\HOLTokenImp{}}\\
\;\;\;\;\;\;\;(\HOLFreeVar{n}\;\ensuremath{\equiv}\;\HOLNumLit{1}\;\ensuremath{(}\ensuremath{\bmod}\;\HOLNumLit{4}\ensuremath{)}\;\HOLSymConst{\HOLTokenEquiv{}}\\
\;\;\;\;\;\;\;\;\;\;\HOLSymConst{\HOLTokenUnique{}}(\HOLBoundVar{u}\HOLSymConst{,}\HOLBoundVar{v}).\;\HOLConst{\HOLConst{odd}}\;\HOLBoundVar{u}\;\HOLSymConst{\HOLTokenConj{}}\;\HOLConst{\HOLConst{even}}\;\HOLBoundVar{v}\;\HOLSymConst{\HOLTokenConj{}}\;\HOLFreeVar{n}\;\HOLSymConst{=}\;\HOLBoundVar{u}\HOLSymConst{\ensuremath{\sp{2}}}\;\HOLSymConst{\ensuremath{+}}\;\HOLBoundVar{v}\HOLSymConst{\ensuremath{\sp{2}}})
\end{HOLmath}
\begin{proof}
For the if part $(\Rightarrow)$,
existence is given by Theorem~\ref{thm:fermat-two-squares-exists}, and
uniqueness is provided by Theorem~\ref{thm:fermat-two-squares-unique}.
For the only-if part $(\Leftarrow)$,
an odd prime with \HOLinline{\HOLFreeVar{n}\;\ensuremath{\not\equiv}\;\HOLNumLit{1}\;\ensuremath{(}\ensuremath{\bmod}\;\HOLNumLit{4}\ensuremath{)}} cannot be a sum of two squares by Theorem~\ref{thm:mod-4-not-squares}.
\end{proof}

\section{Two Squares Algorithm}
\label{sec:algorithm}
To make Zagier's proof constructive, we need to compute that single triple fixed by flip map.

Let $n$ be a prime of the form \HOLinline{\HOLNumLit{4}\HOLSymConst{\ensuremath{}}\HOLFreeVar{k}\;\HOLSymConst{\ensuremath{+}}\;\HOLNumLit{1}}.
By Theorem~\ref{thm:zagier-fixes-prime}, the only Zagier fixed point is \HOLinline{\HOLFreeVar{u}\;\HOLSymConst{=}\;(\HOLNumLit{1}\HOLSymConst{,}\HOLNumLit{1}\HOLSymConst{,}\HOLFreeVar{k})}, meaning \HOLinline{\HOLConst{zagier}\;\HOLFreeVar{u}\;\HOLSymConst{=}\;\HOLFreeVar{u}}.
To change the triple $u$, applying \HOLConst{flip} is the obvious choice.
To keep changing the triple, \HOLConst{zagier} should be applied.
Thus by applying the composition \HOLinline{\HOLConst{zagier}\;\HOLSymConst{\HOLTokenCompose}\;\HOLConst{flip}} repeatedly from the known Zagier fixed point, there is hope that the chain will lead to the only flip fixed point.
Figure~\ref{fig:zagier-flip-29} shows that this is indeed the case for \HOLinline{\HOLFreeVar{n}\;\HOLSymConst{=}\;\HOLNumLit{29}}.

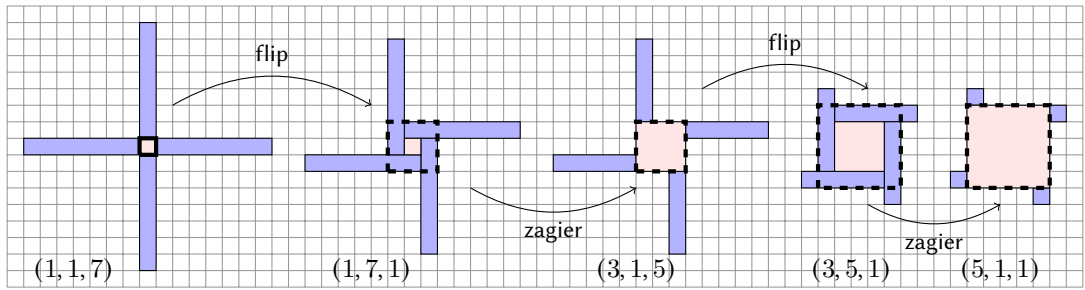
\begin{figure*}[h]  
\begin{center}
\begin{tikzpicture}[scale=0.22] 
\draw[step=1, color=white!60!black] (0,0) grid (65,17);
\windmill{8}{8}{1}{1}{7}  
\windmill{24}{8}{1}{7}{1} 
\windmill{38}{7}{3}{1}{5} 
\windmill{50}{7}{3}{5}{1} 
\windmill{58}{6}{5}{1}{1} 
\node at (4,1)  {$(1,1,7)$};
\node at (22,1) {$(1,7,1)$};
\node at (38,1) {$(3,1,5)$};
\node at (51,1) {$(3,5,1)$};
\node at (60,1) {$(5,1,1)$};
\mind[ultra thick]{8}{8}{1}
\mind{23}{7}{3}
\mind{38}{7}{3}
\mind{49}{6}{5}
\mind{58}{6}{5}
\begin{scope}[scale=10]
\coordinate (a) at (1.0,1.1); 
\coordinate (b) at (2.2,1.1); 
\draw[->] (a) to [bend left] node[midway,above] {\HOLConst{flip}} (b);
\coordinate (a) at (2.8,0.6); 
\coordinate (b) at (3.8,0.6); 
\draw[->] (a) to [bend right] node[midway,below] {\HOLConst{zagier}} (b);
\coordinate (a) at (4.2,1.2); 
\coordinate (b) at (5.2,1.2); 
\draw[->] (a) to [bend left] node[midway,above] {\HOLConst{flip}} (b);
\coordinate (a) at (5.2,0.5); 
\coordinate (b) at (6.0,0.5); 
\draw[->] (a) to [bend right] node[midway,below] {\HOLConst{zagier}} (b);
\end{scope}
\end{tikzpicture}
\caption{The iteration chain of \HOLinline{\HOLFreeVar{n}\;\HOLSymConst{=}\;\HOLNumLit{29}} by the composition \HOLinline{\HOLConst{zagier}\;\HOLSymConst{\HOLTokenCompose}\;\HOLConst{flip}}, from Zagier fix to flip fix.}
\Description{This figure show the iteration chain of n = 29, by the composition of first flip then Zagier. The chain starts from Zagier fix, ends in flip fix.}
\label{fig:zagier-flip-29} 
\end{center}
\end{figure*}
In terms of windmills, the flip map keeps the central square, but flips the arms of rectangles from $y$-by-$z$ to $z$-by-$y$. This generally changes the mind of the windmill. The Zagier map keeps the mind, but changes the central square.
Similar to the mind being an invariant of the Zagier map,
the absolute difference $\left|y - z\right|$ is an invariant of the flip map.
If the Zagier map can reduce this difference, the successive iterations of \HOLinline{\HOLConst{zagier}\;\HOLSymConst{\HOLTokenCompose}\;\HOLConst{flip}} will be able to locate the flip fixed point.

\subsection{Flip Fix Search}
\label{sec:flip-fix-search}
To find the fixed point of the flip map, we can experiment with this pseudo-code:

\bigskip
\fbox{\begin{minipage}{0.36\textwidth}
\begin{list}{$\circ$}{}
\item \emph{Input}: a number \HOLinline{\HOLFreeVar{n}\;\HOLSymConst{=}\;\HOLNumLit{4}\HOLSymConst{\ensuremath{}}\HOLFreeVar{k}\;\HOLSymConst{\ensuremath{+}}\;\HOLNumLit{1}}.
\item \emph{Output}: a triple fixed by the flip map.
\item \emph{Method}:
\item start with \HOLinline{\HOLFreeVar{u}\;\HOLSymConst{=}\;(\HOLNumLit{1}\HOLSymConst{,}\HOLNumLit{1}\HOLSymConst{,}\HOLFreeVar{k})}, the Zagier fix.
\item while ($u$ is not a flip fix) :
\item \qquad $u \leftarrow \HOLinline{(\HOLConst{zagier}\;\HOLSymConst{\HOLTokenCompose}\;\HOLConst{flip})\;\HOLFreeVar{u}}$
\item end while.
\end{list}
\end{minipage}}

\bigskip
\noindent
In an HOL4 interactive session, this pseudo-code can be implemented directly as:\footnote{This pseudo-code can be implemented directly in any programming language that supports while-loops and tuples.}
\begin{definition}
\label{def:two-sq-def}
Computing the flip fixed point of \HOLinline{\HOLFreeVar{n}\;\HOLSymConst{=}\;\HOLNumLit{4}\HOLSymConst{\ensuremath{}}\HOLFreeVar{k}\;\HOLSymConst{\ensuremath{+}}\;\HOLNumLit{1}} using a \HOLConst{WHILE} loop.
\begin{HOLmath}
\;\;\HOLConst{two_sq}\;\HOLFreeVar{n}\;\HOLTokenDefEquality{}\\
\;\;\;\;\HOLConst{WHILE}\;((\HOLSymConst{\HOLTokenNeg{}})\;\HOLSymConst{\HOLTokenCompose}\;\HOLConst{found})\;(\HOLConst{zagier}\;\HOLSymConst{\HOLTokenCompose}\;\HOLConst{flip})\;(\HOLNumLit{1}\HOLSymConst{,}\HOLNumLit{1}\HOLSymConst{,}\HOLFreeVar{n}\;\HOLConst{\HOLConst{div}}\;\HOLNumLit{4}),\\
\quad\textrm{where}
\;\;\HOLConst{found}\;(\HOLFreeVar{x}\HOLSymConst{,}\HOLFreeVar{y}\HOLSymConst{,}\HOLFreeVar{z})\;\HOLTokenDefEquality{}\;\HOLFreeVar{y}\;\HOLSymConst{=}\;\HOLFreeVar{z}
\end{HOLmath}
\end{definition}
\noindent
This simple while-loop may or may not terminate. We shall take up this issue in Section~\ref{sec:termination}.
For primes of the form \HOLinline{\HOLNumLit{4}\HOLSymConst{\ensuremath{}}\HOLFreeVar{k}\;\HOLSymConst{\ensuremath{+}}\;\HOLNumLit{1}}, it terminates and seems to work.
To prove its correctness, we shall develop a theory of permutation iteration, then apply the theory to this algorithm.

\section{Permutation Orbits}
\label{sec:orbits}
In general, the composition of two involutions is no longer an involution, but just a permutation.
Let $\varphi \colon S \rightarrow S$ be a permutation,
a bijection on the set $S$, denoted by \HOLinline{\HOLFreeVar{\varphi}\;\HOLConst{\HOLConst{permutes}}\;\HOLFreeVar{S}}.
For an element \HOLinline{\HOLFreeVar{x}\;\HOLSymConst{\HOLTokenIn{}}\;\HOLFreeVar{S}} 
the iteration sequence $\varphi(x)$, \HOLinline{\HOLFreeVar{\varphi}\ensuremath{\sp{\HOLNumLit{2}}(\HOLFreeVar{x})}}, \HOLinline{\HOLFreeVar{\varphi}\ensuremath{\sp{\HOLNumLit{3}}(\HOLFreeVar{x})}}, \etc, form its \emph{orbit}.
The smallest positive index $n$ such that \HOLinline{\HOLFreeVar{\varphi}\ensuremath{\sp{\HOLFreeVar{n}}(\HOLFreeVar{x})}\;\HOLSymConst{=}\;\HOLFreeVar{x}} is called the \emph{period} of $x$ under $\varphi$.
If such a positive index does not exist, the period is defined to be $0$.
In HOL4, the definition makes use of \HOLConst{OLEAST}, the optional \HOLConst{LEAST} operator:
\begin{definition}
\label{def:period-def}
The period of function iteration of an element is the least nonzero index for the element iterate to wrap around, otherwise zero.
\begin{HOLmath}
\;\;\HOLConst{\HOLConst{period}}\;\HOLFreeVar{\varphi}\;\HOLFreeVar{x}\;\HOLTokenDefEquality{}\\
\;\;\;\;\HOLKeyword{case}\;\HOLConst{OLEAST}\;\HOLBoundVar{k}.\;\HOLNumLit{0}\;\HOLSymConst{\HOLTokenLt{}}\;\HOLBoundVar{k}\;\HOLSymConst{\HOLTokenConj{}}\;\HOLFreeVar{\varphi}\ensuremath{\sp{\HOLBoundVar{k}}(\HOLFreeVar{x})}\;\HOLSymConst{=}\;\HOLFreeVar{x}\;\HOLKeyword{of}\\
\;\;\;\;\HOLTokenBar{}\;\HOLConst{\HOLConst{none}}\;\ensuremath{\triangleright}\;\HOLNumLit{0}\\
\;\;\;\;\HOLTokenBar{}\;\HOLConst{\HOLConst{some}}\;\HOLBoundVar{k}\;\ensuremath{\triangleright}\;\HOLBoundVar{k}
\end{HOLmath}
\end{definition}
\noindent
When the set $S$ is finite, the iterates cannot be always distinct.
Thus the permutation orbit of any $x \in S$ is finite, with a nonzero period,
denoted by \HOLinline{\HOLFreeVar{p}\;\HOLSymConst{=}\;\HOLConst{period}\;\HOLFreeVar{\varphi}\;\HOLFreeVar{x}}:
\begin{HOLmath}
\HOLTokenTurnstile{}\HOLConst{\HOLConst{finite}}\;\HOLFreeVar{S}\;\HOLSymConst{\HOLTokenConj{}}\;\HOLFreeVar{\varphi}\;\HOLConst{\HOLConst{permutes}}\;\HOLFreeVar{S}\;\HOLSymConst{\HOLTokenConj{}}\;\HOLFreeVar{x}\;\HOLSymConst{\HOLTokenIn{}}\;\HOLFreeVar{S}\;\HOLSymConst{\HOLTokenImp{}}\\
\;\;\;\;\;\;\;\HOLSymConst{\HOLTokenExists{}}\HOLBoundVar{p}.\;\HOLNumLit{0}\;\HOLSymConst{\HOLTokenLt{}}\;\HOLBoundVar{p}\;\HOLSymConst{\HOLTokenConj{}}\;\HOLBoundVar{p}\;\HOLSymConst{=}\;\HOLConst{\HOLConst{period}}\;\HOLFreeVar{\varphi}\;\HOLFreeVar{x}
\end{HOLmath}
and by definition the period is minimal, which means that there is no wrap around for element iterates when the index is less than the period:
\begin{HOLmath}
\HOLTokenTurnstile{}\HOLNumLit{0}\;\HOLSymConst{\HOLTokenLt{}}\;\HOLFreeVar{j}\;\HOLSymConst{\HOLTokenConj{}}\;\HOLFreeVar{j}\;\HOLSymConst{\HOLTokenLt{}}\;\HOLConst{\HOLConst{period}}\;\HOLFreeVar{\varphi}\;\HOLFreeVar{x}\;\HOLSymConst{\HOLTokenImp{}}\;\HOLFreeVar{\varphi}\ensuremath{\sp{\HOLFreeVar{j}}(\HOLFreeVar{x})}\;\HOLSymConst{\HOLTokenNotEqual{}}\;\HOLFreeVar{x}
\end{HOLmath}
This implies a criterion for an exponent index to be divisible by period:
\begin{theorem}
\label{thm:iterate-period-mod}
\script{iteration}{653}
For a nonzero period $p$ of $x$, $x$ is fixed by the $k$-th iterate of $\varphi$ if and only if $k$ is a multiple of period~$p$.
\begin{HOLmath}
\HOLTokenTurnstile{}\HOLNumLit{0}\;\HOLSymConst{\HOLTokenLt{}}\;\HOLFreeVar{p}\;\HOLSymConst{\HOLTokenConj{}}\;\HOLFreeVar{p}\;\HOLSymConst{=}\;\HOLConst{\HOLConst{period}}\;\HOLFreeVar{\varphi}\;\HOLFreeVar{x}\;\HOLSymConst{\HOLTokenImp{}}\\
\;\;\;\;\;\;\;(\HOLFreeVar{\varphi}\ensuremath{\sp{\HOLFreeVar{k}}(\HOLFreeVar{x})}\;\HOLSymConst{=}\;\HOLFreeVar{x}\;\HOLSymConst{\HOLTokenEquiv{}}\;\HOLFreeVar{k}\;\ensuremath{\equiv}\;\HOLNumLit{0}\;\ensuremath{(}\ensuremath{\bmod}\;\HOLFreeVar{p}\ensuremath{)})
\end{HOLmath}
\end{theorem}
\noindent
Moreover, the period is the same for all iterates in the same orbit:
\begin{HOLmath}
\HOLTokenTurnstile{}\HOLConst{\HOLConst{finite}}\;\HOLFreeVar{S}\;\HOLSymConst{\HOLTokenConj{}}\;\HOLFreeVar{\varphi}\;\HOLConst{\HOLConst{permutes}}\;\HOLFreeVar{S}\;\HOLSymConst{\HOLTokenConj{}}\;\HOLFreeVar{x}\;\HOLSymConst{\HOLTokenIn{}}\;\HOLFreeVar{S}\;\HOLSymConst{\HOLTokenConj{}}\;\HOLFreeVar{y}\;\HOLSymConst{=}\;\HOLFreeVar{\varphi}\ensuremath{\sp{\HOLFreeVar{j}}(\HOLFreeVar{x})}\;\HOLSymConst{\HOLTokenImp{}}\\
\;\;\;\;\;\;\;\HOLConst{\HOLConst{period}}\;\HOLFreeVar{\varphi}\;\HOLFreeVar{y}\;\HOLSymConst{=}\;\HOLConst{\HOLConst{period}}\;\HOLFreeVar{\varphi}\;\HOLFreeVar{x}
\end{HOLmath}

\subsection{Involution Composition}
\label{sec:involution-composition}
When the permutation \HOLinline{\HOLFreeVar{\varphi}\;\HOLSymConst{=}\;\HOLFreeVar{f}\;\HOLSymConst{\HOLTokenCompose}\;\HOLFreeVar{g}}, a composition of two involutions $f$ and \HOLinline{\HOLFreeVar{g}}, we shall investigate whether their fixed points are connected by a chain of composition iterations.
Note the following pattern of function application:
\begin{equation*}
\label{eqn:function-assoc}
\begin{split}
f\ \circ\ (\HOLinline{\HOLFreeVar{g}}\ \circ\ f)\ \circ\ (\HOLinline{\HOLFreeVar{g}}\ \circ\ f)\ \circ\ (\HOLinline{\HOLFreeVar{g}}\ \circ\ f)\\
\ee (f\ \circ\ \HOLinline{\HOLFreeVar{g}})\ \circ\ (f\ \circ\ \HOLinline{\HOLFreeVar{g}})\ \circ\ (f\ \circ\ \HOLinline{\HOLFreeVar{g}})\ \circ\ f
\end{split}
\end{equation*}
by associativity.
Also, $(\HOLinline{\HOLFreeVar{f}\;\HOLSymConst{\HOLTokenCompose}\;\HOLFreeVar{g}})^{-1} \ee \HOLinline{\HOLFreeVar{g}}^{-1}\ \circ\ f^{-1} \ee \HOLinline{\HOLFreeVar{g}\;\HOLSymConst{\HOLTokenCompose}\;\HOLFreeVar{f}}$ for involutions,
so inverse is just reversal of application order in this case.
Let \HOLinline{\HOLFreeVar{p}\;\HOLSymConst{=}\;\HOLConst{\HOLConst{period}}\;(\HOLFreeVar{f}\;\HOLSymConst{\HOLTokenCompose}\;\HOLFreeVar{g})\;\HOLFreeVar{x}} for \HOLinline{\HOLFreeVar{x}\;\HOLSymConst{\HOLTokenIn{}}\;\HOLFreeVar{S}}.
With these notations, we can establish some basic results:
\begin{theorem}
\label{thm:involute-period-1}
\script{iterateCompose}{558}
When $f$ fixes $x$, the period for $x$ is $1$ if and only if \HOLinline{\HOLFreeVar{g}} also fixes $x$.
\begin{HOLmath}
\HOLTokenTurnstile{}\HOLFreeVar{f}\;\HOLConst{involute}\;\HOLFreeVar{S}\;\HOLSymConst{\HOLTokenConj{}}\;\HOLFreeVar{g}\;\HOLConst{involute}\;\HOLFreeVar{S}\;\HOLSymConst{\HOLTokenConj{}}\;\HOLFreeVar{x}\;\HOLSymConst{\HOLTokenIn{}}\;\HOLConst{fixes}\;\HOLFreeVar{f}\;\HOLFreeVar{S}\;\HOLSymConst{\HOLTokenConj{}}\\
\;\;\;\;\;\HOLFreeVar{p}\;\HOLSymConst{=}\;\HOLConst{\HOLConst{period}}\;(\HOLFreeVar{f}\;\HOLSymConst{\HOLTokenCompose}\;\HOLFreeVar{g})\;\HOLFreeVar{x}\;\HOLSymConst{\HOLTokenImp{}}\\
\;\;\;\;\;\;\;(\HOLFreeVar{p}\;\HOLSymConst{=}\;\HOLNumLit{1}\;\HOLSymConst{\HOLTokenEquiv{}}\;\HOLFreeVar{x}\;\HOLSymConst{\HOLTokenIn{}}\;\HOLConst{fixes}\;\HOLFreeVar{g}\;\HOLFreeVar{S})
\end{HOLmath}
\end{theorem}
\noindent
Pick an element $x$ in the set $S$.
For involutions, an iterate of (\HOLinline{\HOLFreeVar{f}\;\HOLSymConst{\HOLTokenCompose}\;\HOLFreeVar{g}}) can be equal to another iterate of (\HOLinline{\HOLFreeVar{g}\;\HOLSymConst{\HOLTokenCompose}\;\HOLFreeVar{f}}):
\begin{theorem}
\label{thm:involute-mod-period}
\script{iterateCompose}{401}
The $i$-th iterate of (\HOLinline{\HOLFreeVar{f}\;\HOLSymConst{\HOLTokenCompose}\;\HOLFreeVar{g}}) equals the $j$-th iterate of (\HOLinline{\HOLFreeVar{g}\;\HOLSymConst{\HOLTokenCompose}\;\HOLFreeVar{f}}) if and only if (\HOLinline{\HOLFreeVar{i}\;\HOLSymConst{\ensuremath{+}}\;\HOLFreeVar{j}}) is a multiple of period $p$.
\begin{HOLmath}
\HOLTokenTurnstile{}\HOLConst{\HOLConst{finite}}\;\HOLFreeVar{S}\;\HOLSymConst{\HOLTokenConj{}}\;\HOLFreeVar{f}\;\HOLConst{involute}\;\HOLFreeVar{S}\;\HOLSymConst{\HOLTokenConj{}}\;\HOLFreeVar{g}\;\HOLConst{involute}\;\HOLFreeVar{S}\;\HOLSymConst{\HOLTokenConj{}}\;\HOLFreeVar{x}\;\HOLSymConst{\HOLTokenIn{}}\;\HOLFreeVar{S}\;\HOLSymConst{\HOLTokenConj{}}\\
\;\;\;\;\;\HOLFreeVar{p}\;\HOLSymConst{=}\;\HOLConst{\HOLConst{period}}\;(\HOLFreeVar{f}\;\HOLSymConst{\HOLTokenCompose}\;\HOLFreeVar{g})\;\HOLFreeVar{x}\;\HOLSymConst{\HOLTokenImp{}}\\
\;\;\;\;\;\;\;((\HOLFreeVar{f}\;\HOLSymConst{\HOLTokenCompose}\;\HOLFreeVar{g})\ensuremath{\sp{\HOLFreeVar{i}}(\HOLFreeVar{x})}\;\HOLSymConst{=}\;(\HOLFreeVar{g}\;\HOLSymConst{\HOLTokenCompose}\;\HOLFreeVar{f})\ensuremath{\sp{\HOLFreeVar{j}}(\HOLFreeVar{x})}\;\HOLSymConst{\HOLTokenEquiv{}}\;\HOLFreeVar{i}\;\HOLSymConst{\ensuremath{+}}\;\HOLFreeVar{j}\;\ensuremath{\equiv}\;\HOLNumLit{0}\;\ensuremath{(}\ensuremath{\bmod}\;\HOLFreeVar{p}\ensuremath{)})
\end{HOLmath}
\end{theorem}
\noindent
When $f$ fixes point $x$,
the iterates \HOLinline{(\HOLFreeVar{f}\;\HOLSymConst{\HOLTokenCompose}\;\HOLFreeVar{g})\ensuremath{\sp{\HOLFreeVar{i}}(\HOLFreeVar{x})}} and \HOLinline{(\HOLFreeVar{f}\;\HOLSymConst{\HOLTokenCompose}\;\HOLFreeVar{g})\ensuremath{\sp{\HOLFreeVar{j}}(\HOLFreeVar{x})}} are related when the sum (\HOLinline{\HOLFreeVar{i}\;\HOLSymConst{\ensuremath{+}}\;\HOLFreeVar{j}}) is special:
\begin{theorem}
\label{thm:involute-two-fix-orbit-1}
\script{iterateCompose}{583}
When $f$ fixes $x$, the $i$-th and $j$-th iterate of (\HOLinline{\HOLFreeVar{f}\;\HOLSymConst{\HOLTokenCompose}\;\HOLFreeVar{g}}) differ by one $f$ application if and only if (\HOLinline{\HOLFreeVar{i}\;\HOLSymConst{\ensuremath{+}}\;\HOLFreeVar{j}}) is a multiple of period $p$.
\begin{HOLmath}
\HOLTokenTurnstile{}\HOLConst{\HOLConst{finite}}\;\HOLFreeVar{S}\;\HOLSymConst{\HOLTokenConj{}}\;\HOLFreeVar{f}\;\HOLConst{involute}\;\HOLFreeVar{S}\;\HOLSymConst{\HOLTokenConj{}}\;\HOLFreeVar{g}\;\HOLConst{involute}\;\HOLFreeVar{S}\;\HOLSymConst{\HOLTokenConj{}}\\
\;\;\;\;\;\HOLFreeVar{x}\;\HOLSymConst{\HOLTokenIn{}}\;\HOLConst{fixes}\;\HOLFreeVar{f}\;\HOLFreeVar{S}\;\HOLSymConst{\HOLTokenConj{}}\;\HOLFreeVar{p}\;\HOLSymConst{=}\;\HOLConst{\HOLConst{period}}\;(\HOLFreeVar{f}\;\HOLSymConst{\HOLTokenCompose}\;\HOLFreeVar{g})\;\HOLFreeVar{x}\;\HOLSymConst{\HOLTokenImp{}}\\
\;\;\;\;\;\;\;((\HOLFreeVar{f}\;\HOLSymConst{\HOLTokenCompose}\;\HOLFreeVar{g})\ensuremath{\sp{\HOLFreeVar{i}}(\HOLFreeVar{x})}\;\HOLSymConst{=}\;\HOLFreeVar{f}\;((\HOLFreeVar{f}\;\HOLSymConst{\HOLTokenCompose}\;\HOLFreeVar{g})\ensuremath{\sp{\HOLFreeVar{j}}(\HOLFreeVar{x})})\;\HOLSymConst{\HOLTokenEquiv{}}\;\HOLFreeVar{i}\;\HOLSymConst{\ensuremath{+}}\;\HOLFreeVar{j}\;\ensuremath{\equiv}\;\HOLNumLit{0}\;\ensuremath{(}\ensuremath{\bmod}\;\HOLFreeVar{p}\ensuremath{)})
\end{HOLmath}
\end{theorem}
\noindent
There is a related result, with a similar proof:
\begin{theorem}
\label{thm:involute-two-fix-orbit-2}
\script{iterateCompose}{689}
When $f$ fixes $x$, the $i$-th and $j$-th iterate of (\HOLinline{\HOLFreeVar{f}\;\HOLSymConst{\HOLTokenCompose}\;\HOLFreeVar{g}}) differ by one \HOLinline{\HOLFreeVar{g}} application if and only if (\HOLinline{\HOLFreeVar{i}\;\HOLSymConst{\ensuremath{+}}\;\HOLFreeVar{j}\;\HOLSymConst{\ensuremath{+}}\;\HOLNumLit{1}}) is a multiple of period $p$.
\begin{HOLmath}
\HOLTokenTurnstile{}\HOLConst{\HOLConst{finite}}\;\HOLFreeVar{S}\;\HOLSymConst{\HOLTokenConj{}}\;\HOLFreeVar{f}\;\HOLConst{involute}\;\HOLFreeVar{S}\;\HOLSymConst{\HOLTokenConj{}}\;\HOLFreeVar{g}\;\HOLConst{involute}\;\HOLFreeVar{S}\;\HOLSymConst{\HOLTokenConj{}}\\
\;\;\;\;\;\HOLFreeVar{x}\;\HOLSymConst{\HOLTokenIn{}}\;\HOLConst{fixes}\;\HOLFreeVar{f}\;\HOLFreeVar{S}\;\HOLSymConst{\HOLTokenConj{}}\;\HOLFreeVar{p}\;\HOLSymConst{=}\;\HOLConst{\HOLConst{period}}\;(\HOLFreeVar{f}\;\HOLSymConst{\HOLTokenCompose}\;\HOLFreeVar{g})\;\HOLFreeVar{x}\;\HOLSymConst{\HOLTokenImp{}}\\
\;\;\;\;\;\;\;((\HOLFreeVar{f}\;\HOLSymConst{\HOLTokenCompose}\;\HOLFreeVar{g})\ensuremath{\sp{\HOLFreeVar{i}}(\HOLFreeVar{x})}\;\HOLSymConst{=}\;\HOLFreeVar{g}\;((\HOLFreeVar{f}\;\HOLSymConst{\HOLTokenCompose}\;\HOLFreeVar{g})\ensuremath{\sp{\HOLFreeVar{j}}(\HOLFreeVar{x})})\;\HOLSymConst{\HOLTokenEquiv{}}\\
\;\;\;\;\;\;\;\;\;\;\HOLFreeVar{i}\;\HOLSymConst{\ensuremath{+}}\;\HOLFreeVar{j}\;\HOLSymConst{\ensuremath{+}}\;\HOLNumLit{1}\;\ensuremath{\equiv}\;\HOLNumLit{0}\;\ensuremath{(}\ensuremath{\bmod}\;\HOLFreeVar{p}\ensuremath{)})
\end{HOLmath}
\end{theorem}
\noindent
These theorems are useful in the study of iteration orbits starting from fixed points.

\subsection{Period Parity}
\label{sec:period-parity}
Given a finite set $S$, and an element $x \in S$, the iterates \HOLinline{(\HOLFreeVar{f}\;\HOLSymConst{\HOLTokenCompose}\;\HOLFreeVar{g})\ensuremath{\sp{\HOLFreeVar{j}}(\HOLFreeVar{x})}} form an orbit, with length equal to the period \HOLinline{\HOLFreeVar{p}\;\HOLSymConst{=}\;\HOLConst{\HOLConst{period}}\;(\HOLFreeVar{f}\;\HOLSymConst{\HOLTokenCompose}\;\HOLFreeVar{g})\;\HOLFreeVar{x}}. 
Figure~\ref{fig:orbits-even-odd} shows two orbits, one with an even period, the other with an odd period. 

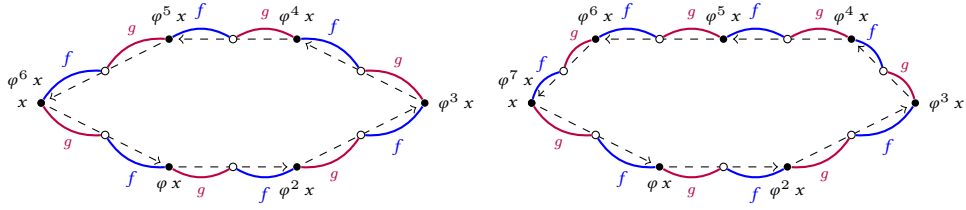
\begin{figure*}[h]
\centering
\begin{tikzpicture}[scale=1.7, 
    ele/.style={fill=black,circle,minimum width=.8pt,inner sep=1pt},
    every fit/.style={ellipse,draw,inner sep=5pt}]
  \node[ele,label=left:{\tiny{$x$}}] (a) at (0,0.5) {};    
  \node[ele,label=below:{\tiny{$\varphi\ x$}}] (b) at (1,0) {};
  \node[ele,label=below:{\tiny{$\varphi^{2}\ x$}}] (c) at (2,0) {};    
  \node[ele,label=right:{\tiny{$\varphi^{3}\ x$}}] (d) at (3,0.5) {};
  \node[ele,label=above:{\tiny{$\varphi^{4}\ x$}}] (e) at (2,1) {};
  \node[ele,label=above:{\tiny{$\varphi^{5}\ x$}}] (f) at (1,1) {};
  \node[ele,label=above:{\tiny{$\varphi^{6}\ x\qquad$}}] (g) at (0,0.5) {};
  \draw[->,dashed,shorten <=2pt,shorten >=2] (a) -- (b);
  \draw[->,dashed,shorten <=2pt,shorten >=2] (b) -- (c);
  \draw[->,dashed,shorten <=2pt,shorten >=2] (c) -- (d);
  \draw[->,dashed,shorten <=2pt,shorten >=2] (d) -- (e);
  \draw[->,dashed,shorten <=2pt,shorten >=2] (e) -- (f);
  \draw[->,dashed,shorten <=2pt,shorten >=2] (f) -- (g);
  \node[ele,draw,fill=white] (ab) at ($(a)!0.5!(b)$) {};
  \node[ele,draw,fill=white] (bc) at ($(b)!0.5!(c)$) {};
  \node[ele,draw,fill=white] (cd) at ($(c)!0.5!(d)$) {};
  \node[ele,draw,fill=white] (de) at ($(d)!0.5!(e)$) {};
  \node[ele,draw,fill=white] (ef) at ($(e)!0.5!(f)$) {};
  \node[ele,draw,fill=white] (fg) at ($(f)!0.5!(g)$) {};
  \draw[thick,color=purple] (a) to [bend right] node[midway,below] {\tiny{\HOLinline{\HOLFreeVar{g}}}} (ab);
  \draw[thick,color=blue] (ab) to [bend right] node[midway,below] {\tiny{$f$}} (b);
  \draw[thick,color=purple] (b) to [bend right] node[midway,below] {\tiny{\HOLinline{\HOLFreeVar{g}}}} (bc);
  \draw[thick,color=blue] (bc) to [bend right] node[midway,below] {\tiny{$f$}} (c);
  \draw[thick,color=purple] (c) to [bend right] node[midway,below] {\tiny{\HOLinline{\HOLFreeVar{g}}}} (cd);
  \draw[thick,color=blue] (cd) to [bend right] node[midway,below] {\tiny{$f$}} (d);
  \draw[thick,color=purple] (d) to [bend right] node[midway,above] {\tiny{\HOLinline{\HOLFreeVar{g}}}} (de);
  \draw[thick,color=blue] (de) to [bend right] node[midway,above] {\tiny{$f$}} (e);
  \draw[thick,color=purple] (e) to [bend right] node[midway,above] {\tiny{\HOLinline{\HOLFreeVar{g}}}} (ef);
  \draw[thick,color=blue] (ef) to [bend right] node[midway,above] {\tiny{$f$}} (f);
  \draw[thick,color=purple] (f) to [bend right] node[midway,above] {\tiny{\HOLinline{\HOLFreeVar{g}}}} (fg);
  \draw[thick,color=blue] (fg) to [bend right] node[midway,above] {\tiny{$f$}} (g);
\end{tikzpicture}
\begin{tikzpicture}[scale=1.7, 
    ele/.style={fill=black,circle,minimum width=.8pt,inner sep=1pt},
    every fit/.style={ellipse,draw,inner sep=5pt}]
  \node[ele,label=left:{\tiny{$x$}}] (a) at (0,0.5) {};    
  \node[ele,label=below:{\tiny{$\varphi\ x$}}] (b) at (1,0) {};
  \node[ele,label=below:{\tiny{$\varphi^{2}\ x$}}] (c) at (2,0) {};    
  \node[ele,label=right:{\tiny{$\varphi^{3}\ x$}}] (d) at (3,0.5) {};
  \node[ele,label=above:{\tiny{$\varphi^{4}\ x$}}] (e) at (2.5,1) {};
  \node[ele,label=above:{\tiny{$\varphi^{5}\ x$}}] (f) at (1.5,1) {};
  \node[ele,label=above:{\tiny{$\varphi^{6}\ x$}}] (g) at (0.5,1) {};
  \node[ele,label=above:{\tiny{$\varphi^{7}\ x\qquad$}}] (h) at (0,0.5) {};
  \draw[->,dashed,shorten <=2pt,shorten >=2] (a) -- (b);
  \draw[->,dashed,shorten <=2pt,shorten >=2] (b) -- (c);
  \draw[->,dashed,shorten <=2pt,shorten >=2] (c) -- (d);
  \draw[->,dashed,shorten <=2pt,shorten >=2] (d) -- (e);
  \draw[->,dashed,shorten <=2pt,shorten >=2] (e) -- (f);
  \draw[->,dashed,shorten <=2pt,shorten >=2] (f) -- (g);
  \draw[->,dashed,shorten <=2pt,shorten >=2] (g) -- (h);
  \node[ele,draw,fill=white] (ab) at ($(a)!0.5!(b)$) {};
  \node[ele,draw,fill=white] (bc) at ($(b)!0.5!(c)$) {};
  \node[ele,draw,fill=white] (cd) at ($(c)!0.5!(d)$) {};
  \node[ele,draw,fill=white] (de) at ($(d)!0.5!(e)$) {};
  \node[ele,draw,fill=white] (ef) at ($(e)!0.5!(f)$) {};
  \node[ele,draw,fill=white] (fg) at ($(f)!0.5!(g)$) {};
  \node[ele,draw,fill=white] (gh) at ($(g)!0.5!(h)$) {};
  \draw[thick,color=purple] (a) to [bend right] node[midway,below] {\tiny{\HOLinline{\HOLFreeVar{g}}}} (ab);
  \draw[thick,color=blue] (ab) to [bend right] node[midway,below] {\tiny{$f$}} (b);
  \draw[thick,color=purple] (b) to [bend right] node[midway,below] {\tiny{\HOLinline{\HOLFreeVar{g}}}} (bc);
  \draw[thick,color=blue] (bc) to [bend right] node[midway,below] {\tiny{$f$}} (c);
  \draw[thick,color=purple] (c) to [bend right] node[midway,below] {\tiny{\HOLinline{\HOLFreeVar{g}}}} (cd);
  \draw[thick,color=blue] (cd) to [bend right] node[midway,below] {\tiny{$f$}} (d);
  \draw[thick,color=purple] (d) to [bend right] node[midway,above] {\tiny{\HOLinline{\HOLFreeVar{g}}}} (de);
  \draw[thick,color=blue] (de) to [bend right] node[midway,above] {\tiny{$f$}} (e);
  \draw[thick,color=purple] (e) to [bend right] node[midway,above] {\tiny{\HOLinline{\HOLFreeVar{g}}}} (ef);
  \draw[thick,color=blue] (ef) to [bend right] node[midway,above] {\tiny{$f$}} (f);
  \draw[thick,color=purple] (f) to [bend right] node[midway,above] {\tiny{\HOLinline{\HOLFreeVar{g}}}} (fg);
  \draw[thick,color=blue] (fg) to [bend right] node[midway,above] {\tiny{$f$}} (g);
  \draw[thick,color=purple] (g) to [bend right] node[midway,above] {\tiny{\HOLinline{\HOLFreeVar{g}}}} (gh);
  \draw[thick,color=blue] (gh) to [bend right] node[midway,above] {\tiny{$f$}} (h);
\end{tikzpicture}
\caption{Orbits of \HOLinline{\HOLFreeVar{\varphi}\;\HOLSymConst{=}\;\HOLFreeVar{f}\;\HOLSymConst{\HOLTokenCompose}\;\HOLFreeVar{g}} for point $x$. Left one has even period $6$, right one has odd period $7$.}
\Description{This figure shows the orbits for point x of the composition: first g than f. Left one has even period 6, right one has odd period 7.}
\label{fig:orbits-even-odd} 
\end{figure*}

In the figure, black dots indicate iterates of \HOLinline{\HOLFreeVar{\varphi}\;\HOLSymConst{=}\;\HOLFreeVar{f}\;\HOLSymConst{\HOLTokenCompose}\;\HOLFreeVar{g}}, in dashes, and white dots indicate the intermediates, with \HOLinline{\HOLFreeVar{g}} first, then $f$, through the arcs.
Since $f$ and \HOLinline{\HOLFreeVar{g}} are involutions, the arcs can go both ways: forward or backward.

Let $\alpha$ denote a fixed point of $f$, and $\beta$ denote a fixed point of \HOLinline{\HOLFreeVar{g}},
\ie, \HOLinline{\HOLFreeVar{f}\;\HOLFreeVar{\alpha}\;\HOLSymConst{=}\;\HOLFreeVar{\alpha}}, and \HOLinline{\HOLFreeVar{g}\;\HOLFreeVar{\beta}\;\HOLSymConst{=}\;\HOLFreeVar{\beta}}.
We shall look at how these fixed points are related, which is crucial in the correctness proof of our algorithm (see Definition~\ref{def:two-sq-def}).

\subsection{Fixed Point Period Even}
\label{sec:fixed-point-period-even}
Consider an orbit with even period starting with $\alpha$, a fixed point of $f$.
Figure~\ref{fig:orbit-fix-even} shows one on the left, and its real picture on the right.

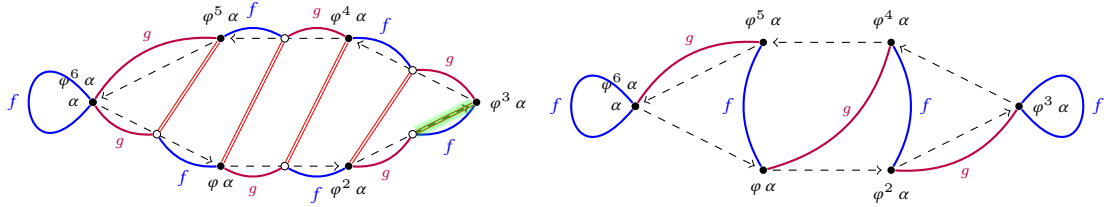
\begin{figure*}[h]
\centering
\begin{tikzpicture}[scale=1.7, 
    ele/.style={fill=black,circle,minimum width=.8pt,inner sep=1pt},
    every fit/.style={ellipse,draw,inner sep=5pt}]
  \node[ele,label=left:{\tiny{$\alpha$}}] (a) at (0,0.5) {};    
  \node[ele,label=below:{\tiny{$\varphi\ \alpha$}}] (b) at (1,0) {};
  \node[ele,label=below:{\tiny{$\varphi^{2}\ \alpha$}}] (c) at (2,0) {};    
  \node[ele,label=right:{\tiny{$\varphi^{3}\ \alpha$}}] (d) at (3,0.5) {};
  \node[ele,label=above:{\tiny{$\varphi^{4}\ \alpha$}}] (e) at (2,1) {};
  \node[ele,label=above:{\tiny{$\varphi^{5}\ \alpha$}}] (f) at (1,1) {};
  \node[ele,label=above:{\tiny{$\varphi^{6}\ \alpha\qquad$}}] (g) at (0,0.5) {};
  \draw[->,dashed,shorten <=2pt,shorten >=2] (a) -- (b);
  \draw[->,dashed,shorten <=2pt,shorten >=2] (b) -- (c);
  \draw[->,dashed,shorten <=2pt,shorten >=2] (c) -- (d);
  \draw[->,dashed,shorten <=2pt,shorten >=2] (d) -- (e);
  \draw[->,dashed,shorten <=2pt,shorten >=2] (e) -- (f);
  \draw[->,dashed,shorten <=2pt,shorten >=2] (f) -- (g);
  \node[ele,draw,fill=white] (ab) at ($(a)!0.5!(b)$) {};
  \node[ele,draw,fill=white] (bc) at ($(b)!0.5!(c)$) {};
  \node[ele,draw,fill=white] (cd) at ($(c)!0.5!(d)$) {};
  \node[ele,draw,fill=white] (de) at ($(d)!0.5!(e)$) {};
  \node[ele,draw,fill=white] (ef) at ($(e)!0.5!(f)$) {};
  \draw[thick,color=blue] (a) to [out=130, in=-130,looseness=50]
       node[midway,left] {\tiny{$f$}} (a);
  \draw[thick,color=purple] (a) to [bend right] node[midway,below] {\tiny{\HOLinline{\HOLFreeVar{g}}}} (ab);
  \draw[thick,color=blue] (ab) to [bend right] node[midway,below] {\tiny{$f$}} (b);
  \draw[thick,color=purple] (b) to [bend right] node[midway,below] {\tiny{\HOLinline{\HOLFreeVar{g}}}} (bc);
  \draw[thick,color=blue] (bc) to [bend right] node[midway,below] {\tiny{$f$}} (c);
  \draw[thick,color=purple] (c) to [bend right] node[midway,below] {\tiny{\HOLinline{\HOLFreeVar{g}}}} (cd);
  \draw[thick,color=blue] (cd) to [bend right] node[midway,below] {\tiny{$f$}} (d);
  \draw[thick,color=purple] (d) to [bend right] node[midway,above] {\tiny{\HOLinline{\HOLFreeVar{g}}}} (de);
  \draw[thick,color=blue] (de) to [bend right] node[midway,above] {\tiny{$f$}} (e);
  \draw[thick,color=purple] (e) to [bend right] node[midway,above] {\tiny{\HOLinline{\HOLFreeVar{g}}}} (ef);
  \draw[thick,color=blue] (ef) to [bend right] node[midway,above] {\tiny{$f$}} (f);
  \draw[thick,color=purple] (f) to [bend right] node[midway,above] {\tiny{\HOLinline{\HOLFreeVar{g}}}} (g);
  \DoubleLine[0.3pt]{f}{ab}{red}{red}
  \DoubleLine[0.3pt]{ef}{b}{red}{red}
  \DoubleLine[0.3pt]{e}{bc}{red}{red}
  \DoubleLine[0.3pt]{de}{c}{red}{red}
  \DoubleLine[0.3pt]{cd}{d}{red}{red}
  \path[ultra thick, glow=green] (cd) -- (d);
\end{tikzpicture}
\begin{tikzpicture}[scale=1.7, 
    ele/.style={fill=black,circle,minimum width=.8pt,inner sep=1pt},
    every fit/.style={ellipse,draw,inner sep=5pt}]
  \node[ele,label=left:{\tiny{$\alpha$}}] (a) at (0,0.5) {};    
  \node[ele,label=below:{\tiny{$\varphi\ \alpha$}}] (b) at (1,0) {};
  \node[ele,label=below:{\tiny{$\varphi^{2}\ \alpha$}}] (c) at (2,0) {};    
  \node[ele,label=right:{\tiny{$\varphi^{3}\ \alpha$}}] (d) at (3,0.5) {};
  \node[ele,label=above:{\tiny{$\varphi^{4}\ \alpha$}}] (e) at (2,1) {};
  \node[ele,label=above:{\tiny{$\varphi^{5}\ \alpha$}}] (f) at (1,1) {};
  \node[ele,label=above:{\tiny{$\varphi^{6}\ \alpha\qquad$}}] (g) at (0,0.5) {};
  \draw[->,dashed,shorten <=2pt,shorten >=2] (a) -- (b);
  \draw[->,dashed,shorten <=2pt,shorten >=2] (b) -- (c);
  \draw[->,dashed,shorten <=2pt,shorten >=2] (c) -- (d);
  \draw[->,dashed,shorten <=2pt,shorten >=2] (d) -- (e);
  \draw[->,dashed,shorten <=2pt,shorten >=2] (e) -- (f);
  \draw[->,dashed,shorten <=2pt,shorten >=2] (f) -- (g);
  \draw[thick,color=blue] (a) to [out=130, in=-130,looseness=50]
       node[midway,left] {\tiny{$f$}} (a);
  \draw[thick,color=blue] (f) to [bend right] node[midway,left] {\tiny{$f$}} (b);
  \draw[thick,color=blue] (e) to [bend left] node[midway,right] {\tiny{$f$}} (c);
  \draw[thick,color=blue] (d) to [out=50, in=-50,looseness=50]
       node[midway,right] {\tiny{$f$}} (d);
  \draw[thick,color=purple] (c) to [bend right] node[midway,below] {\tiny{\HOLinline{\HOLFreeVar{g}}}} (d);
  \draw[thick,color=purple] (e) to [bend left] node[midway,above] {\tiny{\HOLinline{\HOLFreeVar{g}}}} (b);
  \draw[thick,color=purple] (f) to [bend right] node[midway,above] {\tiny{\HOLinline{\HOLFreeVar{g}}}} (g);
\end{tikzpicture}
\caption{Orbit from an $f$ fixed point $\alpha$ with even period $6$. Identical points on the left (marked by two parallel lines) are merged on the right (move white dot to black dot). In particular, on the left the two vertices of the shaded line are the same, forming a fixed point of $f$ on the right.}
\Description{This figure shows an orbit from an f fixed point alpha, with even period 6. Identical points on the left, marked by two parallel lines, are merged on the right, by moving white dot to black dot. In particular, on the left the two vertices of the shaded line are the same, forming a fixed point of f on the right.}
\label{fig:orbit-fix-even} 
\end{figure*}

This orbit is formed by taking the left diagram of Figure~\ref{fig:orbits-even-odd},
but identifying the black dot on $\alpha$ (the leftmost one) with its preceding white dot from $f$, since \HOLinline{\HOLFreeVar{f}\;\HOLFreeVar{\alpha}\;\HOLSymConst{=}\;\HOLFreeVar{\alpha}}, giving the left $f$-loop.
This node $\alpha$ is now preceded by two \HOLinline{\HOLFreeVar{g}}-arcs, one from a black dot and one from a white dot. However, \HOLinline{\HOLFreeVar{g}} is an involution, which is injective, so the two dots are identical. The same reasoning shows that all the dots linked by double lines are identical, so that the orbit on the left can be simplified to the one on the right, taking only black dots.

Moreover, the rightmost black dot and a preceding white dot from $f$ must be the same, due to \HOLinline{\HOLFreeVar{g}}-arcs from identical dots.
This means the half-period iterate, the rightmost black dot, is another fixed point of $f$, say $\alpha'$.
Note that $\alpha' \ne \alpha$, for otherwise the period will be affected.
This example motivates the following:
\begin{theorem}
\label{thm:involute-two-fixes-even}
\script{iterateCompose}{884}
When $f$ fixes $x$, and (\HOLinline{\HOLFreeVar{f}\;\HOLSymConst{\HOLTokenCompose}\;\HOLFreeVar{g}}) has an even period $p$ for $x$,
then $f$ also fixes \HOLinline{(\HOLFreeVar{f}\;\HOLSymConst{\HOLTokenCompose}\;\HOLFreeVar{g})\ensuremath{\sp{\HOLFreeVar{p}\;\HOLConst{div}\;2}(\HOLFreeVar{x})}}, which is not $x$ itself.
\begin{HOLmath}
\HOLTokenTurnstile{}\HOLConst{\HOLConst{finite}}\;\HOLFreeVar{S}\;\HOLSymConst{\HOLTokenConj{}}\;\HOLFreeVar{f}\;\HOLConst{involute}\;\HOLFreeVar{S}\;\HOLSymConst{\HOLTokenConj{}}\;\HOLFreeVar{g}\;\HOLConst{involute}\;\HOLFreeVar{S}\;\HOLSymConst{\HOLTokenConj{}}\\
\;\;\;\;\;\HOLFreeVar{x}\;\HOLSymConst{\HOLTokenIn{}}\;\HOLConst{fixes}\;\HOLFreeVar{f}\;\HOLFreeVar{S}\;\HOLSymConst{\HOLTokenConj{}}\;\HOLFreeVar{p}\;\HOLSymConst{=}\;\HOLConst{\HOLConst{period}}\;(\HOLFreeVar{f}\;\HOLSymConst{\HOLTokenCompose}\;\HOLFreeVar{g})\;\HOLFreeVar{x}\;\HOLSymConst{\HOLTokenConj{}}\\
\;\;\;\;\;\HOLFreeVar{y}\;\HOLSymConst{=}\;(\HOLFreeVar{f}\;\HOLSymConst{\HOLTokenCompose}\;\HOLFreeVar{g})\ensuremath{\sp{\HOLFreeVar{p}\;\HOLConst{div}\;2}(\HOLFreeVar{x})}\;\HOLSymConst{\HOLTokenConj{}}\;\HOLConst{\HOLConst{even}}\;\HOLFreeVar{p}\;\HOLSymConst{\HOLTokenImp{}}\\
\;\;\;\;\;\;\;\HOLFreeVar{y}\;\HOLSymConst{\HOLTokenIn{}}\;\HOLConst{fixes}\;\HOLFreeVar{f}\;\HOLFreeVar{S}\;\HOLSymConst{\HOLTokenConj{}}\;\HOLFreeVar{y}\;\HOLSymConst{\HOLTokenNotEqual{}}\;\HOLFreeVar{x}
\end{HOLmath}
\end{theorem}
\begin{proof}
First we show that $f$ fixes $y$.
Let \HOLinline{\HOLFreeVar{h}\;\HOLSymConst{=}\;\HOLFreeVar{p}\;\HOLConst{\HOLConst{div}}\;\HOLNumLit{2}}.
Since period $p$ is even, $p \ee \HOLinline{\HOLNumLit{2}\HOLSymConst{\ensuremath{}}\HOLFreeVar{h}\;\HOLSymConst{=}\;\HOLFreeVar{h}\;\HOLSymConst{\ensuremath{+}}\;\HOLFreeVar{h}}$.
This implies that \HOLinline{\HOLFreeVar{h}\;\HOLSymConst{\ensuremath{+}}\;\HOLFreeVar{h}\;\ensuremath{\equiv}\;\HOLNumLit{0}\;\ensuremath{(}\ensuremath{\bmod}\;\HOLFreeVar{p}\ensuremath{)}},
so \HOLinline{\HOLFreeVar{y}\;\HOLSymConst{=}\;\HOLFreeVar{f}\;\HOLFreeVar{y}} by Theorem~\ref{thm:involute-two-fix-orbit-1}.
Since (\HOLinline{\HOLFreeVar{f}\;\HOLSymConst{\HOLTokenCompose}\;\HOLFreeVar{g}}) is a permutation, \HOLinline{\HOLFreeVar{y}\;\HOLSymConst{\HOLTokenIn{}}\;\HOLFreeVar{S}}, so \HOLinline{\HOLFreeVar{y}\;\HOLSymConst{\HOLTokenIn{}}\;\HOLConst{fixes}\;\HOLFreeVar{f}\;\HOLFreeVar{S}}.
Next we show that \HOLinline{\HOLFreeVar{y}\;\HOLSymConst{\HOLTokenNotEqual{}}\;\HOLFreeVar{x}}. 
Suppose \HOLinline{\HOLFreeVar{y}\;\HOLSymConst{=}\;\HOLFreeVar{x}}.
Since for finite $S$ the period \HOLinline{\HOLFreeVar{p}\;\HOLSymConst{\HOLTokenNotEqual{}}\;\HOLNumLit{0}},
Theorem~\ref{thm:iterate-period-mod} shows that
$p$ divides \HOLinline{\HOLFreeVar{h}\;\HOLSymConst{=}\;\HOLFreeVar{p}\;\HOLConst{\HOLConst{div}}\;\HOLNumLit{2}}.
Hence \HOLinline{\HOLFreeVar{p}\;\HOLSymConst{=}\;\HOLNumLit{1}}, which is not even.
\end{proof}
\noindent
Therefore if a fixed point of $f$ has an even period under \HOLinline{\HOLFreeVar{\varphi}\;\HOLSymConst{=}\;\HOLFreeVar{f}\;\HOLSymConst{\HOLTokenCompose}\;\HOLFreeVar{g}}, it is not alone. This leads directly to:
\begin{corollary}
\label{cor:involute-fix-singleton-odd}
\script{iterateCompute}{1009}
If $f$ fixes only a single $x$, then \HOLinline{\HOLFreeVar{f}\;\HOLSymConst{\HOLTokenCompose}\;\HOLFreeVar{g}} has an odd period for $x$.
\begin{HOLmath}
\HOLTokenTurnstile{}\HOLConst{\HOLConst{finite}}\;\HOLFreeVar{S}\;\HOLSymConst{\HOLTokenConj{}}\;\HOLFreeVar{f}\;\HOLConst{involute}\;\HOLFreeVar{S}\;\HOLSymConst{\HOLTokenConj{}}\;\HOLFreeVar{g}\;\HOLConst{involute}\;\HOLFreeVar{S}\;\HOLSymConst{\HOLTokenConj{}}\\
\;\;\;\;\;\HOLConst{fixes}\;\HOLFreeVar{f}\;\HOLFreeVar{S}\;\HOLSymConst{=}\;\HOLTokenLeftbrace{}\HOLFreeVar{x}\HOLTokenRightbrace{}\;\HOLSymConst{\HOLTokenConj{}}\;\HOLFreeVar{p}\;\HOLSymConst{=}\;\HOLConst{\HOLConst{period}}\;(\HOLFreeVar{f}\;\HOLSymConst{\HOLTokenCompose}\;\HOLFreeVar{g})\;\HOLFreeVar{x}\;\HOLSymConst{\HOLTokenImp{}}\\
\;\;\;\;\;\;\;\HOLConst{\HOLConst{odd}}\;\HOLFreeVar{p}
\end{HOLmath}
\end{corollary}

\subsection{Fixed Point Period Odd}
\label{sec:fixed-point-period-odd}
Now consider an orbit with odd period starting with $\alpha$, a fixed point of $f$.
Figure~\ref{fig:orbit-fix-odd} shows one on the left, and its real picture on the right.
This orbit is formed by taking the right diagram of Figure~\ref{fig:orbits-even-odd},
but identifying the black dot on $\alpha$ (the leftmost one) with its preceding white dot from $f$, since \HOLinline{\HOLFreeVar{f}\;\HOLFreeVar{\alpha}\;\HOLSymConst{=}\;\HOLFreeVar{\alpha}}, giving the left $f$-loop.
The same reasoning as the even period orbit of Section~\ref{sec:fixed-point-period-even} shows that all the dots linked by double lines are identical, so that the orbit on the left can be simplified to the one on the right, again taking only black dots.

\begin{figure*}[h]
\centering
\begin{tikzpicture}[scale=1.7, 
    ele/.style={fill=black,circle,minimum width=.8pt,inner sep=1pt},
    every fit/.style={ellipse,draw,inner sep=5pt}]
  \node[ele,label=left:{\tiny{$\alpha$}}] (a) at (0,0.5) {};    
  \node[ele,label=below:{\tiny{$\varphi\ \alpha$}}] (b) at (1,0) {};
  \node[ele,label=below:{\tiny{$\varphi^{2}\ \alpha$}}] (c) at (2,0) {};    
  \node[ele,label=right:{\tiny{$\varphi^{3}\ \alpha$}}] (d) at (3,0.5) {};
  \node[ele,label=above:{\tiny{$\varphi^{4}\ \alpha$}}] (e) at (2.5,1) {};
  \node[ele,label=above:{\tiny{$\varphi^{5}\ \alpha$}}] (f) at (1.5,1) {};
  \node[ele,label=above:{\tiny{$\varphi^{6}\ \alpha$}}] (g) at (0.5,1) {};
  \node[ele,label=above:{\tiny{$\varphi^{7}\ \alpha\qquad$}}] (h) at (0,0.5) {};
  \draw[->,dashed,shorten <=2pt,shorten >=2] (a) -- (b);
  \draw[->,dashed,shorten <=2pt,shorten >=2] (b) -- (c);
  \draw[->,dashed,shorten <=2pt,shorten >=2] (c) -- (d);
  \draw[->,dashed,shorten <=2pt,shorten >=2] (d) -- (e);
  \draw[->,dashed,shorten <=2pt,shorten >=2] (e) -- (f);
  \draw[->,dashed,shorten <=2pt,shorten >=2] (f) -- (g);
  \draw[->,dashed,shorten <=2pt,shorten >=2] (g) -- (h);
  \node[ele,draw,fill=white] (ab) at ($(a)!0.5!(b)$) {};
  \node[ele,draw,fill=white] (bc) at ($(b)!0.5!(c)$) {};
  \node[ele,draw,fill=white] (cd) at ($(c)!0.5!(d)$) {};
  \node[ele,draw,fill=white] (de) at ($(d)!0.5!(e)$) {};
  \node[ele,draw,fill=white] (ef) at ($(e)!0.5!(f)$) {};
  \node[ele,draw,fill=white] (fg) at ($(f)!0.5!(g)$) {};
  \draw[thick,color=blue] (a) to [out=130, in=-130,looseness=50]
       node[midway,left] {\tiny{$f$}} (a);
  \draw[thick,color=purple] (a) to [bend right] node[midway,below] {\tiny{\HOLinline{\HOLFreeVar{g}}}} (ab);
  \draw[thick,color=blue] (ab) to [bend right] node[midway,below] {\tiny{$f$}} (b);
  \draw[thick,color=purple] (b) to [bend right] node[midway,below] {\tiny{\HOLinline{\HOLFreeVar{g}}}} (bc);
  \draw[thick,color=blue] (bc) to [bend right] node[midway,below] {\tiny{$f$}} (c);
  \draw[thick,color=purple] (c) to [bend right] node[midway,below] {\tiny{\HOLinline{\HOLFreeVar{g}}}} (cd);
  \draw[thick,color=blue] (cd) to [bend right] node[midway,below] {\tiny{$f$}} (d);
  \draw[thick,color=purple] (d) to [bend right] node[midway,above] {\tiny{\HOLinline{\HOLFreeVar{g}}}} (de);
  \draw[thick,color=blue] (de) to [bend right] node[midway,above] {\tiny{$f$}} (e);
  \draw[thick,color=purple] (e) to [bend right] node[midway,above] {\tiny{\HOLinline{\HOLFreeVar{g}}}} (ef);
  \draw[thick,color=blue] (ef) to [bend right] node[midway,above] {\tiny{$f$}} (f);
  \draw[thick,color=purple] (f) to [bend right] node[midway,above] {\tiny{\HOLinline{\HOLFreeVar{g}}}} (fg);
  \draw[thick,color=blue] (fg) to [bend right] node[midway,above] {\tiny{$f$}} (g);
  \draw[thick,color=purple] (g) to [bend right] node[midway,above] {\tiny{\HOLinline{\HOLFreeVar{g}}}} (h);
  \DoubleLine[0.3pt]{g}{ab}{red}{red}
  \DoubleLine[0.3pt]{fg}{b}{red}{red}
  \DoubleLine[0.3pt]{f}{bc}{red}{red}
  \DoubleLine[0.3pt]{ef}{c}{red}{red}
  \DoubleLine[0.3pt]{e}{cd}{red}{red}
  \DoubleLine[0.3pt]{de}{d}{red}{red}
  \path[ultra thick, glow=green] (de) -- (d);
\end{tikzpicture}
\begin{tikzpicture}[scale=1.7, 
    ele/.style={fill=black,circle,minimum width=.8pt,inner sep=1pt},
    every fit/.style={ellipse,draw,inner sep=5pt}]
  \node[ele,label=left:{\tiny{$\alpha$}}] (a) at (0,0.5) {};    
  \node[ele,label=below:{\tiny{$\varphi\ \alpha$}}] (b) at (1,0) {};
  \node[ele,label=below:{\tiny{$\varphi^{2}\ \alpha$}}] (c) at (2,0) {};    
  \node[ele,label=right:{\tiny{$\varphi^{3}\ \alpha$}}] (d) at (3,0.5) {};
  \node[ele,label=above:{\tiny{$\varphi^{4}\ \alpha$}}] (e) at (2.5,1) {};
  \node[ele,label=above:{\tiny{$\varphi^{5}\ \alpha$}}] (f) at (1.5,1) {};
  \node[ele,label=above:{\tiny{$\varphi^{6}\ \alpha$}}] (g) at (0.5,1) {};
  \node[ele,label=above:{\tiny{$\varphi^{7}\ \alpha\qquad$}}] (h) at (0,0.5) {};

  \draw[->,dashed,shorten <=2pt,shorten >=2] (a) -- (b);
  \draw[->,dashed,shorten <=2pt,shorten >=2] (b) -- (c);
  \draw[->,dashed,shorten <=2pt,shorten >=2] (c) -- (d);
  \draw[->,dashed,shorten <=2pt,shorten >=2] (d) -- (e);
  \draw[->,dashed,shorten <=2pt,shorten >=2] (e) -- (f);
  \draw[->,dashed,shorten <=2pt,shorten >=2] (f) -- (g);
  \draw[->,dashed,shorten <=2pt,shorten >=2] (g) -- (h);
  \draw[thick,color=purple] (f) to [bend left] node[midway,right] {\tiny{\HOLinline{\HOLFreeVar{g}}}} (b);
  \draw[thick,color=purple] (g) to [bend left] node[midway,below] {\tiny{\HOLinline{\HOLFreeVar{g}}}} (a);
  \draw[thick,color=purple] (e) to [bend left] node[midway,right] {\tiny{\HOLinline{\HOLFreeVar{g}}}} (c);
  \draw[thick,color=purple] (d) to [out=50, in=-50,looseness=50]
       node[midway,right] {\tiny{\HOLinline{\HOLFreeVar{g}}}} (d);
  \draw[thick,color=blue] (a) to [out=130, in=-130,looseness=50]
       node[midway,left] {\tiny{$f$}} (a);
  \draw[thick,color=blue] (d) to [bend right] node[midway,above] {\tiny{$f$}} (e);
  \draw[thick,color=blue] (c) to [bend right] node[midway,right] {\tiny{$f$}} (f);
  \draw[thick,color=blue] (b) to [bend right] node[midway,left]  {\tiny{$f$}} (g);
\end{tikzpicture}
\caption{Orbit from an $f$ fixed point $\alpha$ with odd period $7$. Identical points on the left (marked by two parallel lines) are merged~on the right (move white dot to black dot). In particular, on the left the two vertices of the shaded line are the same, forming a fixed point of $g$ on the right.}
\Description{This figure shows an orbit from an f fixed point alpha with odd period 7. Identical points on the left, marked by two parallel lines, are merged on the right, by moving white dot to black dot. In particular, on the left the two vertices of the shaded line are the same, forming a fixed point of g on the right.}
\label{fig:orbit-fix-odd} 
\end{figure*}
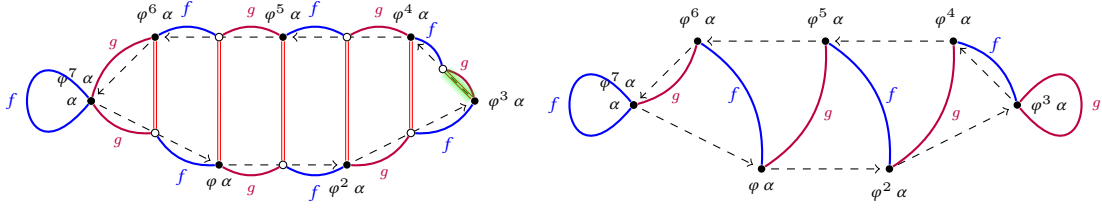

Moreover, the rightmost black dot and a preceding white dot from \HOLinline{\HOLFreeVar{g}} must be the same, due to $f$-arcs from identical dots.
This means the half-period iterate, the rightmost black dot, must be a fixed point of \HOLinline{\HOLFreeVar{g}}, say $\beta$.
If \HOLinline{\HOLFreeVar{\beta}\;\HOLSymConst{=}\;\HOLFreeVar{\alpha}}, then period \HOLinline{\HOLFreeVar{p}\;\HOLSymConst{=}\;\HOLNumLit{1}}, in accordance with Theorem~\ref{thm:involute-period-1}.
This example motivates the following:
\begin{theorem}
\label{thm:involute-two-fixes-odd}
\script{iterateCompose}{980}
When $f$ fixes $x$, and (\HOLinline{\HOLFreeVar{f}\;\HOLSymConst{\HOLTokenCompose}\;\HOLFreeVar{g}}) has an odd period $p$ for $x$,
then \HOLinline{\HOLFreeVar{g}} fixes \HOLinline{(\HOLFreeVar{f}\;\HOLSymConst{\HOLTokenCompose}\;\HOLFreeVar{g})\ensuremath{\sp{\HOLFreeVar{p}\;\HOLConst{div}\;2}(\HOLFreeVar{x})}}, which is not $x$ itself if and only if \HOLinline{\HOLFreeVar{p}\;\HOLSymConst{\HOLTokenNotEqual{}}\;\HOLNumLit{1}}.
\begin{HOLmath}
\HOLTokenTurnstile{}\HOLConst{\HOLConst{finite}}\;\HOLFreeVar{S}\;\HOLSymConst{\HOLTokenConj{}}\;\HOLFreeVar{f}\;\HOLConst{involute}\;\HOLFreeVar{S}\;\HOLSymConst{\HOLTokenConj{}}\;\HOLFreeVar{g}\;\HOLConst{involute}\;\HOLFreeVar{S}\;\HOLSymConst{\HOLTokenConj{}}\\
\;\;\;\;\;\HOLFreeVar{x}\;\HOLSymConst{\HOLTokenIn{}}\;\HOLConst{fixes}\;\HOLFreeVar{f}\;\HOLFreeVar{S}\;\HOLSymConst{\HOLTokenConj{}}\;\HOLFreeVar{p}\;\HOLSymConst{=}\;\HOLConst{\HOLConst{period}}\;(\HOLFreeVar{f}\;\HOLSymConst{\HOLTokenCompose}\;\HOLFreeVar{g})\;\HOLFreeVar{x}\;\HOLSymConst{\HOLTokenConj{}}\\
\;\;\;\;\;\HOLFreeVar{y}\;\HOLSymConst{=}\;(\HOLFreeVar{f}\;\HOLSymConst{\HOLTokenCompose}\;\HOLFreeVar{g})\ensuremath{\sp{\HOLFreeVar{p}\;\HOLConst{div}\;2}(\HOLFreeVar{x})}\;\HOLSymConst{\HOLTokenConj{}}\;\HOLConst{\HOLConst{odd}}\;\HOLFreeVar{p}\;\HOLSymConst{\HOLTokenImp{}}\\
\;\;\;\;\;\;\;\HOLFreeVar{y}\;\HOLSymConst{\HOLTokenIn{}}\;\HOLConst{fixes}\;\HOLFreeVar{g}\;\HOLFreeVar{S}\;\HOLSymConst{\HOLTokenConj{}}\;(\HOLFreeVar{y}\;\HOLSymConst{=}\;\HOLFreeVar{x}\;\HOLSymConst{\HOLTokenEquiv{}}\;\HOLFreeVar{p}\;\HOLSymConst{=}\;\HOLNumLit{1})
\end{HOLmath}
\end{theorem}
\begin{proof}
First we show that \HOLinline{\HOLFreeVar{g}} fixes $y$.
Let \HOLinline{\HOLFreeVar{h}\;\HOLSymConst{=}\;\HOLFreeVar{p}\;\HOLConst{\HOLConst{div}}\;\HOLNumLit{2}}.
Since period $p$ is odd, $p \ee \HOLinline{\HOLNumLit{2}\HOLSymConst{\ensuremath{}}\HOLFreeVar{h}\;\HOLSymConst{\ensuremath{+}}\;\HOLNumLit{1}\;\HOLSymConst{=}\;\HOLFreeVar{h}\;\HOLSymConst{\ensuremath{+}}\;\HOLFreeVar{h}\;\HOLSymConst{\ensuremath{+}}\;\HOLNumLit{1}}$.
Thus \HOLinline{\HOLFreeVar{h}\;\HOLSymConst{\ensuremath{+}}\;\HOLFreeVar{h}\;\HOLSymConst{\ensuremath{+}}\;\HOLNumLit{1}\;\ensuremath{\equiv}\;\HOLNumLit{0}\;\ensuremath{(}\ensuremath{\bmod}\;\HOLFreeVar{p}\ensuremath{)}},
so \HOLinline{\HOLFreeVar{y}\;\HOLSymConst{=}\;\HOLFreeVar{g}\;\HOLFreeVar{y}} by Theorem~\ref{thm:involute-two-fix-orbit-2}.
As (\HOLinline{\HOLFreeVar{f}\;\HOLSymConst{\HOLTokenCompose}\;\HOLFreeVar{g}}) is a permutation, \HOLinline{\HOLFreeVar{y}\;\HOLSymConst{\HOLTokenIn{}}\;\HOLFreeVar{S}}, so \HOLinline{\HOLFreeVar{y}\;\HOLSymConst{\HOLTokenIn{}}\;\HOLConst{fixes}\;\HOLFreeVar{g}\;\HOLFreeVar{S}}.
Theorem~\ref{thm:involute-period-1} ensures that: \HOLinline{\HOLFreeVar{y}\;\HOLSymConst{=}\;\HOLFreeVar{x}\;\HOLSymConst{\HOLTokenEquiv{}}\;\HOLFreeVar{p}\;\HOLSymConst{=}\;\HOLNumLit{1}}.
\end{proof}

\subsection{Fixed Point Orbits}
\label{sec:fixed-point-orbits}
Let \HOLinline{\HOLFreeVar{\varphi}\;\HOLSymConst{=}\;\HOLFreeVar{f}\;\HOLSymConst{\HOLTokenCompose}\;\HOLFreeVar{g}},
and $\alpha, \beta$ be fixed points of $f, \HOLinline{\HOLFreeVar{g}}$, respectively.
Theorem~\ref{thm:involute-two-fixes-even} and Theorem~\ref{thm:involute-two-fixes-odd} show that:
\begin{itemize}[leftmargin=*] 
\item if the period $p$ of $\alpha$ is even, its orbit has another fixed point of $f$ at the \HOLinline{\HOLFreeVar{h}\;\HOLSymConst{=}\;\HOLFreeVar{p}\;\HOLConst{\HOLConst{div}}\;\HOLNumLit{2}} iterate: \HOLinline{\HOLFreeVar{\varphi}\ensuremath{\sp{\HOLFreeVar{h}}(\HOLFreeVar{\alpha})}}.
\item if the period $p$ of $\alpha$ is odd, its orbit has another fixed point of \HOLinline{\HOLFreeVar{g}} at the \HOLinline{\HOLFreeVar{h}\;\HOLSymConst{=}\;\HOLFreeVar{p}\;\HOLConst{\HOLConst{div}}\;\HOLNumLit{2}} iterate: \HOLinline{\HOLFreeVar{\beta}\;\HOLSymConst{=}\;\HOLFreeVar{\varphi}\ensuremath{\sp{\HOLFreeVar{h}}(\HOLFreeVar{\alpha})}}.
\end{itemize}
Figure~\ref{fig:orbit-fix-even} and Figure~\ref{fig:orbit-fix-odd} show that these orbits have no more fixed points. The only fixed point, of either $f$ or $g$, occurs at halfway point of the orbit.

Thus, fixed point orbits lead directly from one fixed point to another. This is because, assuming one of the intermediate iterate is a fixed point, the iteration path will turn back, due to either $f$ or $g$, both being involutions. This will produce an orbit with a shorter period, but period for an orbit is minimal.

Such considerations lead to the following stronger forms of Theorem~\ref{thm:involute-two-fixes-even} and Theorem~\ref{thm:involute-two-fixes-odd}:
\begin{theorem}
\label{thm:involute-two-fixes-even-odd}
\script{iterateCompose}{1100}
When $f$ fixes $x$, the $j$-th iterate of (\HOLinline{\HOLFreeVar{f}\;\HOLSymConst{\HOLTokenCompose}\;\HOLFreeVar{g}}) from $x$ is a fixed point of either $f$ or \HOLinline{\HOLFreeVar{g}} if and only if $j$ is half of the period $p$.
\begin{HOLmath}
\HOLTokenTurnstile{}\HOLConst{\HOLConst{finite}}\;\HOLFreeVar{S}\;\HOLSymConst{\HOLTokenConj{}}\;\HOLFreeVar{f}\;\HOLConst{involute}\;\HOLFreeVar{S}\;\HOLSymConst{\HOLTokenConj{}}\;\HOLFreeVar{g}\;\HOLConst{involute}\;\HOLFreeVar{S}\;\HOLSymConst{\HOLTokenConj{}}\\
\;\;\;\;\;\HOLFreeVar{x}\;\HOLSymConst{\HOLTokenIn{}}\;\HOLConst{fixes}\;\HOLFreeVar{f}\;\HOLFreeVar{S}\;\HOLSymConst{\HOLTokenConj{}}\;\HOLFreeVar{p}\;\HOLSymConst{=}\;\HOLConst{\HOLConst{period}}\;(\HOLFreeVar{f}\;\HOLSymConst{\HOLTokenCompose}\;\HOLFreeVar{g})\;\HOLFreeVar{x}\;\HOLSymConst{\HOLTokenConj{}}\;\HOLConst{\HOLConst{even}}\;\HOLFreeVar{p}\;\HOLSymConst{\HOLTokenImp{}}\\
\;\;\;\;\;\;\;\HOLSymConst{\HOLTokenForall{}}\HOLBoundVar{j}.\;\HOLNumLit{0}\;\HOLSymConst{\HOLTokenLt{}}\;\HOLBoundVar{j}\;\HOLSymConst{\HOLTokenConj{}}\;\HOLBoundVar{j}\;\HOLSymConst{\HOLTokenLt{}}\;\HOLFreeVar{p}\;\HOLSymConst{\HOLTokenImp{}}\\
\;\;\;\;\;\;\;\;\;\;\;\;\;((\HOLFreeVar{f}\;\HOLSymConst{\HOLTokenCompose}\;\HOLFreeVar{g})\ensuremath{\sp{\HOLBoundVar{j}}(\HOLFreeVar{x})}\;\HOLSymConst{\HOLTokenIn{}}\;\HOLConst{fixes}\;\HOLFreeVar{f}\;\HOLFreeVar{S}\;\HOLSymConst{\HOLTokenEquiv{}}\;\HOLBoundVar{j}\;\HOLSymConst{=}\;\HOLFreeVar{p}\;\HOLConst{\HOLConst{div}}\;\HOLNumLit{2})\\
\HOLTokenTurnstile{}\HOLConst{\HOLConst{finite}}\;\HOLFreeVar{S}\;\HOLSymConst{\HOLTokenConj{}}\;\HOLFreeVar{f}\;\HOLConst{involute}\;\HOLFreeVar{S}\;\HOLSymConst{\HOLTokenConj{}}\;\HOLFreeVar{g}\;\HOLConst{involute}\;\HOLFreeVar{S}\;\HOLSymConst{\HOLTokenConj{}}\\
\;\;\;\;\;\HOLFreeVar{x}\;\HOLSymConst{\HOLTokenIn{}}\;\HOLConst{fixes}\;\HOLFreeVar{f}\;\HOLFreeVar{S}\;\HOLSymConst{\HOLTokenConj{}}\;\HOLFreeVar{p}\;\HOLSymConst{=}\;\HOLConst{\HOLConst{period}}\;(\HOLFreeVar{f}\;\HOLSymConst{\HOLTokenCompose}\;\HOLFreeVar{g})\;\HOLFreeVar{x}\;\HOLSymConst{\HOLTokenConj{}}\;\HOLConst{\HOLConst{odd}}\;\HOLFreeVar{p}\;\HOLSymConst{\HOLTokenImp{}}\\
\;\;\;\;\;\;\;\HOLSymConst{\HOLTokenForall{}}\HOLBoundVar{j}.\;\HOLNumLit{0}\;\HOLSymConst{\HOLTokenLt{}}\;\HOLBoundVar{j}\;\HOLSymConst{\HOLTokenConj{}}\;\HOLBoundVar{j}\;\HOLSymConst{\HOLTokenLt{}}\;\HOLFreeVar{p}\;\HOLSymConst{\HOLTokenImp{}}\\
\;\;\;\;\;\;\;\;\;\;\;\;\;((\HOLFreeVar{f}\;\HOLSymConst{\HOLTokenCompose}\;\HOLFreeVar{g})\ensuremath{\sp{\HOLBoundVar{j}}(\HOLFreeVar{x})}\;\HOLSymConst{\HOLTokenIn{}}\;\HOLConst{fixes}\;\HOLFreeVar{g}\;\HOLFreeVar{S}\;\HOLSymConst{\HOLTokenEquiv{}}\;\HOLBoundVar{j}\;\HOLSymConst{=}\;\HOLFreeVar{p}\;\HOLConst{\HOLConst{div}}\;\HOLNumLit{2})
\end{HOLmath}
\end{theorem}
This completes our tour of the theory of permutation orbits and fixed points.
The results provide the key to formally prove that our two-squares algorithm by iterations is correct.

\section{Correctness of Algorithm}
\label{sec:correctness}
The algorithm to compute the flip fixed point from the known Zagier fixed point, given in Definition~\ref{def:two-sq-def}, makes use of a while-loop.
A while-loop consists of a guard $G$ and a body $B$, starting with an element $x$. The body is a function on $x$, producing iterates
$B(x)$, \HOLinline{\HOLFreeVar{B}\ensuremath{\sp{\HOLNumLit{2}}(\HOLFreeVar{x})}}, \HOLinline{\HOLFreeVar{B}\ensuremath{\sp{\HOLNumLit{3}}(\HOLFreeVar{a})}}, \etc.
The guard is a predicate on each iterate: the loop continues only if the test result by the guard stays true.

In HOL4, the \HOLConst{WHILE} loop with guard $G$ and body $B$ starting with $x$ is defined as:
\begin{HOLmath}
\;\;\HOLConst{WHILE}\;\HOLFreeVar{G}\;\HOLFreeVar{B}\;\HOLFreeVar{x}\;\HOLTokenDefEquality{}\;\HOLKeyword{if}\;\HOLFreeVar{G}\;\HOLFreeVar{x}\;\HOLKeyword{then}\;\HOLConst{WHILE}\;\HOLFreeVar{G}\;\HOLFreeVar{B}\;(\HOLFreeVar{B}\;\HOLFreeVar{x})\;\HOLKeyword{else}\;\HOLFreeVar{x}
\end{HOLmath}
from which one can easily show by induction that:
\begin{HOLmath}
\HOLTokenTurnstile{}(\HOLSymConst{\HOLTokenForall{}}\HOLBoundVar{j}.\;\HOLBoundVar{j}\;\HOLSymConst{\HOLTokenLt{}}\;\HOLFreeVar{k}\;\HOLSymConst{\HOLTokenImp{}}\;\HOLFreeVar{G}\;(\HOLFreeVar{B}\ensuremath{\sp{\HOLBoundVar{j}}(\HOLFreeVar{x})}))\;\HOLSymConst{\HOLTokenImp{}}\\
\;\;\;\;\;\;\;\HOLConst{WHILE}\;\HOLFreeVar{G}\;\HOLFreeVar{B}\;\HOLFreeVar{x}\;\HOLSymConst{=}\\
\;\;\;\;\;\;\;\;\;\HOLKeyword{if}\;\HOLFreeVar{G}\;(\HOLFreeVar{B}\ensuremath{\sp{\HOLFreeVar{k}}(\HOLFreeVar{x})})\;\HOLKeyword{then}\;\HOLConst{WHILE}\;\HOLFreeVar{G}\;\HOLFreeVar{B}\;(\HOLFreeVar{B}\ensuremath{\sp{\HOLFreeVar{k}\;\HOLSymConst{\ensuremath{+}}\;\HOLNumLit{1}}(\HOLFreeVar{x})})\;\HOLKeyword{else}\;\HOLFreeVar{B}\ensuremath{\sp{\HOLFreeVar{k}}(\HOLFreeVar{x})}
\end{HOLmath}
giving this expected result:
\begin{theorem}
\label{thm:iterate-while-thm}
\script{iterateCompute}{922}
The \HOLinline{\HOLConst{WHILE}} loop delivers the first body iterate that fails the guard test.
\begin{HOLmath}
\HOLTokenTurnstile{}(\HOLSymConst{\HOLTokenForall{}}\HOLBoundVar{j}.\;\HOLBoundVar{j}\;\HOLSymConst{\HOLTokenLt{}}\;\HOLFreeVar{k}\;\HOLSymConst{\HOLTokenImp{}}\;\HOLFreeVar{G}\;(\HOLFreeVar{B}\ensuremath{\sp{\HOLBoundVar{j}}(\HOLFreeVar{x})}))\;\HOLSymConst{\HOLTokenConj{}}\;\HOLSymConst{\HOLTokenNeg{}}\HOLFreeVar{G}\;(\HOLFreeVar{B}\ensuremath{\sp{\HOLFreeVar{k}}(\HOLFreeVar{x})})\;\HOLSymConst{\HOLTokenImp{}}\\
\;\;\;\;\;\;\;\HOLConst{WHILE}\;\HOLFreeVar{G}\;\HOLFreeVar{B}\;\HOLFreeVar{x}\;\HOLSymConst{=}\;\HOLFreeVar{B}\ensuremath{\sp{\HOLFreeVar{k}}(\HOLFreeVar{x})}
\end{HOLmath}
\end{theorem}

\subsection{Iterate with WHILE}
\label{sec:iterate-while}
From Section~\ref{sec:orbits}, we learn that for two involutions $f$ and \HOLinline{\HOLFreeVar{g}},
a fixed point $\alpha$ of $f$ is paired up with a fixed point $\beta$ of \HOLinline{\HOLFreeVar{g}}
whenever the period of $\alpha$ under the composition \HOLinline{\HOLFreeVar{\varphi}\;\HOLSymConst{=}\;\HOLFreeVar{f}\;\HOLSymConst{\HOLTokenCompose}\;\HOLFreeVar{g}} is odd.
In fact, $\beta$ lies in the orbit of $\alpha$ at halfway point, the iterate at half period.
Since a while-loop also gives an iterate, we have:
\begin{theorem}
\label{thm:involute-involute-fixes-while}
\script{iterateCompose}{1536}
For two involutions $f$ and \HOLinline{\HOLFreeVar{g}}, if $f$ fixes $x$ with an odd period,
a WHILE loop with \HOLinline{\HOLFreeVar{f}\;\HOLSymConst{\HOLTokenCompose}\;\HOLFreeVar{g}} from $x$ can reach a fixed point of \HOLinline{\HOLFreeVar{g}}.
\begin{HOLmath}
\HOLTokenTurnstile{}\HOLConst{\HOLConst{finite}}\;\HOLFreeVar{S}\;\HOLSymConst{\HOLTokenConj{}}\;\HOLFreeVar{f}\;\HOLConst{involute}\;\HOLFreeVar{S}\;\HOLSymConst{\HOLTokenConj{}}\;\HOLFreeVar{g}\;\HOLConst{involute}\;\HOLFreeVar{S}\;\HOLSymConst{\HOLTokenConj{}}\\
\;\;\;\;\;\HOLFreeVar{x}\;\HOLSymConst{\HOLTokenIn{}}\;\HOLConst{fixes}\;\HOLFreeVar{f}\;\HOLFreeVar{S}\;\HOLSymConst{\HOLTokenConj{}}\;\HOLFreeVar{p}\;\HOLSymConst{=}\;\HOLConst{\HOLConst{period}}\;(\HOLFreeVar{f}\;\HOLSymConst{\HOLTokenCompose}\;\HOLFreeVar{g})\;\HOLFreeVar{x}\;\HOLSymConst{\HOLTokenConj{}}\;\HOLConst{\HOLConst{odd}}\;\HOLFreeVar{p}\;\HOLSymConst{\HOLTokenImp{}}\\
\;\;\;\;\;\;\;\HOLConst{WHILE}\;(\HOLTokenLambda{}\HOLBoundVar{t}.\;\HOLFreeVar{g}\;\HOLBoundVar{t}\;\HOLSymConst{\HOLTokenNotEqual{}}\;\HOLBoundVar{t})\;(\HOLFreeVar{f}\;\HOLSymConst{\HOLTokenCompose}\;\HOLFreeVar{g})\;\HOLFreeVar{x}\;\HOLSymConst{\HOLTokenIn{}}\;\HOLConst{fixes}\;\HOLFreeVar{g}\;\HOLFreeVar{S}
\end{HOLmath}
\end{theorem}
\begin{proof}
Let guard \HOLinline{\HOLFreeVar{G}\;\HOLSymConst{=}\;(\HOLTokenLambda{}\HOLBoundVar{t}.\;\HOLFreeVar{g}\;\HOLBoundVar{t}\;\HOLSymConst{\HOLTokenNotEqual{}}\;\HOLBoundVar{t})}, and body \HOLinline{\HOLFreeVar{B}\;\HOLSymConst{=}\;\HOLFreeVar{f}\;\HOLSymConst{\HOLTokenCompose}\;\HOLFreeVar{g}}.
If period \HOLinline{\HOLFreeVar{p}\;\HOLSymConst{=}\;\HOLNumLit{1}},
then \HOLinline{\HOLFreeVar{x}\;\HOLSymConst{\HOLTokenIn{}}\;\HOLConst{fixes}\;\HOLFreeVar{g}\;\HOLFreeVar{S}} by Theorem~\ref{thm:involute-period-1}.
So \HOLinline{\HOLSymConst{\HOLTokenNeg{}}\HOLFreeVar{G}\;\HOLFreeVar{x}},
and \HOLinline{\HOLConst{WHILE}\;\HOLFreeVar{G}\;\HOLFreeVar{B}\;\HOLFreeVar{x}\;\HOLSymConst{=}\;\HOLFreeVar{x}} since the condition is not met at the start.
Therefore \HOLinline{\HOLConst{WHILE}\;\HOLFreeVar{G}\;\HOLFreeVar{B}\;\HOLFreeVar{x}\;\HOLSymConst{\HOLTokenIn{}}\;\HOLConst{fixes}\;\HOLFreeVar{g}\;\HOLFreeVar{s}}.

If period \HOLinline{\HOLFreeVar{p}\;\HOLSymConst{\HOLTokenNotEqual{}}\;\HOLNumLit{1}}, let \HOLinline{\HOLFreeVar{h}\;\HOLSymConst{=}\;\HOLFreeVar{p}\;\HOLConst{\HOLConst{div}}\;\HOLNumLit{2}},
and \HOLinline{\HOLFreeVar{z}\;\HOLSymConst{=}\;\HOLFreeVar{B}\ensuremath{\sp{\HOLFreeVar{h}}(\HOLFreeVar{x})}}.
Since \HOLinline{\HOLNumLit{1}\;\HOLSymConst{\HOLTokenLt{}}\;\HOLFreeVar{p}}, $0 < h < p$.
Also $f$ and \HOLinline{\HOLFreeVar{g}} are involutions, so \HOLinline{\HOLFreeVar{B}\;\HOLConst{\HOLConst{permutes}}\;\HOLFreeVar{S}}.
Hence \HOLinline{\HOLFreeVar{z}\;\HOLSymConst{\HOLTokenIn{}}\;\HOLConst{fixes}\;\HOLFreeVar{g}\;\HOLFreeVar{S}} by Theorem~\ref{thm:involute-two-fixes-odd}.
so \HOLinline{\HOLSymConst{\HOLTokenNeg{}}\HOLFreeVar{G}\;\HOLFreeVar{z}}.

We claim \HOLinline{\HOLSymConst{\HOLTokenForall{}}\HOLBoundVar{j}.\;\HOLBoundVar{j}\;\HOLSymConst{\HOLTokenLt{}}\;\HOLFreeVar{h}\;\HOLSymConst{\HOLTokenImp{}}\;\HOLFreeVar{G}\;(\HOLFreeVar{B}\ensuremath{\sp{\HOLBoundVar{j}}(\HOLFreeVar{x})})}.
To see this, let \HOLinline{\HOLFreeVar{y}\;\HOLSymConst{=}\;\HOLFreeVar{B}\ensuremath{\sp{\HOLFreeVar{j}}(\HOLFreeVar{x})}}, which is an element of $S$.
If \HOLinline{\HOLFreeVar{j}\;\HOLSymConst{=}\;\HOLNumLit{0}}, then \HOLinline{\HOLFreeVar{y}\;\HOLSymConst{=}\;\HOLFreeVar{x}}.
Since period \HOLinline{\HOLFreeVar{p}\;\HOLSymConst{\HOLTokenNotEqual{}}\;\HOLNumLit{1}},
\HOLinline{\HOLFreeVar{y}\;\HOLSymConst{\HOLTokenNotIn{}}\;\HOLConst{fixes}\;\HOLFreeVar{g}\;\HOLFreeVar{S}} by Theorem~\ref{thm:involute-period-1}, so \HOLinline{\HOLFreeVar{G}\;\HOLFreeVar{y}}.
If \HOLinline{\HOLFreeVar{j}\;\HOLSymConst{\HOLTokenNotEqual{}}\;\HOLNumLit{0}}, then $0 < j < h < p$, and \HOLinline{\HOLFreeVar{j}\;\HOLSymConst{\HOLTokenNotEqual{}}\;\HOLFreeVar{h}}.
Hence \HOLinline{\HOLFreeVar{y}\;\HOLSymConst{\HOLTokenNotIn{}}\;\HOLConst{fixes}\;\HOLFreeVar{g}\;\HOLFreeVar{S}} by Theorem~\ref{thm:involute-two-fixes-even-odd},
so \HOLinline{\HOLFreeVar{G}\;\HOLFreeVar{y}} again. The claim is proved.

By the claim and \HOLinline{\HOLSymConst{\HOLTokenNeg{}}\HOLFreeVar{G}\;\HOLFreeVar{z}},
apply Theorem~\ref{thm:iterate-while-thm} to conclude
$\HOLinline{\HOLConst{WHILE}\;\HOLFreeVar{G}\;\HOLFreeVar{B}\;\HOLFreeVar{x}\;\HOLSymConst{=}\;\HOLFreeVar{z}} \in \HOLinline{\HOLConst{fixes}\;\HOLFreeVar{g}\;\HOLFreeVar{S}}$.
\end{proof}

\subsection{Two Squares by WHILE}
\label{sec:two-squares-while}
We have developed the theory to show that the algorithm in Section~\ref{sec:algorithm} is correct:
\begin{theorem}
\label{thm:two-sq-thm}
\script{twoSquares}{840}
For a prime of the form \HOLinline{\HOLNumLit{4}\HOLSymConst{\ensuremath{}}\HOLFreeVar{k}\;\HOLSymConst{\ensuremath{+}}\;\HOLNumLit{1}},
the two squares algorithm of Definiton~\ref{def:two-sq-def} gives a flip fixed point.
\begin{HOLmath}
\HOLTokenTurnstile{}\HOLConst{prime}\;\HOLFreeVar{n}\;\HOLSymConst{\HOLTokenConj{}}\;\HOLFreeVar{n}\;\ensuremath{\equiv}\;\HOLNumLit{1}\;\ensuremath{(}\ensuremath{\bmod}\;\HOLNumLit{4}\ensuremath{)}\;\HOLSymConst{\HOLTokenImp{}}\\
\;\;\;\;\;\;\;\HOLConst{two_sq}\;\HOLFreeVar{n}\;\HOLSymConst{\HOLTokenIn{}}\;\HOLConst{fixes}\;\HOLConst{flip}\;(\HOLConst{mills}\;\HOLFreeVar{n})
\end{HOLmath}
\end{theorem}
\begin{proof}
Let \HOLinline{\HOLFreeVar{S}\;\HOLSymConst{=}\;\HOLConst{mills}\;\HOLFreeVar{n}}, \HOLinline{\HOLFreeVar{\varphi}\;\HOLSymConst{=}\;\HOLConst{zagier}\;\HOLSymConst{\HOLTokenCompose}\;\HOLConst{flip}},
\HOLinline{\HOLFreeVar{u}\;\HOLSymConst{=}\;(\HOLNumLit{1}\HOLSymConst{,}\HOLNumLit{1}\HOLSymConst{,}\HOLFreeVar{n}\;\HOLConst{\HOLConst{div}}\;\HOLNumLit{4})}, and period \HOLinline{\HOLFreeVar{p}\;\HOLSymConst{=}\;\HOLConst{\HOLConst{period}}\;\HOLFreeVar{\varphi}\;\HOLFreeVar{u}}.
By Definition~\ref{def:two-sq-def},
and noting that \HOLinline{(\HOLSymConst{\HOLTokenNeg{}})\;\HOLSymConst{\HOLTokenCompose}\;\HOLConst{found}\;\HOLSymConst{=}\;(\HOLTokenLambda{}\HOLBoundVar{t}.\;\HOLConst{flip}\;\HOLBoundVar{t}\;\HOLSymConst{\HOLTokenNotEqual{}}\;\HOLBoundVar{t})},
this is to show: \HOLinline{\HOLConst{WHILE}\;(\HOLTokenLambda{}\HOLBoundVar{t}.\;\HOLConst{flip}\;\HOLBoundVar{t}\;\HOLSymConst{\HOLTokenNotEqual{}}\;\HOLBoundVar{t})\;\HOLFreeVar{\varphi}\;\HOLFreeVar{u}\;\HOLSymConst{\HOLTokenIn{}}\;\HOLConst{fixes}\;\HOLConst{flip}\;\HOLFreeVar{S}}.

Since a prime is not a square, we have \HOLinline{\HOLConst{\HOLConst{finite}}\;\HOLFreeVar{S}}.
Now \HOLinline{\HOLFreeVar{\varphi}\;\HOLConst{\HOLConst{permutes}}\;\HOLFreeVar{S}} as Zagier map and flip map are both involutions,
by Theorem~\ref{thm:zagier-involute-mills-prime} and Theorem~\ref{thm:flip-involute-mills},
and \HOLinline{\HOLConst{fixes}\;\HOLConst{zagier}\;\HOLFreeVar{S}\;\HOLSymConst{=}\;\HOLTokenLeftbrace{}\HOLFreeVar{u}\HOLTokenRightbrace{}} by Theorem~\ref{thm:zagier-fixes-prime}.
Thus \HOLinline{\HOLFreeVar{u}\;\HOLSymConst{\HOLTokenIn{}}\;\HOLConst{fixes}\;\HOLConst{zagier}\;\HOLFreeVar{S}},
and period $p$ is odd by Corollary~\ref{cor:involute-fix-singleton-odd}.
So \HOLinline{\HOLConst{WHILE}\;(\HOLTokenLambda{}\HOLBoundVar{t}.\;\HOLConst{flip}\;\HOLBoundVar{t}\;\HOLSymConst{\HOLTokenNotEqual{}}\;\HOLBoundVar{t})\;\HOLFreeVar{\varphi}\;\HOLFreeVar{u}\;\HOLSymConst{\HOLTokenIn{}}\;\HOLConst{fixes}\;\HOLConst{flip}\;\HOLFreeVar{S}} by Theorem~\ref{thm:involute-involute-fixes-while}.
\end{proof}
\noindent
It is almost trivial to convert \HOLinline{\HOLConst{two_sq}\;\HOLFreeVar{n}} to following algorithm:
\begin{definition}
\label{def:two-squares-def}
Compute the two squares for Fermat's two squares theorem.
\begin{HOLmath}
\;\;\HOLConst{two_squares}\;\HOLFreeVar{n}\;\HOLTokenDefEquality{}\;(\HOLKeyword{let}\;(\HOLBoundVar{x}\HOLSymConst{,}\HOLBoundVar{y}\HOLSymConst{,}\HOLBoundVar{z})\;=\;\HOLConst{two_sq}\;\HOLFreeVar{n}\;\HOLKeyword{in}\;(\HOLBoundVar{x}\HOLSymConst{,}\HOLBoundVar{y}\;\HOLSymConst{\ensuremath{+}}\;\HOLBoundVar{z}))
\end{HOLmath}
\end{definition}
\noindent
giving the two squares in a pair,
and its correctness is readily demonstrated:
\begin{theorem}
\label{thm:two-squares-thm}
\script{twoSquares}{1041}
The algorithm by Definition~\ref{def:two-squares-def} gives indeed Fermat's two squares.
\begin{HOLmath}
\HOLTokenTurnstile{}\HOLConst{prime}\;\HOLFreeVar{n}\;\HOLSymConst{\HOLTokenConj{}}\;\HOLFreeVar{n}\;\ensuremath{\equiv}\;\HOLNumLit{1}\;\ensuremath{(}\ensuremath{\bmod}\;\HOLNumLit{4}\ensuremath{)}\;\HOLSymConst{\HOLTokenImp{}}\\
\;\;\;\;\;\;\;(\HOLKeyword{let}\;(\HOLBoundVar{u}\HOLSymConst{,}\HOLBoundVar{v})\;=\;\HOLConst{two_squares}\;\HOLFreeVar{n}\;\HOLKeyword{in}\;\HOLFreeVar{n}\;\HOLSymConst{=}\;\HOLBoundVar{u}\HOLSymConst{\ensuremath{\sp{2}}}\;\HOLSymConst{\ensuremath{+}}\;\HOLBoundVar{v}\HOLSymConst{\ensuremath{\sp{2}}})
\end{HOLmath}
\end{theorem}

\begin{table*}[h] 
\caption{Running Fermat's two squares algorithm in a HOL4 session, with timing.}
\Description{This table shows a sample run of Fermat's two squares algorithm in a HOL4 session, with timing information.}
\label{tbl:sample-run}  
\begin{tabular}{p{0.8\textwidth}}
\begin{verbatim}
> time EVAL ``two_squares 97``;
runtime: 0.00770s, gctime: 0.00086s, systime: 0.00077s.
val it = |- two_squares 97 = (9,4): thm
> time EVAL ``two_squares 1999999913``;
runtime: 2m23s, gctime: 14.7s, systime: 11.3s.
val it = |- two_squares 1999999913 = (1093,44708): thm
> time EVAL ``two_squares 12345678949``;
runtime: 6m02s,    gctime: 37.5s,     systime: 26.0s.
val it = |- two_squares 12345678949 = (110415,12418): thm
> EVAL ``9 * 9 + 4 * 4``;
val it = |- 9 * 9 + 4 * 4 = 97: thm
> EVAL ``1093 * 1093 + 44708 * 44708``;
val it = |- 1093 * 1093 + 44708 * 44708 = 1999999913: thm
> EVAL ``110415 * 110415 + 12418 * 12418``;
val it = |- 110415 * 110415 + 12418 * 12418 = 12345678949: thm
\end{verbatim}
\end{tabular}
\end{table*}

Table~\ref{tbl:sample-run} shows a sample run in HOL4 session on a typical laptop,
using \HOLConst{EVAL} for evaluation and prefix \HOLConst{time} to obtain timing statistics.
Note that these \HOLConst{EVAL} executions are based on optimised symbolic rewriting in HOL4, thus orders of magnitude slower than running native code.

\paragraph*{Other algorithms}
A prime has a finite set of windmill triples, by Theorem~\ref{thm:mills-finite}.
Fermat's two squares for a prime $n$ with \HOLinline{\HOLFreeVar{n}\;\ensuremath{\equiv}\;\HOLNumLit{1}\;\ensuremath{(}\ensuremath{\bmod}\;\HOLNumLit{4}\ensuremath{)}}, which must exist by Theorem~\ref{thm:fermat-two-squares-exists},
can be found by a brute-force search: subtract $n$ by successive odd squares, and check whether the difference is a square. Although there are better ways to test a square than the  square-root test, they are not simple to implement.

Don Zagier, after his one-sentence proof, referred to an effective algorithm by Wagon~\cite{Wagon-1990-acm} to compute the two squares. The algorithm requires finding a quadratic non-residue of the given prime $n$.

The advantage of our algorithm in Definition~\ref{def:two-squares-def} over such alternative methods is that only addition and subtraction are performed.
The implementation is rather straightforward. The issue of termination is discussed next.

\subsection{Terminating Condition}
\label{sec:termination}
As mentioned in Section~\ref{sec:flip-fix-search}, for our algorithm the WHILE loop may or may not terminate. To gaurantee termination, convert the WHILE loop to a countdown loop, as follows. First, ensure that the input number $n$ is not a square, so that \HOLinline{\HOLConst{mills}\;\HOLFreeVar{n}} is finite (Theorem~\ref{thm:mills-finite}), and check \HOLinline{\HOLFreeVar{n}\;\ensuremath{\equiv}\;\HOLNumLit{1}\;\ensuremath{(}\ensuremath{\bmod}\;\HOLNumLit{4}\ensuremath{)}}, so that \HOLinline{(\HOLNumLit{1}\HOLSymConst{,}\HOLNumLit{1}\HOLSymConst{,}\HOLFreeVar{n}\;\HOLConst{\HOLConst{div}}\;\HOLNumLit{4})\;\HOLSymConst{\HOLTokenIn{}}\;\HOLConst{mills}\;\HOLFreeVar{n}}, \ie, \HOLinline{\HOLConst{mills}\;\HOLFreeVar{n}\;\HOLSymConst{\HOLTokenNotEqual{}}\;\HOLSymConst{\HOLTokenEmpty{}}}.

Obviously for any triple \HOLinline{(\HOLFreeVar{x}\HOLSymConst{,}\HOLFreeVar{y}\HOLSymConst{,}\HOLFreeVar{z})\;\HOLSymConst{\HOLTokenIn{}}\;\HOLConst{mills}\;\HOLFreeVar{n}}, each $x, y$ or $y$ is less than $n$, hence \HOLinline{\ensuremath{|}\HOLConst{mills}\;\HOLFreeVar{n}\ensuremath{|}\;\HOLSymConst{\HOLTokenLt{}}\;\HOLFreeVar{n}\HOLSymConst{\ensuremath{\sp{3}}}}.
Now, use a countdown loop from \HOLinline{\HOLFreeVar{n}\HOLSymConst{\ensuremath{\sp{3}}}} to $0$, start with the triple \HOLinline{(\HOLNumLit{1}\HOLSymConst{,}\HOLNumLit{1}\HOLSymConst{,}\HOLFreeVar{n}\;\HOLConst{\HOLConst{div}}\;\HOLNumLit{4})} for the \HOLinline{\HOLConst{zagier}\;\HOLSymConst{\HOLTokenCompose}\;\HOLConst{flip}} iteration. 

The iterations trace an orbit. At half-way point, the orbit hits either a flip fixed point, detected by \HOLinline{\HOLFreeVar{y}\;\HOLSymConst{=}\;\HOLFreeVar{z}}, when the period is odd (Theorem~\ref{thm:involute-two-fixes-odd}), or another Zagier fixed point, detected by \HOLinline{\HOLFreeVar{x}\;\HOLSymConst{=}\;\HOLFreeVar{y}}, when the period is even (Theorem~\ref{thm:involute-two-fixes-even}). They provide actual exits from the countdown loop, much earlier than the count drops to zero.

\subsection{Lessons Learnt}
\label{sec:lessons}
This formalisation work can be a self-contained project in a theorem-proving workshop.
The ideas are simple, but formulating the theorems properly is not simple.
For example, at first the author would like to prove:
\begin{equation*}
\label{eqn:zagier-inv}
\HOLinline{\HOLSymConst{\HOLTokenForall{}}\HOLBoundVar{x}\;\HOLBoundVar{y}\;\HOLBoundVar{z}.\;\HOLConst{zagier}\;(\HOLConst{zagier}\;(\HOLBoundVar{x}\HOLSymConst{,}\HOLBoundVar{y}\HOLSymConst{,}\HOLBoundVar{z}))\;\HOLSymConst{=}\;(\HOLBoundVar{x}\HOLSymConst{,}\HOLBoundVar{y}\HOLSymConst{,}\HOLBoundVar{z})}.
\end{equation*}
The interactive session produces several subgoals which he cannot resolve immediately.
A comparison of Definition~\ref{def:zagier-def} with Equation~\eqref{eqn:zagier-map}
shows differences in boundary cases.
Finally, some insight from windmills resolves why the boundaries are ignored, and provides the pre-condition \HOLinline{\HOLFreeVar{x}\;\HOLSymConst{\HOLTokenNotEqual{}}\;\HOLNumLit{0}\;\HOLSymConst{\HOLTokenConj{}}\;\HOLFreeVar{z}\;\HOLSymConst{\HOLTokenNotEqual{}}\;\HOLNumLit{0}}, see Equation~\eqref{eqn:zagier-involute}. The result is Theorem~\ref{thm:zagier-involute-mills-prime}.

Explaining the Zagier map is an involution through the mind of a windmill poses some challenges. Definition~\ref{def:zagier-def} of the Zagier map has $3$ branches, so the initial effort is to treat just $3$ cases. The first case is immediate, but the second case runs into a mess. It is only after drawing a lot of windmills that the author realises these finer points:
\begin{itemize}[leftmargin=*] 
\item there are $3$ types: $x < y, x \ee y$, and $x > y$ for the windmill triple $(x,y,z)$,
\item the $3$ types are further subdivided due to geometry of the mind, giving $5$ cases in total,
\item the $5$ cases can be condensed into $3$ branches, as the definition shows.
\end{itemize}
The result is Table~\ref{tbl:zagier-map} in Section~\ref{sec:windmill-mind}.

For the permutation orbits in Section~\ref{sec:orbits},
the proofs about relations between iterates start as long-winded arguments treating if-part and only-if part separately.
Putting them in this paper prompts the author to rethink the logic.
The polished proofs simply employ a chain of logical equivalences.

Fixed point orbits have either even or odd period, as treated in Section~\ref{sec:fixed-point-period-even} and Section~\ref{sec:fixed-point-period-odd}. 
The drawing of the diagrams helps to refine the proofs to be short and sweet, making good use of theorems already proved.

About the correctness proof of the algorithm using a while-loop in Section~\ref{sec:correctness}, the author initially applied Hoare logic assertions to derive the desired iterate upon loop exit.
\ifdefined\READY
This is awkward, as pointed out by Michael Norrish who knows the HOL4 theorem-prover inside out.
\else
This is awkward, as pointed out by (omitted for anonymous review).
\fi
The reason is that \HOLinline{\HOLConst{WHILE}} is \emph{defined} as iteration of the body in HOL4.
The section had since been rewritten.

\paragraph*{Development Effort}
The proofs have been streamlined after several revisions.
Such refinements result in the script line counts for various theories developed, shown in Table~\ref{tbl:hol4-line-counts}.

\begin{table}[h] 
\caption{Statistics of various theories in this work.}
\Description{This table gives the statistics of various theories in this formalisation work.}
\label{tbl:hol4-line-counts}  
\[
\begin{array}{llr}
\text{HOL4 Theory}    & \text{Description}                        & \text{\#Lines}\\
\hline
\text{involute}       & \text{basic involution}                   & 231\\
\text{iteration}      & \text{function iteration and period}      & 917\\
\text{iterateCompose} & \text{iteration of involute composition}  & 1648\\
\text{iterateCompute} & \text{iteration period computation}       & 939\\
\text{windmill}       & \text{windmills and their involutions}    & 1844\\
\text{twoSquares}     & \text{two-squares by windmills}           & 1317\\
\end{array}
\]
\end{table}

\noindent
The scripts are fully documented, including the traditional proofs as comment before each theorem. Although comments almost double the script size, the line counts are still indicative of the effort to convert ideas into formal proofs.

\subsection{Related Work}
\label{sec:related-work}
As noted in Section~\ref{sec:introduction}, Fermat's two squares theorem has been formalised. However, none of these formal proofs is constructive, in the sense that there is no formal proof of an algorithm to compute the two squares for a prime satisfying the theorem.

Fermat's two squares theorem has two parts: existence and uniqueness.
All formal proofs include the existence part (see Theorem~\ref{thm:fermat-two-squares-exists}) , using classic and modern existence proofs: the method of infinite descent is used in one system, Gaussian integers are employed in three systems, both Heath-Brown's proof and Zagier's proof are treated in two systems.

Only two formal proofs include the uniqueness part (see Theorem~\ref{thm:fermat-two-squares-unique}): Th{\'e}ry~\cite{Thery-2004} proved by algebraic identities and divisibility, and Hughes~\cite{Hughes-2019} proved by unique factorisation of Gaussian integers.

Recently, Dubach and Muehlboeck~\cite{Dubach-Muehlboeck-2021-acm} formalised Zagier's proof using involutions in Coq's Mathematical Components Library. They illustrated their proof using the windmills as per this paper, and extended the use of involutions on the same set to formalise also an integer-partition proof of Fermat's two squares theorem by Christopher~\cite{Christopher-2016-acm}.

A summary of these formal proofs, in chronological order, is given in Table~\ref{tbl:chronology-formalise-two-squares}.

\begin{table*}[h] 
\caption{Chronology of formalisation of Fermat's two squares theorem.}
\Description{This table lists, in chronological order, the formal proofs of Fermat's two squares theorem, by various authors in different theorem provers.}
\label{tbl:chronology-formalise-two-squares}  
\begin{tabular}{llll}
Year & Author(s)[reference] & Theorem Prover & Comment\\
\hline
2004 & Laurent Th{\'e}ry~\cite{Thery-2004} & Coq & Gaussian integers, with uniqueness\\
2007 & Roelof Oosterhuis~\cite{Oosterhuis-2007} & Isabelle & Euler's proof with infinite descent\\
2009 & Marco Riccardi~\cite{Riccardi-2009} & Mizar & Heath-Brown's proof with involutions\\
2010 & John Harrison~\cite{Harrison-2010} & HOL Light & Zagier's proof with involutions\\
2012 & Anthony Narkawicz~\cite{Narkawicz-2012-acm} & NASA PVS & Zagier's proof with involutions\\
2015 & Mario Carneiro~\cite{Carneiro-2015} & MetaMath & Gaussian integers\\
2016 & Rob Arthan~\cite{Arthan-2016} & ProofPower & Heath-Brown's proof with involutions\\
2019 & Chris Hughes~\cite{Hughes-2019} & Lean & Principal Ideal Ring of Gaussian integers, with uniqueness\\
2021 & Dubach and Muehlboeck~\cite{Dubach-Muehlboeck-2021-acm} & Coq & Zagier's and Christopher's proofs with involutions\\
\end{tabular}
\end{table*}

\section{Conclusion}
\label{sec:conclusion}
About Fermat's two squares theorem,
G. H. Hardy wrote in his 1940 essay \emph{A Mathematician's Apology}~\cite[Section 13]{Hardy-1940-acm}:

\bigskip
\begin{mquote}[1em] 
\textit{This is Fermat's theorem, which is ranked, very justly, as one of the finest of arithmetic. Unfortunately, there is no proof within the comprehension of anybody but a fairly expert mathematician.}
\end{mquote}

\bigskip
\noindent
This work has been a rewarding exercise in formalisation,
delivering a proof of Fermat's Theorem~\ref{thm:fermat-two-squares-thm} using only natural numbers, involutions, and counting.
There is a certain sense of mathematical beauty when a non-trivial result can be shown by elementary means, borrowing elegant ideas by Zagier and Spivak.
Moreover, by developing a theory of involution iteration, an algorithm to compute the two squares of the theorem can be formally shown to be correct.

\paragraph*{Future Work}
The theory in Section~\ref{sec:orbits}, about orbits and fixed points, can be developed using group actions, since the iteration indices form an addition cyclic group under \HOLConst{mod} $p$, where $p$ is the orbit period. 
One can exploit the symmetry in permutation orbits, especially for permutations arising from two involutions, to improve the algorithm, as shown in the analysis by Shiu~\cite{Shiu-1996}. 
In HOL4, this direction can start from the algebra of group theory in Chan and Norrish~\cite{Chan-Norrish-cpp-2012}.
A formal analysis of the performance of the algorithm for two squares described in Definition~\ref{def:two-sq-def} can be modelled using an approach in Chan~\cite{Chan-ANU-2019-acm}.

\ifdefined\READY
\section*{Acknowledgements}
\label{sec:acknowledgements}
\addcontentsline{toc}{section}{Acknowledgments} 
Many thanks to Michael Norrish for his careful review of the draft, providing useful advice and helpful recommendations to improve this paper.
The author is also grateful to the anonymous reviewers who pointed out typographical errors and suggested clarifications.
This paper has been revised to incorporate their comments.
\else
\fi

\ifdefined\READY
\else
\section*{Appendices}
\appendix

\section{Cross-reference with Proof Script}
\label{app:cross-reference-theorems}
This is supplementary material for review.
Refer to Table~\ref{tbl:cross-reference-of-theorems} for a
mapping of theorems in this paper with names in scripts.

\begin{table*}[h] 
\caption{Cross-reference of theorems and definitions. HOL4 proof scripts uploaded as \q{script.zip.}}
\Description{This is a table giving the cross-reference of therems and definitions in the scripts. HOL4 proof scripts have been uploaded as script.zip file.}
\label{tbl:cross-reference-of-theorems}  
\begin{tabular}{lllr}
\centering
\textbf{Location of $\dots$} & \textbf{Script: (name)Script.sml} &
\textbf{Name of theorem or definition} & \textbf{Line \#}\\
\hline
Theorem~\ref{thm:fermat-two-squares-thm}
       & \q{twoSquares} & \q{fermat_two_squares_iff} & 619\\
Definition~\ref{def:windmill-def}
       & \q{windmill} & \q{windmill_def} & 348\\
Definition~\ref{def:mills-def}
       & \q{windmill} & \q{mills_def} & 481\\
Theorem~\ref{thm:mills-finite}
       & \q{windmill} & \q{mills_finite_non_square} & 714\\
Theorem~\ref{thm:mills-trivial-prime}
       & \q{windmill} & \q{windmill_trivial_prime} & 428\\
Definition~\ref{def:involute-pairs-fixes-def}
       & \q{involuteFix} & \q{pairs_def} & 269\\
Definition~\ref{def:involute-pairs-fixes-def}
       & \q{involuteFix} & \q{fixes_def} & 264\\
Theorem~\ref{thm:involute-two-fixes-both-odd}
       & \q{involuteFix} & \q{involute_two_fixes_both_odd} & 1182\\
Definition~\ref{def:flip-def}
       & \q{windmill} & \q{flip_def} & 884\\
Theorem~\ref{thm:flip-involute-mills}
       & \q{windmill} & \q{flip_involute_mills} & 933\\
Definition~\ref{def:zagier-def}
       & \q{windmill} & \q{zagier_def} & 959\\
Theorem~\ref{thm:zagier-involute-mills-prime}
       & \q{windmill} & \q{zagier_involute_mills_prime} & 1475\\
Theorem~\ref{thm:zagier-fixes-prime}
       & \q{twoSquares} & \q{zagier_fixes_prime} & 162\\
Theorem~\ref{thm:fermat-two-squares-exists}
       & \q{twoSquares} & \q{fermat_two_squares_exists_odd_even} & 441\\
Theorem~\ref{thm:fermat-two-squares-unique}
       & \q{twoSquares} & \q{fermat_two_squares_unique_thm} & 205\\
Theorem~\ref{thm:mod-4-not-squares}
       & \q{helperTwosq} & \q{mod_4_not_squares} & 419\\
Definition~\ref{def:two-sq-def}
       & \q{twoSquares} & \q{two_sq_def} & 816\\
Definition~\ref{def:period-def}
       & \q{iteration} & \q{iterate_period_def} & 348\\
Theorem~\ref{thm:iterate-period-mod}
       & \q{iteration} & \q{iterate_period_mod} & 653\\
Theorem~\ref{thm:involute-period-1}
       & \q{iterateCompose} & \q{involute_involute_fixes_both} & 558\\
Theorem~\ref{thm:involute-mod-period}
       & \q{iterateCompose} & \q{iterate_involute_mod_period} & 401\\
Theorem~\ref{thm:involute-two-fix-orbit-1}
       & \q{iterateCompose} & \q{involute_involute_fix_orbit_1} & 583\\
Theorem~\ref{thm:involute-two-fix-orbit-2}
       & \q{iterateCompose} & \q{involute_involute_fix_orbit_2} & 689\\
Theorem~\ref{thm:involute-two-fixes-even}
       & \q{iterateCompose} & \q{involute_involute_fix_orbit_fix_even_distinct} & 884\\
Corollary~\ref{cor:involute-fix-singleton-odd}
       & \q{iterateCompose} & \q{involute_involute_fix_sing_period_odd} & 1009\\
Theorem~\ref{thm:involute-two-fixes-odd}
       & \q{iterateCompose} & \q{involute_involute_fix_orbit_fix_odd_distinct_iff} & 980\\
Theorem~\ref{thm:involute-two-fixes-even-odd}
       & \q{iterateCompose} & \q{involute_involute_fix_even_period_fix} & 1100\\
Theorem~\ref{thm:involute-two-fixes-even-odd}
       & \q{iterateCompose} & \q{involute_involute_fix_odd_period_fix} & 1146\\
Theorem~\ref{thm:iterate-while-thm}
       & \q{iterateCompute} & \q{iterate_while_thm} & 922\\
Theorem~\ref{thm:involute-involute-fixes-while}
       & \q{iterateCompose} & \q{involute_involute_fix_odd_period_fix_while} & 1536\\
Theorem~\ref{thm:two-sq-thm}
       & \q{twoSquares} & \q{two_sq_thm} & 840\\
Definition~\ref{def:two-squares-def}
       & \q{twoSquares} & \q{two_squares_def} & 1025\\
Theorem~\ref{thm:two-squares-thm}
       & \q{twoSquares} & \q{two_squares_thm} & 1041\\
\hline
\end{tabular}
\end{table*}
\fi

\bibliographystyle{ACM-Reference-Format}
\bibliography{all}


\end{document}